\documentclass[12pt]{article}

\usepackage[leqno]{amsmath}
\usepackage{amsthm}
\usepackage{graphicx}
\usepackage{setspace}
\usepackage[round]{natbib}
\usepackage{amssymb}
\usepackage{multirow}
\usepackage{rotating}
\usepackage{amsfonts}
\usepackage{mathtools}
\usepackage{pdflscape}
\usepackage{longtable}
\usepackage{tikz}
\usepackage[capposition=top]{floatrow}
\usepackage{geometry}
\usepackage{booktabs}
\usepackage{tabularx}

\usepackage{bm}
\usepackage{comment}
\usepackage{soul}
\usepackage{xr}
\usepackage{enumitem}

\usepackage[T1]{fontenc}
\usepackage[latin1]{inputenc}
\usepackage{graphics}
\usepackage[english]{babel}
\usepackage{ae,amsmath,lscape,mathabx,bm}
\usepackage{hyperref,url}
\usepackage[normalem]{ulem}

\usepackage[capposition=top]{floatrow}
\everymath{\displaystyle}

\usepackage{todonotes}
\usepackage{caption}
\usepackage{subcaption}

\usepackage{relsize}
\usepackage{exscale}

\usepackage{appendix}
\usepackage{xcolor}

\newcommand{\commentbs}[1]{}
\newcommand{\abs}[1]{\lvert#1\rvert}

\def\bkR{{\rm I\kern-.17em R}}
\def\bkQ{{\rm I\kern-.17em Q}}
\def\bkN{{\rm I\kern-.17em N}}
\def\bkZ{{\rm I\kern-.17em Z}}

\newenvironment{stdarray}{\[ \left\{ \begin{array}{lcl}}{\end{array} \right. \]}
\newcommand{\bstd}{\begin{stdarray}}
\newcommand{\estd}{\end{stdarray}}
\newcommand{\beq}{\begin{equation}}
\newcommand{\eeq}{\end{equation}}

\newtheorem{assumption}{Assumption}
\newtheorem*{assumption*}{Assumption}

\newtheorem{lemma}{Lemma}

\newtheorem{proposition}{Proposition}
\theoremstyle{definition}
\newtheorem{definition}{Definition}
\newtheorem{example}{Example}
\newtheoremstyle{break}% name
  {9pt}%      Space above, empty = `usual value'
  {9pt}%      Space below
  {\itshape}% Body font
  {}%         Indent amount (empty = no indent, \parindent = para indent)
  {\bfseries}% Thm head font
  {.}%        Punctuation after thm head
  {\newline}% Space after thm head: \newline = linebreak
  {}%         Thm head spec

\theoremstyle{break}

\theoremstyle{break}

\theoremstyle{remark}

\numberwithin{equation}{section}
\usepackage{cleveref}
\crefname{figure}{figure}{figures}
\creflabelformat{figure}{#2#1#3}

\crefname{equation}{equation}{equations}
\crefrangeformat{equation}{(#3#1#4) to~(#5#2#6)}
\creflabelformat{equation}{(#2#1#3)}

\crefname{lemma}{lemma}{lemmas}
\creflabelformat{lemma}{#2#1#3}
\crefname{design}{design}{designs}
\creflabelformat{design}{#2#1#3}
\crefname{proposition}{proposition}{propositions}
\creflabelformat{proposition}{#2#1#3}
\crefname{condition}{condition}{conditions}
\creflabelformat{condition}{#2#1#3}
\crefname{assumption}{assumption}{assumptions}
\creflabelformat{assumption}{#2#1#3}
\crefname{remark}{remark}{remarks}
\creflabelformat{remark}{#2#1#3}
\crefname{appendix}{appendix}{appendices}
\creflabelformat{appendix}{#2#1#3}

\def\citeapos#1{\citeauthor{#1}'s (\citeyear{#1})}
\def\uniset{{\rm 1\kern-.40em 1}}
\newcommand{\calZ}{\mathcal{Z}}

\newcommand{\calR}{\mathcal{R}}

\newcommand{\calT}{\mathcal{T}}
\newcommand{\acalT}{\vert \mathcal{T} \vert}
\newcommand{\calTs}{\calT^\ast}
\newcommand{\calZs}{\calZ^\ast}

\newcommand{\calZnotZs}{\calZ\, \setminus\, \calZs}

\newcommand{\calTnotTs}{\calT\, \setminus\, \calTs}

\newcommand{\calZnotA}{\calZ\, \setminus \, A}

\newcommand{\calZnotAi}{\calZ\, \setminus \, A_i}

\newcommand{\bcalB}{{\boldsymbol{\mathcal{B}}}}
\newcommand{\bcalK}{{\boldsymbol{\mathcal{K}}}}

\global\long\def\cup{\operatorname*{\mathsmaller{\bigcup}}}

\begin{document}

%\doublespacing

\onehalfspacing

\title{\Large\textbf{Treatment Effects with Targeting Instruments}\footnote{This  is a shorter, revised  version of our ``Filtered and Unfiltered Treatment Effects with Targeting Instruments'' (\href{https://arxiv.org/pdf/2007.10432v1.pdf}{first arXiv version July 2020}).}}

\author{\textbf{Sokbae Lee}\footnote{Department of Economics, Columbia University and Centre for Microdata Methods and Practice, Institute for Fiscal Studies, sl3841@columbia.edu.}
\and
\textbf{Bernard Salani\'e}\footnote{Department of Economics, Columbia University and FGV EPGE, Rio de Janeiro, bsalanie@columbia.edu.}
}

\date{May 4, 2026}

\maketitle

\thispagestyle{empty}

\begin{abstract}
Multivalued treatments are commonplace in applications. We explore the use of discrete-valued instruments to control for selection bias in this setting. Our discussion revolves around the concept of targeting:  which instruments target which treatments. It allows us to establish conditions under which counterfactual averages and treatment effects are point- or partially-identified for composite complier groups. We explore the additional identifying power of a positive selection assumption.  We illustrate its usefulness by revisiting the findings of \citet{kline2016} on the  Head Start Impact Study. We derive informative bounds that   suggest less beneficial effects of Head Start expansions
than their parametric estimates.
% Under a plausible positive selection assumption, we derive informative bounds that are not compatible with the estimates of \citet{kline2016}.}

\bigskip


\end{abstract}

\clearpage

\setcounter{page}{1}

\section*{Introduction}

   Much of the literature on the evaluation of treatment effects 
has concentrated on the paradigmatic ``binary/binary'' example, in which both   treatment and instrument only take two values. 
Multivalued treatments are 
common in actual policy implementations, however, as are multivalued instruments. Many different programs aim to help train job seekers for instance, and each of them has its own eligibility rules. Tax and benefit regimes distinguish many categories of taxpayers and eligible recipients. The choice of a college and major has  many dimensions too, and responds to a variety of financial help programs and other incentives.    Finally, more and more randomized experiments in economics  resort  to  factorial designs\footnote{\citet{Factorial:Designs:REStat} review recent applications of factorial designs.}.

   Existing work on  multivalued treatments under selection on observables includes
  \citet{imbens2000role},
  \citet{cattaneo2010efficient},
  and
  \citet{Ao:et:al}
  among others.
  As the training, education choice, and tax-benefit examples illustrate,   in non-experimental settings multivalued treatments   are also subject to selection on unobservables.  The use of instruments to evaluate the effects of multivalued treatments under selection on unobservables
 has received increasing attention in  recent literature.
   In previous work \citep*{LeeSalanie2018}, we analyzed the case when  enough continuous instruments are available. Identification is of course more difficult when  instruments only take discrete values.
   We explore in this paper the use of such discrete-valued instruments in order to control for selection bias when  evaluating  discrete-valued treatments.  Our goal is to find plausible conditions on treatment assignment and on the distribution of outcomes under which counterfactual averages and treatment effects are  point- or partially identified for various (sometimes composite) complier groups. This distinguishes our paper from the recent contributions of  \cite{baietal:monotonicityaverage}, which focuses on population-wide average outcomes, and of \cite{goff:oa}, which studies  identification without any assumption on outcomes.

  In the binary/binary model, the analyst can often take for granted that switching on the binary instrument makes treatment (weakly) more likely  for all or no  observations. This is satisfied under the local average treatment effect (LATE)-monotonicity assumption \citep[e.g.,][]{LATE1994,vytlacil2002independence,HV2007-handbook}.
  With multiple instrument values and multiple treatments,  there may be no natural ordering of instrument or treatment values that would give meaning to the word ``monotonicity''.   \citet{heckmanpinto-pdt} defined an   ``unordered monotonicity''  property; various papers have proposed other definitions of (qualified) monotonicity\footnote{See \citet{nafjeevanpinto:2022} for a detailed  analysis of some of these proposals.}. 
  Even when such a condition holds, there exist several  groups  of compliers---individuals whose treatment assignment changes with the  value of the instrument. 
 Since this may give rise to a multiplicity of cases, existing literature has often added assumptions that reduce this complexity.  \citet{angrist1995two} analyzed two-stage least squares (TSLS) estimation when  the  treatment takes a finite number of ordered values. 
  Several recent papers  have studied the case of binary treatments with multiple instruments.
  \citet{MTW:2020}
 and 
 \citet{Goff}
 analyzed the identifying power of different monotonicity assumptions in this context\footnote{\cite{MTW:2024} further apply their framework of monotonicity with multiple instruments to marginal treatment effects \citep[e.g.,][]{heckman2001policy,MIV2005,CHV-AER}.}. Others have studied models with binary instruments and multivalued or continuous treatments.
 \cite{Torgovitsky2015}, \cite{DF2015}, 
 \cite{Huang:et:al:2009},
 \citet{caetano_escanciano_2020},   \citet{Feng:2024}, and \citet{AngristSantosTecchio2025} developed identification results for different models.

 In a wide-ranging contribution, \cite{heckmanpinto-pdt} derived  results on  partial identification in discrete-instrument, discrete-treatment models; they also  showed how additional identifying assumptions, such as unordered monotonicity, can be applied to shrink the identified set of treatment effects for various complier groups. While their results are very general, they are not as transparent as one would like. Our 
approach to this issue is different: we seek a parsimonious framework within which we can make constructive progress, and that can still be useful in many applications.
 In order to reduce the 
 complexity of the problem,  we start by imposing an additive  random-utility
  model  (ARUM) structure,  as did \citet{HUV2006,HUV2008} and \citet{HV2007-handbook-2}.  Under ARUM, the selection into treatment depends on mean values and additive,  observation-specific shocks. Some, but not all, ARUM models satisfy the  unordered monotonicity property of \citet{heckmanpinto-pdt}, which was applied by \citet{pinto2021}  to the Moving to Opportunity program.
  
 In many applications, a  value of the treatment is especially salient; since it often is the no-treatment value $t=0$, we call it the ``control''. Under ARUM, each treatment $t$ generates a change in the mean value, relative to the control, that depends on the value $z$ of the instrument.  
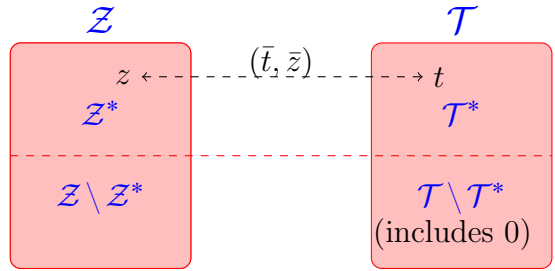
It is  natural  to speak of an instrument value $z$ {\em targeting\/} a treatment value $t$  when  it maximizes this change in  mean value. Most of our paper relies on the  assumption of {\em strict targeting\/},  which obtains when each instrument only changes the mean values of the treatments it  targets. Strict targeting holds for instance in models of imperfect compliance when the cost of non-compliance does not depend on its nature.  Some of our results also require   {\em one-to-one targeting}, where each non-zero instrument targets one treatment only, and each treatment  (apart from  the control) is targeted by one instrument only.   Finally, we speak of {\em universal targeting\/} when both strict and one-to-one targeting apply.

As we will see, each of these targeting assumptions generates testable implications. These tests can also be used to match instrument values and the treatment values that they target. They only rely on estimates of the generalized propensity scores, which are directly identifiable from the data. In particular, the testable implications of universal targeting that we derive are identical to those of \cite{Bai:Tabord-Meehan:2025:arXiv} and are therefore sharp.

Our use of ``targeting'' instruments is similar in spirit  to Section~7.3 of~\cite{HV2007-handbook-2}\footnote{See also the  recent contribution by \citet{Buchinsky:Gertler:Pinto:2023}, which  uses  revealed preference arguments.}. We define it differently and we seek to identify a more general class of treatment effects.
The term ``targeting'' is inspired by the time-honored Targeting Principle\footnote{Early references include \cite{tinbergen:econpolicybk} and \cite{bhagwati:1971}.}.  Some  policies  act directly on final outcomes, and others  aim to modify choices. Our use of the term ``targeting'' refers to the latter. Take a Roy model in which workers choose among occupations on the basis of their net utilities; we observe the choice of occupation and the wage in that occupation. A safety regulation that reduces the disutility of labor for (say) construction workers is, in our terminology,  an instrument that targets the choice to be a construction worker. Policymakers might also seek to
 increase average incomes by offering a college credit. While their final aim is to increase  wages (an outcome), we would say that the college credit is an instrument that targets the choice to go to college---a treatment variable.

 % To illustrate, consider the effect of various programs indexed by $t$ on some outcomes of interest.  Let each  instrument value 
 %  $z$ stand for a  policy regime under which  the access to some programs is made easier or harder than in a control group.  We assume that selection into treatment satisfies  ARUM.
 %  For simplicity, we will  use the term ``subsidy'' to refer to the corresponding changes in the mean values, which may also arise from variations in eligibility conditions.  Then a policy regime $z$ {\em targets\/} a treatment $t$ if it has the highest subsidy for this program among all policy regimes. {\em Strict targeting\/} requires that all policy regimes $z^\prime$ that do not target $t$  have the same (lower) value of the subsidy for $t$. {\em One-to-one targeting} requires that each subsidy be directed to one treatment, and vice-versa.

 %\Bernardnote{The following replaces the previous paragraph on subsidies.}

  To illustrate the usefulness of our framework, we  apply it to the Head Start Impact Study (HSIS),  a randomized experiment that sought to evaluate the value added of Head Start preschools.  \citet{kline2016} revisited the HSIS; they  took into account   the presence of a substitute treatment (alternative preschools in this case).  They found that Head Start was only beneficial for children who would not have attended another preschool program instead.  In this study, the instrument is binary: a child is offered admission to Head Start or not. Treatment is ternary, as the child may end up in Head Start, in an alternative preschool, or not be enrolled in preschool. In our language, ``no preschool'' is the reference treatment. Head Start offers of course target Head Start; and since the instrument is binary, targeting is trivially strict. This makes it an example of universal targeting.

% \Bernardnote{Maybe we could drop or drastically shorten the next paragraph}
% To illustrate, consider a typical randomized experiment with imperfect compliance:  (i) individuals are randomly assigned to a treatment branch $t$ (including a non-treatment option $0$) based on the instrument value $z$ that they draw; (ii) some individuals self-select into a treatment branch $t^\prime$ that they prefer, even though they did not draw the corresponding instrument value $z^\prime$.
% In our terminology, $z$ targets $t$ and $z^\prime$ targets $t^\prime$.  
% Often this mapping is one-to-one; this is what our one-to-one targeting assumption states.  
% Strict targeting of a treatment branch $t$ is more restrictive, as its name indicates.  One way to interpret it in a non-compliance context is that it is equally difficult for an average individual in the population to select into treatment $t$ when she was not assigned to that branch, no matter what the experimenter's intended treatment branch was.

As another example, consider a randomized experiment with imperfect compliance. Individuals are assigned to a treatment arm via instrument (targeting), but may self-select into an alternative treatment  (imperfect compliance). When random assignments map uniquely to treatments, the design satisfies our one-to-one targeting assumption. Strict targeting further implies that selecting into treatment is equally costly for all individuals not assigned to it.
Consider for instance the interventions reported in \citet{Angrist:et:al:2009}  and in \citeauthor{attanasioetal:colombia:2014}   (\citeyear{attanasioetal:colombia:2014,attanasioetal:colombia:2020}).
These are 4-way factorial randomized experiments: each subject is randomly assigned to a control group, to receive treatment~1, to receive treatment~2, or to receive both treatments. By definition, this is one-to-one targeting. Compliance was very imperfect in \cite{Angrist:et:al:2009}, and it is described as ``high'' in the other two papers. If subjects self-selected into treatments on the basis of their expected benefits, then strict targeting is a natural assumption. 
%[Check Besley's paper
%\url{https://link.springer.com/chapter/10.1007/978-1-349-21587-4_5}
%]

Combining  ARUM and  assumptions on targeting allows us to point-identify the size of some complier groups  and the corresponding counterfactual averages and treatment effects on the outcomes, and to partially identify others.  
 We use two examples to demonstrate the identification power and implications of ARUM and  targeting. Our   first example is a $2 \times T$ model where a binary instrument targets only one of $T \geq 3$ treatment values, as in~\citet{kline2016}.  In our second example, 
   three unordered treatment values target three instrument values. This  $3 \times 3$ model was also studied by \citet{kirkeboen2016}\footnote{See also more recent work by \citet{Bhuller22}, \citet{Heinesen:2025}, and \citet{nos22}.}.  Unlike them, we do not  assume that the  data contains information on next-best alternatives. Whereas  the $2\times T$ model satisfies unordered monotonicity under our strongest targeting assumptions, the $3\times 3$ model does not\footnote{
  It does satisfy the weaker generalized monotonicity assumption of \cite{baietal:monotonicityaverage}, however.}.

 We obtain novel identification results for both examples; they lead to new estimands or bounds for average treatment effects on various groups.
  Additional identifying assumptions can refine these bounds. Our leading  example is what we call {\em positive selection}. This assumes that the average outcome for a given treatment $t$ is larger for some response group   than for another. Consider for instance the binary instrument case. It seems natural to assume that the always-takers of a treatment get more  from it than compliers who only take it if they are incentivized to do so.  Positive selection  also obtains under weak assumptions in the generalized Roy model. More generally, let us return to our earlier illustration  of a randomized experiment under imperfect compliance.  Consider the  response group  of individuals who would end up in treatment $t^\prime$ both when drawing $z$ and when drawing  $z^\prime$. We would expect this response group to  have better outcomes under $t^\prime$, on average, than the response group that exhibits perfect compliance to $z$ and $z^\prime$ draws---assuming that these two response groups end up in the same treatment branches for all other instrument values.  This falls exactly under our positive selection assumption. It adds  identifying power in both of our leading examples.

  %\Bernardnote{Added this from the answer to R2.}
 To show the value of our approach, we apply it to the reanalysis of Head Start by \cite{kline2016}. We   confirm the importance of taking into consideration alternative preschools when evaluating Head Start.   Unlike \citet{kline2016},  we do not rely on parametric selection models.   Under a plausible positive selection assumption, our estimates suggest that the large difference between complier groups that they find  can only be rationalized under {\em negative\/} selection into Head Start. As a by-product, we provide an upper bound on the welfare effect of expanding access to Head Start. Interestingly, the estimated upper bound turns out to be lower than the point estimate of \cite{kline2016};  and it yields a lower marginal value for  public funds used in expanding access to Head Start.

The  paper is organized as follows.
\Cref{sec:model} defines our framework.
 In \Cref{sec:targeting}, we     define  and discuss the concepts of targeting,  one-to-one targeting, and strict targeting. Section~\ref{sec:identif}  derives
their implications for the identification of population shares,
 counterfactual averages,
 and the  effects of the treatments on  various complier groups;  it also defines and illustrates   positive selection. 
Finally, we present estimation results for   Head Start  in \Cref{sec:empirical}.
The Appendices contain the proofs of all  propositions and lemmata, along with some additional material.

\section{The Framework}\label{sec:model}

In all of the paper, we denote observations as $i=1,\ldots,n$. Each observation consists of covariates $X_i$, instruments $Z_i$, outcome variables $Y_i$,  and 
treatments $T_i$. We assume that the covariates $X_i$ are exogenous to treatment assignment and outcomes. Since they will not play any role in our identification strategy, we condition on  the covariates throughout and we omit them from the notation. Our results should therefore be interpreted as conditional on $X$.

We assume that observations are independent and identically distributed. Random sampling rules out that the treatment status of one observation influences other observations. This further implies that the outcome for a specific observation does not impact the outcomes of other members within the population. In other words, we rely on the Stable Unit Treatment Value Assumption (SUTVA).

%\Bernardnote{I shortened the next paragraphs, and deleted the relevance assumption (ex Assumption 2)---targeting is relevance on steroids.}

  We focus  in this paper 
on  treatment variables that take discrete values,   which we label
$t \in \calT$.  For simplicity, we will call $T=t$ ``treatment $t$''.  
%These values   needn't be\/} ordered: e.g. $t=2$, when available,   may not represent\/} ``more treatment'' than $t=1$. 
These values do not have to be ordered; e.g., when $t=2$ is available, it does not necessarily indicate ``more treatment'' than $t=1$.
We assume that   the only available  instruments are discrete-valued, and we label their values as $z \in \calZ$.  Finally, for any $z\in\calZ$ and $t\in \calT$, we denote the {\em generalized propensity score\/} by
\[
P(t|z) \equiv \Pr(T_i=t\vert Z_i=z).
\]

\subsection{Restricting Heterogeneity}
As in most of this literature, we will 
need an assumption that restricts the heterogeneity in the  counterfactual mappings $T_i(z)$, that is, potential treatments.  In the binary/binary model, this is most often done by imposing LATE-monotonicity. 	 As is well-known, LATE-monotonicity imposes that (denoting instrument values as $z=0,1$) 
	(i) or (ii) must hold:
	\begin{enumerate}
		\item[(i)] 
	for each observation $i$,  $T_i(1)\geq T_i(0)$;
	\item[(ii)] for each observation $i$, 
	 $T_i(0)\geq T_i(1)$.
	\end{enumerate}

With more than two treatment values and/or more than two instrument values, there are many ways to restrict the heterogeneity in treatment assignment. Since treatments  may not be ordered in any meaningful
 way, we cannot apply the results in~\citet{angrist1995two} for instance.  \citet{MTW:2020}  state  several versions of monotonicity for a binary treatment model with $\abs{\calZ}>2$. 	They propose a ``partial monotonicity'' assumption which 
 applies binary LATE-monotonicity component by component. This requires that the instruments  be 
 interpretable as combinations of component instruments, which is not necessarily the case here.

 To cut through this complexity, we assume from now on that assignment to treatment can be represented by an Additive Random-Utility Model (ARUM), that is  
 by a discrete choice problem with additively separable errors:
 \[
 T_i(z) = \arg\max_{t \in \calT} (U_{z}(t) + u_{it})
 \] 
 for some real numbers $U_z(t)$ which are common across observations,  and random vectors ${(u_{it})}_{t\in \calT}$ that are distributed independently of $Z_i$. We do not restrict the codependence of the  random variables $u_{it}$. 
 %{\color{blue}The  usual  models of multinomial choice belong to this family. ARUM also includes ordered treatments, for which 
 %$u_{it}\equiv \sigma(t) u_i$  for some one-to-one positive function $\sigma$.}\Simonnote{We need to double check this part}

In a randomized experiment with perfect compliance, we would have  $U_z(t^\prime)=-\infty$ and $U_{z^\prime}(t)=-\infty$. With imperfect compliance, these mean values are finite; if for instance $u_{it^\prime}-u_{it}>U_z(t)-U_z(t^\prime)$, individual $i$ will get into treatment $t^\prime$ when drawing $z$ would normally assign her to $t$.

% \Bernardnote{Added above.}

 Imposing an ARUM structure will greatly simplify our discussion of treatment assignment. It incorporates a substantial restriction, however. Suppose that observation $i$ has treatment values $t$ under $z$ and $t^\prime$ under $z^\prime$. By the ARUM structure, this implies
 \begin{align*}
	U_z(t) + u_{it} \geq  U_z(t^\prime) + u_{it^\prime}\\
	U_{z^\prime}(t^\prime) + u_{it^\prime} \geq  U_{z^\prime}(t) + u_{it}.
 \end{align*} 
Combining these two restrictions implies an ``increasing differences'' property:
 \[
U_{z^\prime}(t^\prime)	- U_{z^\prime}(t) \geq  U_z(t^\prime)-U_z(t).
\]
     This inequality  in turn is incompatible with the existence of an observation $j$ that has treatment values $t^\prime$ under $z$ and $t$ under $z^\prime$. Thus we rule out  ``direct two-way flows'': if a change in the value of an instrument causes an observation to shift from a treatment value $t$ to a treatment value $t^\prime$, it can cause no other observation to switch from $t^\prime$ to $t$. The argument above is a special case of the general discussion in~\citet{heckmanpinto-pdt}; their Theorem~{T-3}  shows that the treatment assignment models that satisfy unordered monotonicity  for each pair of  instrument values  can be represented as an ARUM. Not all ARUM models satisfy unordered monotonicity, however; unordered monotonicity excludes a more general class of two-way flows. We will illustrate  this point on one of our leading examples in \Cref{subsec:ternary-ternary}.

\subsection{Assignment to Treatment}

Assumption~\ref{assn:ARUM} defines the class of models of assignment to treatment that we analyze in this paper.

\begin{assumption}[ARUM]\label{assn:ARUM}
	The treatment assignment model consists of:
	\begin{enumerate}
		\item a finite set $\calT =\{0,1,\ldots, \abs{\calT}-1\}$;
		\item a finite set of instrument values $\calZ=\{0,1,\ldots, \abs{\calZ}-1\}$; 
		\item an ARUM  model of treatment:
		\[
		T_i(z) = \arg\max_{t\in \calT}(U_z(t)+u_{it}),
		\] 
where the vector $(u_{it})_{t\in\calT}$ is distributed independently of $Z_i$ and is continuously
distributed with full support on $\mathbb{R}^{\acalT}$.
\end{enumerate}
\end{assumption}

We will often refer to the   $U_z(t)$ as ``mean values''. This is only meant to simplify the exposition; it is consistent with, but need not refer to,  preferences on the part of the agent.
Note that when $\calT=\{0,1\}$, Assumption~\ref{assn:ARUM} is just the standard monotonicity assumption, with a threshold-crossing rule
\[
	T_i(z)=\uniset(u_{i0}-u_{i1}\leq U_z(1)-U_z(0)).
	\] 
If we add a third treatment value so that $\calT=\{0,1,2\}$, the ARUM assumption starts to bite as it excludes direct two-way flows in the treatment model. However, 
%the combination of Assumptions~\ref{assn:valid} and~\ref{assn:ARUM} 
Assumption~\ref{assn:ARUM} 
is far from  sufficient to identify interesting treatment effects in general. In order to better understand what is needed, we now resort to the notion of {\it response-groups\/} of observations, whose members share the same mapping from instruments $z$ to  treatments $t$. We first state a general definition\footnote{This is analogous to  the definitions in \citet{heckmanpinto-pdt}.}.

\begin{definition}[Response-vectors and Response-groups]\label{def:cgroups}
Let $R$ be an element of  the  Cartesian product $\calT^{\calZ}$ and  $R(z)\in \calT$ denote its component for instrument value $z\in\calZ$. 
\begin{itemize}
	\item Observation $i$ has 
	(elemental) {\em response-vector\/}  $R$ if and only if for all $z\in \calZ$, $T_i(z)=R(z)$. The set $C_R$ denotes   the set of observations with response-vector $R$ and we call it a {\em response-group}. 
	\item We extend the definition in the natural way to incompletely specified mappings, where each $R(z)$ is a subset of $\calT$. We call the corresponding response-vectors and response-groups {\em composite}.
	 In particular, if $R(z)=\calT$ we denote it by an asterisk in the corresponding position.
\end{itemize}
\end{definition}

To illustrate, consider  the binary instrument/binary treatment case. It has 
 a priori $2^2=4$ response vectors, $R \in \{00, 01, 10, 11\}$ with corresponding response-groups
$C_{00}, C_{01}, C_{10}, C_{11}$. In this notation, the first number refers to a treatment value with $z=0$ and the second number with $z=1$. For instance, $C_{01}$ refers to those with $T_i(0) = 0$ and $T_i(1) = 1$,  while the composite response-group $C_{\ast 1}$, for which $R(0)=\{0,1\}$,  represents the union of $C_{01}$ and $C_{11}$.
The LATE-monotonicity assumption implies that either $C_{01}$ or $C_{10}$ is empty.

\section{Targeting}\label{sec:targeting}

We start by introducing additional assumptions on the underlying treatment model. 
We will illustrate these assumptions  on   three examples: the ``binary instrument model'' or the ``$2\times T$'' model;     the ``$3\times 3$ model''; and a generalized Roy model.
 We first define them briefly.

\begin{example}[The binary instrument ($2\times T$) model\label{ex:binaryinstrument0}]
$\calT=\{0,1,\ldots,T-1\}$ and $\calZ=\{0,1\}$.  This could for instance represent an intent-to-treat model, where agents in the control group $Z=0$ are not treated  ($T=0$) and agents with $Z=1$ self-select the  type of the treatment $T\geq 1$ or opt out altogether ($T=0$). $\qed$
	%Let $\calT=\{0\} \cup \calTs$ with  $\calTs=\{1,\ldots,\abs{\calT}-1\}$, and $\calZ=\{0,1\}$.
		\end{example}

		When $\abs{\calT}=3$, treatment assignment  can be 
represented in 
the $(u_{i1}-u_{i0},u_{i2}-u_{i0})$ plane. The points of coordinates  $P_z=(U_z(0)-U_z(1),U_z(0)-U_z(2))$ 
play an important role as for a given $z$, 
\begin{itemize}
	\item $T_i(z)=0$ to the south-west  of $P_z$;
	\item $T_i(z)=1$ to the right of $P_z$ and below the diagonal that goes through it;
			\item $T_i(z)=2$ above $P_z$ and above the diagonal that goes through it.
\end{itemize}
Treatment assignment is illustrated in  Figure~\ref{fig:treatT3} for a given $z$,  
where the origin is in $P_z$.
We will make recurrent use of this type of figure.

\begin{example}[$3\times 3$  model\label{ex:ternaryternary}]
Assume that $\calZ=\{0,1,2\}$ and $\calT=\{0,1,2\}$.  
As a leading example, \citet{kirkeboen2016} investigate the $3\times 3$  model in order to
analyze the effect of students' choice of field  of study  on their earnings;  each  instrument value shifts the eligibility of a student for a given field. We will return to this application  in Section~\ref{sub:klm}.
\end{example}

 Finally, our framework also includes multivalued generalized Roy models (see \cite{EisenhauerHV:2015}).
\begin{example}[A Generalized Roy Model\label{ex:galroy}]
Suppose that agents choose occupations $t=0,\ldots,\abs{\calT}-1$ on the basis  of their expected  wages $w_i(t)=\bar{w}(t)+\eta_{it}$, net of labor disutilities $d_i(z,t)=\bar{d}_z(t)+ v_{it}$ which depend on the values of the instruments:
\[
T_i(z)= \arg\max_{t=0,\ldots,\abs{\calT}-1}(w_i(t) -d_i(z,t)).
\]
Potential wages are $Y_i(t)=w_i(t)+\varepsilon_{it}$.
We observe $Z_i$, the chosen occupation $T_i=T_i(Z_i)$, and  realized wages $Y_i=Y_i(T_i)$. If the vector of variables $\{(\eta_{it}, v_{it})\}_{t=0}^{\abs{\calT}-1}$ is independent of $Z_i$, this is an ARUM  model with $U_z(t)=\bar{w}(t)-\bar{d}_z(t)$ and $u_{it}=\eta_{it}-v_{it}$.
\end{example}

\begin{figure}[htb]
\begin{tikzpicture}[scale=0.5] 
\draw [blue!20!white, fill=blue!20!white] (0,0) -- (5,5) -- (5,-5) -- (0,-5);
\draw [green!20!white, fill=green!20!white] (0,0) -- (5,5)  -- (-5,5) -- (-5,0);
\draw [red!20!white, fill=red!20!white] (0,0) -- (0,-5)  -- (-5,-5) -- (-5,0);

\draw[->,>=latex] (-5,0)  -- (5,0) node[below] {$u_{i1}-u_{i0}$};
\draw[->,>=latex] (0,-5)  -- (0,5) node[above] {$u_{i2}-u_{i0}$};
\draw[-] (0,0)  -- (5,5);
		\fill[blue] (0.0,0.0) circle (.1cm);
	\node[blue] at (0.3,-0.3) {$P_z$};
\node[blue] at (-3,-3) {$T_i(z)=0$};
%	\node[blue] at (2,1) {$D_i(z)=1$};
	\node[blue] at (3,-3) {$T_i(z)=1$};
\node[blue] at (-3,3) {$T_i(z)=2$};
%		\node[blue] at (1,3) {$D_i(z)=2$};
\end{tikzpicture} 
\caption{Treatment assignment for $\abs{T}=3$ for given $z$}\label{fig:treatT3}
\end{figure}

 ``Targeting'' will be the common thread in our analysis. Just as in general economic discussions a policy measure may target a particular outcome, we will speak of instruments (in the econometric sense) targeting the assignment to a particular treatment. 

Under  \Cref{assn:ARUM}, assignment to treatment is governed by the differences in mean values
$(U_z(t)-U_z(\tau))$  and by the differences in unobservables $u_{it}-u_{i\tau}$. Only the former depend on the instrument. 
From now on, we assume that there is a reference treatment value $t_0$ whose mean utility does not depend on the value of the instrument:

\begin{assumption}[Reference Treatment]\label{assn:ref:treat}
There exists $t_0\in\calT$ such that $z\in \calZ \to U_z(t_0)$ is constant.
Without loss of generality, we renumber treatment values so that $t_0=0$; and we normalize utilities with $U_z(0)=0$ for all $z\in\mathcal{Z}$.
\end{assumption}

In many applications, $t=0$ is a ``no-treatment'' value, and instruments only change the mean utilities of the other treatments. For instance, tuition subsidies, investment credits, and invitations to training programs have no effect for those who do  not attend college, do not invest, or choose not to train.  \Cref{assn:ref:treat} seems 
natural in such cases\footnote{In the generalized Roy model (Example~\ref{ex:galroy}), it holds if the disutility of occupation~0 does not depend on the values of the instruments.}. For a counter-example, consider a program of unconditional cash transfers with different values $z$, for which we observe the purchases $t$ of several categories of goods  the following month. If a household decides to save the transfer ($t=0$), its mean (discounted) utility will still depend  on the value of the transfer $z$ that it received\footnote{In that case one could define targeting with the function $\tilde{U}_z(t)= U_z(t)-U_z(0)$. We have not explored the consequences of this alternative definition.}. 

%\Bernardnote{I added the above as a response to R1.}

Given Assumption~\ref{assn:ref:treat}, we will say that an instrument value $z$ {\em targets\/} a treatment value  $t$ if it maximizes the mean utility  $U_z(t) - U_z(0)=U_z(t)$.

\begin{definition}[Targeted  Treatments and Targeting Instruments]
    \label{def:tst}
    For any $z\in \calZ$ and $t\in\calT$, let 
 \[
 \bar{U}(t)  \equiv \max_{z\in \calZ} U_z(t)  \; \mbox{ and } \; Z^\ast(t) \equiv \arg\max_{z\in \calZ} U_z(t)
 \]
denote the maximum value of $U_z(t)$ over $z\in \calZ$ and 
 the set of maximizers, respectively.
 If  $Z^\ast(t)$ is not all of $\calZ$, then we will say that the instrument values $z\in Z^\ast(t)$ {\it target\/} treatment value $t$; and we write $t\in T^\ast(z)$. If $z$ does not target any treatment, we define $T^\ast(z)=\emptyset$.
We denote   $\calTs$ the set of
  targeted treatments and 
  $\calZs = \cup_{t\in\calTs} \;  Z^\ast(t)$ the set of targeting instruments. 
\end{definition}

Definition~\ref{def:tst} calls for several remarks. First, Assumption~\ref{assn:ref:treat} implies that   $Z^\ast(0)=\calZ$.  Therefore $t=0$ is not in $\calTs$; the set $\calTs$ may exclude other treatment values, however.
If a treatment value $t$ is not targeted  ($t\not\in \calTs$), by definition the function $z\to U_z(t)$ is constant over $z\in\calZ$, with value $\bar{U}(t)$.
If an instrument value $z$ does not target any treatment ($z\not\in \calZs$), then  $U_z(t)<\bar{U}(t)$ for every $t\in\calTs$.
 While non-targeted treatment values ($t\in \calTnotTs$) have mean values that do not respond to changes in the instruments, these mean values may and in general will differ across treatments. The probability that an individual observation takes a treatment $t \in \calTnotTs$ also generally depends on the value of the instrument.

It is important to note here that the  values $U_z(t)$ and therefore the targeting maps $Z^\ast$ and $T^\ast$ are not observable; any assumption on targeting instruments and targeted treatments must be a priori and   context-dependent. These prior assumptions  have  consequences that can be tested, however.   As a simple example, $z=1$ targets $t=1$ in the $2\times 2$ model under monotonicity; this implies that $P(1\vert 1)\geq P(1\vert 0)$. We will show that targeting always yields linear inequalities on generalized propensity scores in the class of models defined by Assumptions~\ref{assn:ARUM} and~\ref{assn:ref:treat}.

\subsection{Top targeted treatments and  top alternative treatments}

\begin{definition}[Top targeted and   top alternative treatments]\label{def:top-targets}
	Take any observation $i$ in the population, and an instrument value $z\in\calZ$.
\begin{enumerate}[label=(\roman*)]
\item
	If $z$ is a targeting instrument (that is, $z \in \calZs$), let 
	\[
		V^\ast_i(z) = \max_{t\in \calTs(z)} (\bar{U}(t)+u_{it})
	\]
	and $t^{\ast}_i(z)$ denote the  maximizer, which is unique almost surely. 
	We call $t^\ast_i(z)$ the {\it top targeted treatment} for observation $i$ under instrument value $z$. 
    
    If $z\not\in\calZs$, we define $V^\ast_i(z)=-\infty$.  
\item Also  define
	\[
        \underline{V}_i(z) = \max_{t\in \calT\setminus\calTs(z)} 
        \left(U_z(t) + u_{it}\right)
	\]
 and  let $\underline{t}_i(z)$ denote the almost surely unique  maximizer. 
We call $\underline{t}_i(z)$ the {\it top alternative treatment} for observation $i$ under $z$.
\end{enumerate}	
\end{definition}

By definition, the treatment $T_i(z)$ assigned to $i$ under $z$  can only be the top targeted treatment  $t^\ast_i(z)$ (if $V^\ast_i(z)>\underline{V}_i(z)$) or the top alternative treatment $\underline{t}_i(z)$ (if $V^\ast_i(z)<\underline{V}_i(z)$). If $z$ is not a targeting instrument, then $T_i(z)$ must be $\underline{t}_i(z)$. 

Definition~\ref{def:top-targets} could be extended to broader settings than our ARUM model. On the other hand, the ARUM structure interacts with our definition of targeting to constrain treatment assignment in interesting ways.
Suppose that an instrument value $z$ targets a treatment value $t$: $U_z(t)=\bar{U}(t)$. If  some observation $i$ is assigned a treatment value $T_i(z)=t^\prime\neq t$, we must have
\[
\bar{U}(t)+u_{it} < U_z(t^\prime)+u_{it^\prime}.
\]
Now suppose that an instrument value $z^\prime$  targets this $t^\prime$; then 
 $U_z(t^\prime)\leq U_{z^\prime}(t^\prime)=\bar{U}(t^\prime)$, so that
\begin{equation}\label{eq:tzt}
U_{z^\prime}(t)+u_{it}\leq  \bar{U}(t)+u_{it} < U_z(t^\prime)+u_{it^\prime}\leq U_{z^\prime}(t^\prime)+u_{it^\prime}
\end{equation}
and $T_i(z^\prime)\neq t$. If $t^\prime=0$, we can use Assumption~\ref{assn:ref:treat} to get a stronger implication: since $U_z(0)=U_{z^{\prime\prime}}(0)=0$ for all $z^{\prime\prime}$,  \eqref{eq:tzt} holds for all $z^{\prime\prime}$ and $T_i(z^{\prime\prime})\neq t$. We summarize this in \Cref{prop:csq-target}.

\begin{proposition}[Consequences of Targeting]\label{prop:csq-target}
If $z$ targets $t$ and $T_i(z)=t^\prime\neq t$, 
then under Assumptions %~\ref{assn:valid}, 
\ref{assn:ARUM}, and \ref{assn:ref:treat}:
\begin{enumerate}
\item[(i)]  for all $z^\prime$ that target $t^\prime$, we have $T_i(z^\prime)\neq t$; \
\item[(ii)] if $t^\prime=0$, then $T_i(z^{\prime\prime})\neq t$ for all $z^{\prime\prime}\in\calZ$.
\end{enumerate}
\end{proposition}

%\Bernardnote{Add testable implications on $P(t\vert z)$}

\Cref{prop:csq-target}~(i) implies that the events $T_i(z)=t^\prime$ and $T_i(z^\prime)=t$ are mutually exclusive. Hence, the sum of their probabilities cannot exceed one:
\begin{equation}
P(t^\prime\mid z) + P(t\mid z^\prime)\leq 1 \quad \text{if } \;
z \text{ targets } t   \text{ and }   z^\prime \text{ targets } t^\prime\neq t.
\label{ineq:targeting1}
\end{equation}
Similarly, \Cref{prop:csq-target}~(ii) implies that the event $T_i(z)=0$ is disjoint from the event that $T_i(z'')=t$ for \emph{any} $z'' \in \calZ$. Since the probability of this latter event is at least $\max_{z''} P(t \mid z'')$, we obtain the following implication:
\begin{equation}
P(0\mid z) + \max_{z^{\prime\prime}\in \mathcal{Z}, z^{\prime\prime}\neq z} P(t\mid z^{\prime\prime})\leq 1 \quad \text{if } \;
z \text{ targets } t.
\label{ineq:targeting2}
\end{equation}
%  
% It is easy to see that these inequalities are sharp: for  any mapping $z\to T^\ast(z)$, any set of generalized propensity scores $P(z\vert t)$ that satisfies \Cref{ineq:targeting1} and \Cref{ineq:targeting2} can be rationalized by a model that satisfies Assumptions~\ref{assn:ARUM} and~\ref{assn:ref:treat}.}

% \Bernardnote{Not  proved yet! may be wrong.}

\subsection{One-to-one Targeting}

Now suppose that each  $z$ consists of a set of (possibly zero or negative) subsidies $S_z(t)$ for treatments $t\in\calT$. If there is a no-subsidy regime $z=0$ with $S_0(t)=0$ for all $t$, it   seems natural to write the mean value  as  $U_z(t)=U_0(t)+S_z(t)$. Then  for any treatment $t$, the set $Z^\ast(t)$ consists of the instrument values $z$ that subsidize $t$ most heavily. 
As this illustration suggests,  the sets $Z^\ast(t)$ may not be singletons, and they may well intersect. 
We now introduce a more restrictive definition that rules out these two possibilities.

\begin{definition}[One-to-one targeting]
Targeting is {\em one-to-one} when both $Z^\ast: \calTs\to\calZs$ and $T^\ast:\calZs\to\calTs$ are functions.    
\end{definition}
Under one-to-one targeting, we will often write ``$z=t$'' if $z$ targets $t$; this is without loss of generality.
Let us illustrate these varieties of targeting on Example~\ref{ex:ternaryternary}.

\medskip

%	\begin{example}[Targeted binary instrument model\label{ex:binaryinstrument}]
	
%	\Bernardnote{Took out the first example}
	
% \noindent
% {\bf Example~\ref{ex:binaryinstrument0} continued.}	\Simonnote{This seems odd with the description before ``agents in the control group $Z=0$ are not treated.'' Perhaps we simplify the example and go with $\calTs = \calZs =\{ 1 \}$}
% 	Take our binary instrument model (\Cref{ex:binaryinstrument0}) and denote $a_t=U_0(t)$ and $b_t=U_1(t)$. Suppose that    the observations with $z=1$  receive a subsidy for the  treatment $t=1$, so that $b_1>a_1$. Then $t=1$ is targeted by $z=1$. If another treatment value, say $t=2$, is also targeted, then targeting can only be injective if $Z^\ast(2)=\{0\}$; and no other treatment value can be targeted.  This gives $a_2>b_2$ and $b_t=a_t$ for $t>2$. If these conditions hold, then targeting is also one-to-one.

% %	\end{example}

% \medskip

\begin{table}[htbp]
	{\small
		\centering
		\begin{tabular}{cccc}
		  \toprule
		 & $t=0$ & $t=1$ &  $t=2$  \\
		  \midrule
		$z=0$ &  $0$ & $a$ & $d$ \\
		$z=1$ &  $0$ & $b$ & $e$ \\
		$z=2$ &  $0$ & $c$ & $f$ \\
		   \bottomrule
		\end{tabular}
		\caption{Values of $U_z(t)$ in the $3\times 3$ model}
		\label{tab:ternaryternary:modestargeting}
	}
	\end{table}
	
\noindent
{\bf Example~\ref{ex:ternaryternary} continued.}
 Table~\ref{tab:ternaryternary:modestargeting} shows  the values of $U_z(t)$ in the $3\times 3$ model of Example~\ref{ex:ternaryternary}. Suppose that $t=1$ is targeted; choose some $z$ that targets it and relabel it as $z=1$. This means that 
\[
	b\geq \max(a,c)  \; \mbox{ and } \; b>\min(a,c).
\]
If $t=2$ is also targeted  by some $z\neq 1$,  we relabel this instrument value as $z=2$. 
%\Bernardnote{eliminated the injective mention.}
 This gives
\[
	f\geq \max(d,e)  \; \mbox{ and } \; f>\min(d,e).
\]
Finally, if targeting is one-to-one we have $b>\max(a,c)$ and $f>\max(d,e)$.

%\todo{I moved up Proposition~\ref{proprespvectors1} as it seemed strange, on reflection, to start the strict targeting section w/o strict targeting.}

\medskip

One-to-one targeting implies some useful restrictions on response-groups.  Take an observation $i$ and $t\in\calTs$. 
Under one-to-one targeting, treatment $t$ is targeted only by $z=t$. 
As a result, by \Cref{prop:csq-target}, 
 if $T_i(t)=t^\prime$ and $t^\prime\in \calTs\setminus \{t\}$ then $T_i(t^\prime)\neq t$; and  if $T_i(t)=0$ then $T_i(z)$ can never equal $t$.
These restrictions are summarized in the following proposition.

\begin{proposition}[Response-groups under one-to-one targeting]\label{proprespvectors1}
Under Assumptions %~\ref{assn:valid}, 
\ref{assn:ARUM}, \ref{assn:ref:treat}, and one-to-one targeting, take a targeted treatment  $t\in\calTs$.
\begin{enumerate}[label=(\roman*)]
	\item All response-groups $C_R$ with $R(t)=t^\prime\in\calTs\setminus \{t\}$ and $R(t^\prime)=t$ are empty. 
	\item All response-groups $C_R$ with $R(t)=0$ and $R(z)=t$ for some  $z\neq t$ are empty.
\end{enumerate}
\end{proposition}

\medskip

\noindent
{\bf Example~\ref{ex:ternaryternary} (continued)}
Return to the $3\times 3$ model and to Table~\ref{tab:ternaryternary:modestargeting}. Suppose that both $t=1$ and $t=2$ are targeted.  Under the conditions of Proposition~\ref{proprespvectors1}, we have 
$b>\max(a,c)$ and $f>\max(d,e)$.

 Since the points $P_z$ have coordinates $(-U_z(1),-U_z(2))$,
\begin{itemize}
	\item $P_1=(-b,-e)$ must lie to the left of $P_0=(-a,-d)$ and %(strictly) 
	of $P_2=(-c,-f)$,
	\item $P_2$ must lie below $P_0$ and %(strictly)
	 $P_1$.
\end{itemize}
This is easily rephrased in terms of the response-vectors of definition~\ref{def:cgroups}. First note that 
in the $3\times 3$ case, there are  a priori $3^3=27$ response-vectors,  $R=000$ to $R=222$, with corresponding response-groups $C_{000}$ to $C_{222}$. Groups $C_{ddd}$ are ``always-takers''\footnote{Observations in group $C_{000}$ are usually called the ``never-takers''. We prefer not to break the symmetry in our notation. We hope this will not cause confusion.} of treatment value $d$. All other groups are ``compliers'' of some kind, in that their treatment changes under some changes in the instrument. We will also pay special attention to some non-elemental groups. For instance, $C_{0\ast 2}$ will denote the group who is assigned treatment 0 under $z=0$ and treatment 2 under $z=2$, and any treatment under $z=1$. That is,
\[
C_{0\ast 2} = C_{002} \cup C_{012} \cup C_{022}.
\]
One-to-one targeting implies the emptiness of four composite groups out of the 27 possible.   For any treatment value $\tau$, Proposition~\ref{proprespvectors1}(ii) rules out  group $C_{10\tau}$ since  this group has $R(1)=0$ and $R(0)=1$. It rules out $C_{\tau 01}$ as $R(1)=0$ and $R(2)=1$. This eliminates the composite groups 
$C_{10\ast}$ and $C_{\ast 01}$. The same argument applies to composite groups  $C_{\ast 20}$ and  $C_{2\ast 0}$, which have $R(2)=0$ and $R(1)=2$ or $R(2)=2$. 

These four composite groups correspond to 10 elemental groups\footnote{Specifically, they are: $C_{100}, C_{101}, C_{102}, C_{001}, C_{201}, C_{020}, C_{120}, C_{220}, C_{200}$, and $C_{210}$.}.
This still leaves us with 17  elemental groups, and potentially  complex assignment patterns.  
Consider for instance Figure~\ref{fig:terterum}. It shows one possible configuration for the $3\times 3$ model; the  positions for $P_0, P_1$ and $P_2$ are  consistent with one-to-one targeting.

\begin{figure}[htb]
	\begin{tikzpicture} [scale=0.5] 
\draw [blue!20!white, fill=blue!20!white] (-5,1.5) -- (-2.5,1.5) -- (1,5) -- (-5,5);
\node[black] at (-2.5,3.5) {$C_{222}$};
\draw [green!20!white, fill=green!20!white] (1.5,-2.5) -- (1.5,-5)  -- (5,-5) -- (5, 1);
\node[black] at (3.5,-3.5) {$C_{111}$};
\draw [pink!20!white, fill=pink!20!white] (1.5,-2.5) -- (0,-2.5)  -- (0,-5) -- (1.5, -5);
\node[black] at (0.7,-3.5) {$C_{110}$};
\draw [yellow!20!white, fill=yellow!20!white] (0,-2.5) -- (0,-5)  -- (-2.5,-5) -- (-2.5, -2.5);
\node[black] at (-1.5,-3.5) {$C_{010}$};
\draw [blue!20!white, fill=blue!20!white] (-2.5,-2.5) -- (-5,-2.5)  -- (-5,-5) -- (-2.5, -5);
\node[black] at (-3.5,-3.5) {$C_{000}$};
\draw [yellow!20!white, fill=yellow!20!white] (-5,1.5) -- (-2.5,1.5)  -- (-2.5,0) -- (-5, 0);
\node[black] at (-3.5,0.7) {$C_{202}$};
\draw [pink!20!white, fill=pink!20!white] (-5,0) -- (-5,-2.5)  -- (-2.5,-2.5) -- (-2.5, 0);
\node[black] at (-3.5,-1.3) {$C_{002}$};
\draw [blue!20!white, fill=blue!20!white] (-2.5,-2.5) -- (0,-2.5)  -- (0, 0) -- (-2.5, 0);
\node[black] at (-1.3,-1.3) {$C_{012}$};
\draw [yellow!20!white, fill=yellow!20!white] (0,0) -- (0,-2.5)  -- (1.5, -2.5) -- (5, 1) -- (5,5) --(0,0);
\node[black] at (2,0.5) {$C_{112}$};
\draw [pink!20!white, fill=pink!20!white] (0,0) -- (5,5)  -- (1, 5) -- (-2.5, 1.5) -- (-2.5,0) --(0,0);
\node[black] at (0.7,2.6) {$C_{212}$};
	\draw[->,>=latex] (-5,0)  -- (5,0) node[below] {$u_{i1}-u_{i0}$};
	\draw[->,>=latex] (0,-5)  -- (0,5) node[above] {$u_{i2}-u_{i0}$};
		\draw[-] (0,0)  -- (5,5);
		\draw[-] (0,0)  -- (0,-5);
		\draw[-] (0,0)  -- (-5,0);
		\draw[-] (1.5,-2.5)  -- (-5,-2.5);
		\draw[-] (1.5,-2.5)  -- (1.5,-5);
		\draw[-] (1.5,-2.5)  -- (5,1);
		\draw[-] (-2.5,1.5)  -- (1,5);
		\draw[-] (-2.5,1.5)  -- (-5,1.5);
		\draw[-] (-2.5,1.5)  -- (-2.5,-5);
		\fill[blue] (0.0, 0.0) circle (.1cm);	          
		\fill[blue] (1.5,-2.5) circle (.1cm);	          
		\fill[blue] (-2.5,1.5) circle (.1cm);	      
	\node[blue] at (0.35,-0.35) {$P_0$};
		\node[blue] at (-2.4,1.8) {$P_1$};
		\node[blue] at (1.2,-2.0) {$P_2$};
	\end{tikzpicture} 
	\caption{A $3\times 3$ example}\label{fig:terterum}
	\end{figure}

   The number of distinct response-groups (ten in this case) and the contorted shape of the $C_{212}$ and $C_{112}$ groups in Figure~\ref{fig:terterum} point to the difficulties we  face in  identifying response-groups without further assumptions. Moreover, this is only one possible configuration: other cases exist, which would bring up other response-groups. 

 \citet[pp.\ 16--20]{heckmanpinto-pdt}, \citet{pinto2021}, and~\cite{kirkeboen2016} also studied the $3\times 3$ model;  they proposed sets of assumptions that identify some treatment effects. The example in
 \citet[pp.\ 16--20]{heckmanpinto-pdt} is rather specific. We show in \Cref{appendix:model:pinto} how  to apply our framework  to  the Moving to Opportunity experiment studied in \citet{pinto2021}. The setup  in~\cite{kirkeboen2016} is most similar to ours; we will return to the differences between our approach and theirs in Section~\ref{subsec:ternary-ternary}. $\qed$

\subsection{Strict Targeting}\label{sub:strict}
Figure~\ref{fig:terterum} suggests that if we could make sure that $P_1$ is directly to the left of $P_0$, the shape of $C_{212}$ would become nicer---and group $C_{202}$ would be empty. Bringing $P_2$ directly   under $P_0$ would have a similar effect. This translates directly into assumptions on the dependence of the $U_z(t)$ on the instruments: the first one imposes
	$d=e$ and the second one imposes 
	$a=c$.   This can be interpreted as policy regime $z=1$ (resp.\ $z=2$) subsidizing treatment $t=1$ (resp.\ $z=2$) only.
To return to the general model, there are applications in which
the instruments $z\in Z^\ast(t)$, which maximize $U_z(t)$, do not shift assignment between the other values of the treatment.
The following definition is a direct extension of this discussion.
	
	\begin{definition}[Strict Targeting of Treatment $t$]\label{def:strict-targeting}
		Take any  targeted treatment value  $t\in\calTs$. It is {\rm strictly targeted} if the function
$z \in \calZ \to U_{z}(t)$
takes the same value for all  instruments that do not target $t$ (the values $z\not\in Z^\ast(t)$). 
 We  denote this common value by $\underline{U}(t)$, and we will say of the instrument values $z\in Z^\ast(t)$  that they {\rm strictly target} $t$.
	\end{definition}

Note that strict targeting only bites  if $\calZ$ contains at least three instrument values. If $\abs{\calZ}=2$ (one binary instrument, as in our Example~\ref{ex:binaryinstrument0}) and say $z=1$ targets $t$, then $\calZ \setminus Z^\ast(t)$ can only consist of  $z=0$ and Assumption~\ref{assn:full-targeting} trivially holds.

    Suppose for instance that the data comes from a randomized experiment, where the instrument value $z=t$ targets treatment $t$. If compliance is imperfect, an individual will trade off the benefits from switching to a treatment $t^\prime\neq t$ with the costs of the effort required. Strict targeting obtains when the cost of switching to $t^\prime$ does not depend on the value of $t$.
%\Bernardnote{Added from answer to R2.}

Under strict targeting, turning on instrument $z\in Z^\ast(t)$ promotes treatment $t$ without affecting the   mean values  $U_z(t^\prime)$ of other treatment values $t^\prime$.  This explains
our use of the term ``strict targeting''. 
In this  ARUM specification, an instrument in  $Z^\ast(t)$ plays the   same role as a  price discount on  good $t$ in a  model of demand for goods whose mean values only depend on their own prices.  In the language of program subsidies, all $z\in Z^\ast(t)$ subsidize $t$ at the same high rate, and all other instrument values offer the same, lower subsidy.

% Note that while we only state the assumption for  the targeted treatments $t\in \calTs$, it holds by definition for all non-targeted treatments. Since $Z^\ast(t)=\mathcal{Z}$ for these treatment values, $\underline{U}(t)=\bar{U}(t)$ is the common value of $U_z(t)$ over all of $\calZ$  when $t\not\in \calTs$.

%Strict targeting  also holds in our  Example~\ref{ex:TbyT0}. 
% It only holds in the $4\times 3$  design of \Cref{ex:rectangular} if $U_{0\times 1}(1)=U_{0\times 0}(1)$ and $U_{1\times 0}(2)=U_{0\times 0}(2)$   (so that $z=0\times 1$ does not subsidize $t=1$ and $z=1\times 0$ does not subsidize $t=2$).

Finally, we should emphasize that one-to-one targeting and strict targeting are logically independent assumptions: neither one implies the other. Consider  the $3\times 3$ model of Example~\ref{ex:ternaryternary} under one-to-one targeting;  strict targeting only holds for $t=1$ if $a=c$, and for $t=2$ if $d=e$. On the other hand, the $3\times 3$ model with $b>a=c$ and $e>d=f$ satisfies strict targeting but not one-to-one targeting, as $z=1$ targets both $t=1$ and $t=2$.

\subsubsection{Consequences of Strict Targeting}\label{sub:csqces-strict-targeting}

Now consider  the general model. If a treatment $t$ is strictly targeted, then $U_z(t)$ can only take one of two values: $\bar{U}(t)$ if z targets $t$, and $\underline{U}(t)$ otherwise. By definition, if $t$ is not targeted then the value of $U_z(t)$ does not depend on $z$; we also denote it $\underline{U}(t)$. 
We will assume in this subsection that all targeted treatments are strictly targeted:
\begin{assumption}[Full strict targeting]\label{assn:full-targeting}
 If $t$ is in $\calTs$, then $t$ is strictly targeted.
\end{assumption}

Under full strict targeting, the values of $U_z(t)$ are given in Table~\ref{tab:uxt:strict}. 
The functions $\underline{V}_i$ and $\underline{t}_i$ that define the top alternative treatment for a given observation  are constant over  $\calZ\setminus\calZs$. The assigned treatment $T_i(z)$ is $\underline{t}_i$ for all such values of $z$, as well as for $z\in\calZs$ if  
 \[
 V^\ast_i(z)=\max_{t\in\calTs(z)} (\bar{U}_t+u_{it}) < \underline{V}_i=\max_{t\not\in\calTs(z)} (\underline{U}_t+u_{it});
 \]
it is  the maximizer  $t^\ast_i(z)$ of $V^\ast_i(z)$ otherwise.

\begin{table}[htbp]
	{\small
	\centering
	\begin{tabular}{ccc}
	  \toprule
	  & $t\in T^\ast(z)$ & $t\not\in T^\ast(z)$\\ \midrule
	 $z\in Z^\ast$ &  $\bar{U}(t)$ & $\underline{U}(t)$ \\
	 $z\not\in Z^\ast$ & - & $\underline{U}(t)$ \\ \bottomrule
	\end{tabular}
	\caption{Values of $U_z(t)$ under full strict targeting}
		\label{tab:uxt:strict}
	}
\end{table}

Note that in a sense,  
all instrument values in $\calZnotZs$ are equivalent under universal strict targeting. If $z$ and $z^\prime$ are two such values, then both functions $U_{z}$ and $U_{z^\prime}$ equal $\underline{U}$ on all of $\calT$ and   
     the counterfactual  treatments $T_i(z)$ and $T_i(z^\prime)$  must be  $\underline{t}_i$ for any observation $i$.

\subsection{Universal targeting}\label{sub:universal}

We now impose one-to-one targeting as well as strict targeting for every targeted treatment value:

\begin{assumption}[Universal targeting]\label{assn:universal}
 If $t$ is in $\calTs$, then $t$ is strictly targeted, and it is uniquely targeted by $z=t$.
\end{assumption}
Note that \Cref{assn:universal} implies \Cref{assn:full-targeting}.

\begin{table}[htbp]
    % 1. Increase row height slightly for better math readability
    \renewcommand{\arraystretch}{1.2} 
    {\small
    \centering
    \begin{tabular}{ccc}
    \toprule
    % 2. Multirow in the first column
    \multirow{2}{*}{$z\in Z^\ast$}  & $t=z$ & $t\neq z$ \\ 
    
    % 3. Partial horizontal line spanning columns 2 and 3 only
    % (lr) trims the line slightly on left and right for valid separation
    \cmidrule(lr){2-3} 
    
      & $\bar{U}(t)$ & $\underline{U}(t)$ \\ \midrule
    
    $z\not\in Z^\ast$ & \multicolumn{2}{l}{$\underline{U}(0)$, $\underline{U}(1)$, \ldots, $\underline{U}(|\mathcal{T}|-1)$}  \\ \bottomrule
    \end{tabular}
    \caption{Values of $U_z(t)$ under universal targeting}
    \label{tab:uxt:strict:1to1}
    }
\end{table}

Under \Cref{assn:universal}, \Cref{tab:uxt:strict}
becomes~\Cref{tab:uxt:strict:1to1} and we have:

\begin{proposition}[Treatment assignment under universal targeting\label{corclasses-strict-11}]
	 Take any observation $i$. Let $A_i$ be the (possibly empty) subset of $\calTs$ over which  $T_i(t)=t$. Then under Assumptions %~\ref{assn:valid}, 
     \ref{assn:ARUM} and~\ref{assn:universal},
	\begin{enumerate}
			% \item if targeting is one-to-one, $T_i(z)=\underline{t}_i(z)$ for all $z\in \calZnotAi$;
   %          	\item if targeting is universally strict (\Cref{assn:full-targeting}), $T_i(z)=\underline{t}_i$ does not depend on $z$  for all $z\in \calZnotZs$;
        \item[(i)]  if $z\in \calZnotAi$ then $T_i(z)=\underline{t}_i$.
	 \item[(ii)] $\underline{t}_i$ belongs to the set $\bar{A}_i\equiv A_i\cup (\calT\setminus\calTs)$.
	\item[(iii)]
	 The pair $(A_i,\underline{t}_i)$ defines an elemental response group which we denote $C(A_i,\underline{t}_i)$. The family of sets $\{ C(A,t) \; \vert \; A\subset \calTs, t\in \bar{A} \}$ form a partition of the set of observations.
     \end{enumerate}
\end{proposition}

Note that the $C(A,t)$  notation is just a shortcut: every $C(A,t)$ is an elemental group,  and every elemental group is a $C(A,t)$. If for instance $\acalT=6$, it is just more convenient to write $C(\{1,3\}, 2)$ than to write $C_{212322}$.

If the set $A_i$ is non-empty and  $\underline{t}_i\in A_i$,   the observation  $i$ is what one could call  a {\em strict $A_i$-complier}:   when the value of the instrument moves from $\calZnotA_i$ to a $t\in A_i$, observation $i$   switches from its top alternative treatment $\underline{t}_i$  to the top targeted treatment $t$. In the 3-by-3 model with $\calTs=\{1,2\}$, there are three groups of strict compliers:   $C_{010}=C(\{1\}, 0)$, $C_{002}=C(\{2\},0)$, and $C_{012}=C(\{1, 2\},0)$.

Universal targeting brings us  very close to the main identifying assumption in \citet[Assumption B-2a, p.\ 5006]{HV2007-handbook-2}: the indicator variable  $\uniset(Z=t)$ can be used as the $Z^{[t]}$ in their assumption.  \citeauthor*{HV2007-handbook-2} use their Assumption B-2a to identify the effect of the preferred treatment $t$ relative  to the next-best treatment. Their complier group consists of those   individuals who choose treatment $t$ under $Z=z$ and another treatment under $Z=z^\prime$.  This  can be  a very heterogeneous group, as our examples will show. To paraphrase \citet[p.\ 5013]{HV2007-handbook-2}:  the mean effect of treatment $t$ versus the next best option is a weighted average over $t^\prime\in \calT\setminus\{t\}$ of the effect of treatment $t$ versus treatment $t^\prime$, conditional on $t^\prime$ being the next best option, weighted by the probability that $t^\prime$ is the next best option. In contrast, we seek a complete characterization of all treatment effects that can be identified under this set of assumptions.

\begin{table}[htb!]
    \centering
    \begin{tabular}{ccccc} % {\color{purple}vertical rule removed per JoE style}
         & $t=0$ & $t\in \mathcal{T}\setminus(\mathcal{T}^\ast\cup\{0\})$ & $t\in\mathcal{T}^\ast$
         & $t^\prime=z^\prime$ \\ \hline
                    $z\not\in \mathcal{Z}^\ast$ & 0 & $\underline{U}_t$ & $\underline{U}_{t}$ & $\underline{U}_{t^\prime}$ \\
         $z^\prime\in\calZs\setminus\{t\}$ & 0 & $\underline{U}_t$ & $\underline{U}_t$ & $\bar{U}_{t^\prime}$  \\
          $z=t\in\calZs$ & 0 & $\underline{U}_t$ & $\bar{U}_t$  & $\underline{U}_{t^\prime}$ \\
    \end{tabular}
    \caption{Universal targeting}
    \label{tab:ust1}
\end{table}

It is easy to show that universal targeting yields  two sets  of testable implications. First, consider any treatment value $t\in\calT$. It is clear from the first row of Table~\ref{tab:ust1} that the mean utilities $U_z(t)$ are the same for  all $z\in\calZ\setminus \calZs$. Since $0\not\in\calZs$, it follows that 
\begin{equation}\label{eq:test:targeting1}
P(t\vert z) =P(t\vert 0) \text{ for every } t\in\calT  \text{ and } z\in\calZ\setminus\calZs.
\end{equation}

Next, let $t\in\mathcal{T}^\ast$ be a targeted treatment value and take an instrument $z^\prime$  that targets a different treatment $t^\prime \neq t$. Comparing the rows of Table~\ref{tab:ust1} reveals the ordering of utilities. When $z=t$, treatment $t$ is boosted; when $z=0$, it is at baseline; when $z=t^\prime$, the rival is boosted (drawing share away from $t$). This implies:
\begin{equation}\label{eq:test:targeting2}
    P(t\vert t)> P(t\vert 0)>P(t\vert t^\prime)  \quad \text{for all } (t, t^\prime) \in (\mathcal{T}^\ast)^2 \text{ with } t \neq t^\prime.
\end{equation}
Combining these inequalities implies that for all targeted instrument values $t\in\mathcal{T}^\ast$, the choice probability function $z \mapsto P(t\mid z)$ is strictly maximized at $z=t$.

% Now suppose that $t\in\calTs$ is a targeted value. If $z^\prime$ targets a $t^\prime\neq t$, then comparing the  three rows of the table shows that $P(t\vert t)> P(t\vert 0)>P(t\vert z^\prime)$. This can be written as
% \begin{equation}\label{eq:test:targeting2}
% P(t\vert t)> P(t\vert 0)>P(t\vert t^\prime) 
% \text{ for all } (t, t^\prime\neq t) \in (\calTs)^2.
% \end{equation}
%  Note that combining these  inequalities implies that for all targeted instrument values $t\in\calTs$, the function $z\in\calZ \to P(t\vert z)$ is maximal in $z=t$ (which is also obvious from \Cref{tab:ust1}).

Moreover, the ``encouragement design'' assumption of a recent paper by \cite{Bai:Tabord-Meehan:2025:arXiv} is exactly equivalent to universal targeting when (in their notation) $J_0>0$ and $z=0$ is what they call a ``base state''.  They show that encouragement design with a base state generates the following testable implications:
\begin{equation}
    \label{eq:test:targeting3}
P(t\vert z) \leq P(t\vert 0) \text{ if } z \text{ does not target } t.
\end{equation}
and that these implications are sharp: any set of propensity scores that satisfies these inequalities can be rationalized by an ARUM of encouragement design with a base state.

% Combining~\Cref{eq:test:targeting1}
% and the rightmost inequalities in~\Cref{eq:test:targeting2} indeed gives~\Cref{eq:test:targeting3} (with equality if and only if $z\not\in\calZs$);  the leftmost inequalities in~\Cref{eq:test:targeting3} follow from the fact that conditional probabilities must add to one.

% This allows us to transport their results on   the sharp implications of encouragement design: under our assumptions, the inequalities~\eqref{eq:test:targeting1} and~\eqref{eq:test:targeting2} exhaust the testable implications of universal targeting for generalized propensity scores.

In our framework, combining \eqref{eq:test:targeting1} for $z\in\calZ\setminus\calZs$ and the rightmost inequality in \eqref{eq:test:targeting2} (strict inequality for rival-targeting $z$) yields  \eqref{eq:test:targeting3}. Conversely, the fact that choice probabilities must sum to one ensures that if $P(t \mid z)$ decreases for all non-target treatments, $P(t \mid t)$ must increase, satisfying the leftmost inequality of \eqref{eq:test:targeting2}. Thus, under our assumptions, conditions \eqref{eq:test:targeting1} and \eqref{eq:test:targeting2} exhaust the testable implications of universal targeting.

\section{Identification}\label{sec:identif}

We use the standard counterfactual notation:
$T_i(z)$ and $Y_i(t,z)$ denote  respectively potential treatments and outcomes.  Let $\uniset(A)$ denote the indicator of set $A$.
The validity of the instruments requires the usual exclusion restriction together with appropriate independence conditions.

\begin{assumption}[Valid Instruments]\label{assn:valid}
\begin{itemize}
    \item[(i)]  
    $Y_i(t,z)=Y_i(t)$ for all $(t,z)\in\mathcal{T}\times\mathcal{Z}$.
    \item[(ii)] 
    For all $(t,z)\in\mathcal{T}\times\mathcal{Z}$ and all response groups $C$,
\begin{align*}
    \Pr(i\in C \mid Z_i=z) &= \Pr(i\in C); \\
   \mathbb{E}(Y_i(t) \mid i\in C, Z_i=z)
    &= \mathbb{E}(Y_i(t) \mid i\in C).
\end{align*}
\end{itemize}
\end{assumption}

One could impose a stronger condition than \Cref{assn:valid}(ii), such as the joint independence of $\{ (Y_i(t), T_i(z)): (t,z) \in \mathcal{T}\times \mathcal{Z}\}$ from $Z_i$. However, we adopt the weaker mean-independence assumption for $Y_i(t)$ in \Cref{assn:valid}(ii), as it is sufficient for the results presented in this paper.
\Cref{assn:valid}(ii) can be viewed as a generalization of Assumption 2 in \citet{huber2017sharp}, who only consider the case of binary instruments and binary treatments.
Under \Cref{assn:valid}, we define $T_i := T_i(Z_i)$ and  $Y_i := Y_i(T_i)$.
Throughout the paper, we assume that we observe $(Y_i, T_i,  Z_i)$ for each $i$.

To simplify the exposition, we introduce one more element of notation. For any $z\in\calZ$ and $t\in \calT$, we define the 
 {\em conditional average outcome\/} by
\[
\bar{E}_z(t) \equiv \mathbb{E}(Y_i \uniset(T_i=t) \mid Z_i=z).
\]
For any response-group $C$ and treatment value $t\in\calT$, we define the {\em group average outcome\/} as $\mathbb{E} (Y_i(t) \mid i \in C)$. Our goal is to identify the proportions of each response-group, $\Pr(i\in C)$, and their group average outcomes.

While the conditional average outcome $\bar{E}_z(t)$ is directly identified from the data, the group average outcomes are not; they combine with the group probabilities to form the conditional average outcomes. We will repeatedly use the following identity:

\begin{lemma}[Group- and conditional average outcomes---Theorem T-1 of \citet{heckmanpinto-pdt}]\label{lemma:ezt-identity}
Let Assumption~\ref{assn:valid} hold. Then, for all $(t,z)\in\mathcal{T}\times\mathcal{Z}$,
 \begin{align*}
 P(t|z) &= \sum_{C_R \,\vert\, R(z)=t} \Pr(i\in C_R); \\
 \bar{E}_z(t) &= 
 \sum_{C_R \; \vert \; R(z)=t} \mathbb{E}(Y_i(t) \; \vert i\in C_R)\Pr(i\in C_R).
 \end{align*}
\end{lemma}

If there are $\abs{\calR}$ response-groups, the first equation in Lemma~\ref{lemma:ezt-identity} is a system of $(\abs{\calT}-1) \times \abs{\calZ}$ equations with $(\abs{\calR}-1)$ unknowns. 
Under monotonicity, the binary-binary model $\abs{\calT} =\abs{\calZ}= 2$ generates the standard LATE case, where the response-groups consist of never-takers ($C_{00}$), compliers ($C_{01}$), and always-takers ($C_{11}$). As is well known, the sizes of these three groups are just identified. On the other hand, even under \Cref{assn:universal}, one restriction is required to identify the sizes of the response-groups for the $3 \times 3$ model.\footnote{We prove this in \Cref{pro:DCM:vector:ternaryternary}.} More generally, the degree of underidentification of group sizes tends to increase exponentially with the number of targeted treatments, like the number of response-groups.

Each of our assumptions on targeting reduces the number of response-groups and therefore the degree of underidentification. Take our strongest assumption: universal targeting. Then the set of response-groups $C=C_R$ such that $R(z)=t$ is as enumerated in Proposition~\ref{prop:ident-ufstrict11}: it consists of
\begin{itemize}
	\item all $C(A,t)$ such that $A\subset \calTs \setminus \{z\}$ and $t\in \bar{A} \equiv A \cup (\calT\setminus\calTs)$;
	\item if $t=z$, all $C(A,\tau)$ for $z\in A\subset \calTs$ and $\tau\in \bar{A}$.
\end{itemize}
This gives directly the system of identifying equations.

\begin{proposition}[Identifying equations under universal targeting]\label{prop:ident-ufstrict11}
	Let Assumptions \ref{assn:ARUM}, \ref{assn:ref:treat}, \ref{assn:universal}, and \ref{assn:valid} hold. Then, for all $(z, t)\in \calZ \times \calT$:
	\begin{equation}\label{eq:identuf:strict11}
	\begin{split}
		\bar{E}_z(t) &= \sum_{A\subset \calTs \setminus \{z\}} \uniset(t \in \bar{A}) \mathbb{E}[Y_i(t) \mid i\in C(A,t)]\Pr(i\in C(A,t)) \\
		&\quad + \sum_{A\subset \calTs} \uniset(t=z\in A) \sum_{\tau\in \bar{A}} \mathbb{E}[Y_i(t) \mid i\in C(A,\tau)] \Pr(i\in C(A,\tau)), \\[1em]
		P(t|z) &= \sum_{A\subset \calTs \setminus \{z\}} \uniset(t \in \bar{A}) \Pr(i\in C(A,t)) \\
		&\quad + \sum_{A\subset \calTs} \uniset(t=z\in A) \sum_{\tau\in \bar{A}} \Pr(i\in C(A,\tau)).
	\end{split}
	\end{equation}
\end{proposition}

\subsection{Positive Selection and the Empirical Content of the Model}\label{sub:pos:sel}

We now introduce an identifying assumption that we call {\em positive selection}. It obtains when a function $h$ of the vector of potential outcomes $\bm{Y}_i\equiv (Y_i(0),\ldots,Y_i(\abs{\calT}-1))$ has a lower expectation for a response-group $C$ than for another response-group $C^\prime$:
\begin{equation}\label{eq:h-pos-sel}
\mathbb{E}(h(\bm{Y}_i) \;\vert \;i\in C) \leq
\mathbb{E}(h(\bm{Y}_i) \;\vert \; i\in C^\prime). 
\end{equation}
Two leading examples are (a) $h(\bm{Y}_i)\equiv Y_i(t)$ for some $t$ and (b) $h(\bm{Y}_i)\equiv Y_i(t)-Y_i(t^\prime)$ for some $t\neq t^\prime$. 
In the $2 \times 2$ model, \citet{huber2017sharp} consider a variant of form~(a), which they call mean dominance.
The identifying power of positive selection depends on the context. We will illustrate it using form~(a)---that is, a restriction on the level of a potential outcome across groups---in \Cref{prop:DCM:late:binaryinstrument:PI}, as well as in our application to Head Start in Section~\ref{sec:empirical}.
We explore form~(b)---a restriction on treatment effect differences across groups---in \Cref{cor:IV:3by3}.

Before turning to specific applications, let us discuss the characterization of the empirical content of our framework under positive selection.
A series of papers\footnote{See \citet{balkepearl:97}, \citet{kitagawa:15}, \citet{mourifiewan:17}, \citet{kedagnimourifie:20}, and \citet{sun:23}.} has provided necessary and, in some cases, sufficient conditions for data to be rationalized under an instrument exclusion restriction. 
Most recently, \citet{Bai:Tabord-Meehan:2025:arXiv} characterized the sharp testable implications of what they call ``encouragement design'' under joint independence, that is, when $Z_i$ is independent of $\{(Y_i(t), T_i(z)) : (t,z) \in \calT \times \calZ\}$.  
Because we assume only mean independence in \Cref{assn:valid} and impose positive selection via \eqref{eq:h-pos-sel}, their testable implications involving outcomes $Y_i$ do not directly apply to our framework. However, their sharp restrictions on the generalized propensity scores $P(t\vert z)$ do apply, and we discuss these in our leading examples below. It is worth noting that while they do not impose an ARUM structure, their ``encouragement design'' condition is satisfied in our context only under our strongest assumption: universal targeting.

In the following sections, we characterize 
the empirical content for the restrictions given in 
\eqref{eq:identuf:strict11}, together with a positive selection assumption $\mathbb{E}(Y_i(t) \mid i \in C) \leq \mathbb{E}(Y_i(t) \mid i \in C^\prime)$ for given $C$ and $C^\prime$. 
The primitives of the model are the response-group probabilities $\Pr(i \in C)$ and the group average outcomes $\mathbb{E}( Y_i(t) \mid i \in C)$. 
The restrictions implied by \eqref{eq:identuf:strict11} generate bilinear equality constraints for
$\mathbb{E}( Y_i(t) \mid i \in C)$ and $ \Pr(i \in C)$, while positive selection yields linear inequality constraints for 
$\mathbb{E}( Y_i(t) \mid i \in C)$.
This makes it difficult to obtain simple yet general characterizations. Instead, we present explicit results for specific models in the next two subsections.

In particular, we study the \emph{identification region of local average treatment effects},
\[
\mathbb{E}[Y_i(t')-Y_i(t)\mid i\in C],
\]
for specific pairs $(t,t')$ and subgroups $C$.

 \subsection{The Binary Instrument Model}\label{subsec:bininstr}
 Recall that with a binary instrument, strict targeting is trivially satisfied.

 \subsubsection{Identification Under Universal Targeting}
 Under one-to-one targeting, $z=1$ targets only one instrument value, which we call $t=1$; and targeting is universal.
 \Cref{prop:ident-ufstrict11} can be applied directly and 
the group probabilities are just identified in our \Cref{ex:binaryinstrument0}.
  %\Bernardnote{Moved the $2\times 2$ example from above, made it $2\times T$.}
 % 
 Proposition~\ref{prop:ident-ufstrict11} gives $2(\abs{\calT}-1)$ independent equations:  for $t\neq 1$,
 \[ 
	P(t\vert 0)= \Pr(i\in C(\emptyset, t)) +\Pr(i\in C(\{1\}, t))   \; \mbox{ and } \; 
	P(t\vert 1)= \Pr(i\in C(\emptyset, t)).
	\]
	Moreover, $C(\emptyset,t)=C_{tt}$ for $t\neq 1$ and $C(\{1\},t)=C_{t1}$ for all $t$.

%\Bernardnote{It seemed more natural to show UM here.}

Note that when $z$ changes from 0 to 1, the only observations that change  treatment are in $C_{t1}$ for $t\neq 1$. Since the corresponding $C_{1t}$ group is empty, there are no ``two-way flows'' and this model satisfies the unordered monotonicity property of \cite{heckmanpinto-pdt}.
Proposition~\ref{pro:DCM:vector:binaryinstrument}  gives  explicit formul\ae\ for  the probabilities of all $(2\abs{\calT}-1)$ response groups.

	\begin{proposition}[Response-group probabilities in \Cref{ex:binaryinstrument0} under universal targeting]\label{pro:DCM:vector:binaryinstrument}
		Under  Assumptions~\ref{assn:ARUM}, \ref{assn:ref:treat},
			\ref{assn:universal}, and \ref{assn:valid},
		 the following probabilities are identified:
		  \begin{align}\label{prob:vectors:ternary/binary}
		  \begin{split}
			 \Pr(C_{11}) &= P(1 \vert 0),  \\
		 \Pr (C_{tt}) &= P(t\vert 1) \mbox{ for } t\neq 1,  \\
		 \Pr (C_{t1}) &= P(t \vert 0)-P(t \vert 1) \mbox{ for } t\neq 1.
		 \end{split}
		 \end{align}		 
		  \end{proposition}

Since $\Pr (C_{t1}) \geq 0$, the model has $(\abs{T}-1)$ simple testable predictions: $P(t\vert 0) \geq  P(t \vert 1)  \mbox{ for } t\neq 1$.
While all the response group probabilities are point-identified, only some group average outcomes are point identified without further restrictions, as shown by Proposition~\ref{pro:DCM:late:binaryinstrument}.
 
 \begin{proposition}[Group average outcomes in \Cref{ex:binaryinstrument0} under universal targeting]\label{pro:DCM:late:binaryinstrument}
	Under  Assumptions~\ref{assn:ARUM}, \ref{assn:ref:treat},
			\ref{assn:universal}, and \ref{assn:valid}, 
 the following group average outcomes are point-identified:
 \begin{align*}
 \mathbb{E} \left[ Y_i(1) \vert  i  \in C_{11} \right]  &= \frac{\bar{E}_0(1)}{ P(1 \vert 0)}, \\
 \mathbb{E} \left[ Y_i(t) \vert  i  \in C_{tt} \right]  &= \frac{\bar{E}_1(t)}{ P( t\vert 1)} \mbox{ for } t\neq 1, \\
 \mathbb{E} \left[ Y_i(t)  \vert  i  \in C_{t1} \right] &= \frac{\bar{E}_0(t) - \bar{E}_1(t)}{P(t \vert 0)-P(t \vert 1)} \mbox{ for } t\neq 1.
 \end{align*}
 However, if $T>2$ the standard Wald estimator only  identifies a convex combination of 
 the LATEs on the complier groups $C_{t1}$:
 \begin{align}
	 \frac{\mathbb{E}(Y_i\vert Z_i=1)-\mathbb{E}(Y_i\vert Z_i=0)}{\Pr(T_i=1 \vert Z_i=1)-\Pr(T_i=1\vert Z_i=0)} &= 
	 \frac{(\bar{E}_1(1)-\bar{E}_0(1))-\sum_{t\neq 1} (\bar{E}_0(t)-\bar{E}_1(t))}{P(1\vert 1)-P(1\vert 0)} \nonumber \\
 &= \sum_{t\neq 1} \alpha_t \mathbb{E} \left[ Y_i(1) -Y_i(t) \vert  i  \in C_{t1} \right], 
 \label{eq:DCM:late:binaryinstrument}
 \end{align}
 where  the weights  $\alpha_t=\Pr(i\in C_{t1}\vert i \in \cup_{\tau\neq 1} C_{\tau 1})=(P(t \vert 0)-P(t \vert 1))/(P(1\vert 1)-P(1\vert 0))$ are identified, positive, and sum to~1. 
If $T=2$, we have $\alpha_0=1$ and the familiar LATE formula
 \[
	 \mathbb{E}(Y_i(1)-Y_i(0)\vert i \in C_{01}) = 
	 \frac{\mathbb{E}(Y_i \vert Z_i=1)-\mathbb{E}(Y_i \vert Z_i=0)}{\Pr(T_i =1 \vert Z_i=1)-\Pr(T_i = 1 \vert Z_i = 0)}.
	 \]
 \end{proposition}

 %}

 \medskip

  \Cref{pro:DCM:late:binaryinstrument} shows that 
 we  only identify a known convex combination of the $(\abs{\calT}-1)$  LATEs\footnote{We use the term ``LATEs'' for
	 the average treatment effects on the various complier groups. Throughout the remainder of the paper, we  assume, as is standard, that probability differences appearing in the denominator 
	of estimands are always nonzero.}.
This formula is reminiscent of \citet[Theorem 1]{angrist1995two}, which deals with a different  model in which treatments are ordered.  It is possible to re-derive our identification results in Propositions \ref{pro:DCM:vector:binaryinstrument} and \ref{pro:DCM:late:binaryinstrument} using the general framework of \citet{heckmanpinto-pdt}. We provide details in \Cref{appendix:HP}.
 
 %\Bernardnote{This seemed like the right place to talk about Len's outcome-agnostic results.}   
 So far, we only imposed restrictions on the process by which treatment values are assigned to observations; this is what  \cite{goff:oa} calls an ``outcome-agnostic'' approach in that it only assumes that the instruments are excluded from the outcome equations.
 It is possible to bound the average treatment effects in a straightforward manner if we assume that the support of the outcomes is known and bounded. 
 One could instead add  restrictions to achieve point identification of average treatment effects for the compliers. Assuming that the ATEs are all equal is one obvious solution. Another one is to assume some degree of  homogeneity of group average outcomes.
 Alternatively, we may consider weaker conditions under which the average treatment effects for the compliers are only partially identified.
 We explore these ideas below.

 \subsubsection{Alternative Identification Constraints}

Consider the binary instrument model  with $T\geq 3$.

\paragraph{Beyond One-to-one Targeting}
First note that the probabilities of the response-groups can be identified under weaker restrictions than one-to-one targeting. 
Suppose for instance that $z=1$ targets all treatment values $t\geq 1$: we have $U_1(t)>U_0(t)$ for all $t\geq 1$. Then the complier groups $C_{t0}$ for $t\geq 1$ must be empty. To see this, suppose that $T_i(0)=t\geq 1$. This implies
$U_0(t)+u_{it}>U_0(0)+u_{i0}=u_{i0}$. Adding up these inequalities gives $U_1(t)+u_{it}>u_{i0}$,
and $T_i(1)$ cannot be $0$. 

All other groups $C_{tt^\prime}$ may exist. This leaves  $\acalT(\acalT-1)$ unknown group probabilities, which is $\acalT/2$ times more than the $2(\acalT-1)$ propensity scores we observe. We need $(\acalT-1)(\acalT-2)$ additional constraints to point-identify all group probabilities.

\paragraph{Single-peaked Mean Utilities}
Now suppose that mean utilities are ``single-peaked'' in the sense 
that  the function $t\to U_1(t)-U_0(t)$   is decreasing 
over $t=1,\ldots,T-1$. This would be  a reasonable assumption if  $z=1$ makes treatment $t=1$ more attractive and
the treatments $t>1$ are ordered by their proximity to $t=1$. 

 %\Bernardnote{Deleted classes here.}
 
 If this holds, then the same argument as above shows that the response groups $C_{tt^\prime}$ must be empty when $t^\prime>t\geq 1$. This eliminates $(\acalT-1)(\acalT-2)/2$ response groups;  we divided by two the number of additional identification constraints that we need.

\paragraph{Truncated Moment Bounds with Bounded Outcomes}

One alternative approach is to use truncated moment bounds, a method initiated by \citet{horowitz1995identification} and popularized by \citet{lee2009training} in the context of randomized experiments with sample attrition. See also \citet{huber2017sharp} for extensions to the IV setting with binary instruments and binary treatments, and \citet{semenova:2025} for refinements of Lee bounds.

\paragraph{Monotone Treatment Responses.}
It is sometimes natural to assume that treatments are ordered and that potential outcomes are weakly increasing in the treatment level, that is,
$Y_i(t)\ge Y_i(t')$ whenever $t\ge t'$, as in \citet{manski1997}. A weaker variant imposes monotonicity only in expectation within a response group $C$,
$\mathbb{E}[Y_i(t)\mid i\in C]\ge \mathbb{E}[Y_i(t')\mid i\in C]$ for all $t\ge t'$, as in \citet{Marx:2024}\footnote{See Section~4.2.2 of the job market paper version of \citet{Marx:2024}, available at \url{https://www.dropbox.com/scl/fi/dq5hnkwwbsjpluuonpcwf/marx-jmp.pdf?rlkey=1zpwmnx5dgtmre5urp5hdzof5&e=4&dl=0}.} .
These restrictions differ conceptually from our positive-selection assumptions. Monotone treatment response (MTR) constrains outcomes across treatment levels within a fixed group, whereas positive selection compares average outcomes across response groups at a given treatment level. Since our analysis focuses on unordered treatments, we  focus here on selection-based restrictions.

%\Bernardnote{Ordered MTR comment, do we want to mention it here?}

 \paragraph{Positive Selection}

 The binary instrument model gives a first example of the power of the positive selection defined in Section~\ref{sub:pos:sel}. Take $\tau\neq 1$ and consider the  complier groups  $C_{\tau 1}$: they all have $t=1$ when $z=1$, but they shift to it from different treatment values  $\tau$ under $z=0$.   Depending on the context, there may be a plausible reason to order   the corresponding group average outcomes when $t=1$. Suppose for instance that $T=3$, and that
\begin{equation}\label{eq:pos:sel23}
	\mathbb{E} \left[ Y_i(1)  \vert  i  \in C_{01} \right] \leq \mathbb{E} \left[ Y_i(1)  \vert  i  \in C_{21} \right].
	\end{equation}
In this $2\times 3$ model under universal targeting, there are five response-groups, as illustrated in \Cref{fig:ternary/binary}. Proposition~\ref{pro:DCM:late:binaryinstrument} shows that  the Wald estimator only identifies 
 \[
 \alpha_0 \mathbb{E} \left[ Y_i(1) -Y_i(0) \vert  i  \in C_{01} \right]+ (1-\alpha_0) \mathbb{E} \left[ Y_i(1) -Y_i(2) \vert  i  \in C_{21} \right],
	\]
	where $\alpha_0=(P(0 \vert 0)-P(0 \vert 1))/(P(1\vert 1)-P(1\vert 0))$ is point-identified.  
\Cref{prop:DCM:late:binaryinstrument:PI} shows that adding inequality~\eqref{eq:pos:sel23} yields   bounds on the corresponding LATEs.

 \begin{figure}[htb]
	 \begin{tikzpicture}[scale=0.5]  
 \draw [blue!20!white, fill=blue!20!white] (-5,0) -- (-4,0) -- (-4,-5) -- (-5,-5);
 \draw [green!20!white, fill=green!20!white] (-5,0) -- (-4,0) -- (1,5) -- (-5,5);
 \draw [blue!20!white, fill=blue!20!white] (0,0) -- (5,5) -- (5,-5) -- (0,-5);
 \draw [yellow!20!white, fill=yellow!20!white] (0,0) -- (5,5) -- (1,5) -- (-4,0);	
 \draw [red!20!white, fill=red!20!white] (0,0) -- (0,-5) -- (-4,-5) -- (-4,0);
 %\draw [blue!20!white, fill=blue!20!white, opacity=0.5] (0,0) -- (-4,0) -- (-4,-5) -- (0,-5);	

	 \draw[->,>=latex] (-5,0)  -- (5,0) node[below] {$u_{i1}-u_{i0}$};
	 \draw[->,>=latex] (0,-5)  -- (0,5) node[above] {$u_{i2}-u_{i0}$};
	 \draw[dashed,red] (0,0) -- (5,5);
				 \draw[dashed,red] (-4,0) -- (1,5);
				 \draw[dashed,red] (-4,0)  -- (-4,-5);
				 \node[red] at (-1.5,-1.5) {$C_{01}$};
 \node[red] at (-4.5,-1.5) {$C_{00}$};
 \node[red] at (-3,3) {$C_{22}$};
 
 \node[red] at (2.5,-1.5) {$C_{11}$};
 \node[red] at (1,2.5) {$C_{21}$};
		 \fill[blue] (0.0, 0.0) circle (.1cm);	          
		 \fill[blue] (-4.0,0.0) circle (.1cm);	          
	 \node[blue] at (0.4,-0.4) {$P_0$};
		 \node[blue] at (-4.2,0.4) {$P_1$};
	 \end{tikzpicture} 
	 \caption{A $2\times 3$ model  with one  targeted treatment}\label{fig:ternary/binary}
	 \end{figure}

\begin{proposition}[Positive selection and treatment effects in the $2\times 3$ model under universal targeting]\label{prop:DCM:late:binaryinstrument:PI}
If 
\begin{align}\label{assn:positivesel:ineq:new}
\mathbb{E}\left[ Y_i(1) \mid i \in C_{01} \right] \leq \mathbb{E}\left[ Y_i(1) \mid i \in C_{21} \right],
\end{align}
then the sharp identified region for the local average treatment effects for $C_{01}$ and $C_{21}$ 
\[
L_{01} := \mathbb{E}\!\left[ Y_i(1) - Y_i(0) \mid i \in C_{01} \right] 
 \; \text{ and } \;
 L_{21} :=\mathbb{E}\!\left[ Y_i(1) - Y_i(2) \mid i \in C_{21} \right]  
\]
 is the half-line in the $(L_{01},L_{21})$ plane with equation 
\begin{multline}\label{half-line-eq}
L_{01} \{ P(0 \mid 0) - P(0 \mid 1) \} + L_{21} \{ P(2 \mid 0) - P(2 \mid 1) \}  \\
=(\bar{E}_1(1) - \bar{E}_0(1)) -(\bar{E}_0(0) - \bar{E}_1(0))
-(\bar{E}_0(2) - \bar{E}_1(2))    
\end{multline}
limited by
\begin{align}\label{iden:late:partial:2by3:new}
\begin{split}
 L_{01}    &\leq    \frac{\bar{E}_1(1) - \bar{E}_0(1)}
{P(1 \mid 1) - P(1 \mid 0)} 
- 
\frac{\bar{E}_0(0) - \bar{E}_1(0)}
{P(0 \mid 0) - P(0 \mid 1)}, \\
L_{21} &\geq 
\frac{\bar{E}_1(1) - \bar{E}_0(1)}
{P(1 \mid 1) - P(1 \mid 0)} 
- 
\frac{\bar{E}_0(2) - \bar{E}_1(2)}
{P(2 \mid 0) - P(2 \mid 1)}.
\end{split}
\end{align}
If \eqref{assn:positivesel:ineq:new} holds with equality, then the two inequalities in \eqref{iden:late:partial:2by3:new} hold at equality as well, and $L_{01}$ and $L_{21}$ are point identified.
\end{proposition}

The bounds on the local average treatment effects for $C_{01}$ and $C_{21}$ can be estimated directly from sample averages and are sharp in the sense that they incorporate all model restrictions.  
Specifically, the characterization relies only on two types of constraints:  
(i) the linear equality constraint in \eqref{half-line-eq}, 
which corresponds to the special case of the second equality in \eqref{eq:DCM:late:binaryinstrument} from \Cref{pro:DCM:late:binaryinstrument}, and  
(ii) the linear inequality constraint implied by positive selection in \eqref{assn:positivesel:ineq:new}.  
Accordingly, the bounds in \eqref{iden:late:partial:2by3:new} constitute the explicit solutions to the associated linear programming problems.  
This explicit sharp characterization is possible because the response-group probabilities are point identified, as established in \Cref{pro:DCM:vector:binaryinstrument}.

%\Simonnote{The proof needs to be updated}

\subsection{The $3\times 3$ Model}\label{subsec:ternary-ternary}
Let us now 
turn to  the $3\times 3$ model of Example~\ref{ex:ternaryternary}, where 
$\calZs=\calTs = \{ 1, 2\}$ and $\calZ=\calT = \{0, 1, 2\}$.
We  assume universal targeting: for all of our results in this section,  we impose Assumptions
\ref{assn:ARUM}, \ref{assn:ref:treat}, \ref{assn:universal}, and \ref{assn:valid};
%~\ref{assn:valid} - \ref{assn:full-targeting}; 
$z=1$ targets $t=1$ and $z=2$ targets $t=2$.

 The set $A$ in \Cref{corclasses-strict-11} can be $\emptyset, \{1\}, \{2\}$, or $\{1,2\}$, with corresponding values of $t$ in $\{0\}, \{0,1\}, \{0,2\}$ or $\{0,1,2\}$ respectively. The set $c(\emptyset,0)$ corresponds to the never-takers  $C_{000}$. For $A=\{1\}$ we get $C_{010}$ and $C_{111}$, and for $A=\{2\}$ we get $C_{002}$ and $C_{222}$. Finally, with $A=\{1,2\}$ we have $C_{012}, C_{112}$, and $C_{212}$.

\begin{figure}[htb]
	\begin{tikzpicture}[scale=0.5]  
		\def\dotpone{-3.5}
\draw [blue!20!white, fill=blue!20!white] (0,-5) -- (0,-4) -- (5,1) -- (5,-5);
\draw [green!20!white, fill=green!20!white] (-5,0) -- (\dotpone,0) -- (1,5) -- (-5,5);
\draw [orange!20!white, fill=orange!20!white] (0,0) -- (5,5) -- (5,1) -- (0,-4);
\draw [yellow!20!white, fill=yellow!20!white] (0,0) -- (5,5) -- (1,5) -- (\dotpone,0);	
\draw [red!20!white, fill=red!20!white, opacity=0.5] (0,0) -- (0,-4) -- (-5,-4) -- (-5,0);
\draw [blue!20!white, fill=blue!20!white, opacity=0.5] (0,0) -- (\dotpone,0) -- (\dotpone,-5) -- (0,-5);	
\draw [green!20!white, fill=green!20!white] (-5,-4) -- (-3.5,-4) -- (-3.5,-5) -- (-5,-5);

	\draw[->,>=latex] (-5,0)  -- (5,0) node[below] {$u_{i1}-u_{i0}$};
	\draw[->,>=latex] (0,-5)  -- (0,5) node[above] {$u_{i2}-u_{i0}$};
	\draw[dashed,red] (0,0) -- (5,5);
		\draw[dashed,red] (0,-4) -- (5,1);
				\draw[dashed,red] (\dotpone,0) -- (1,5);
		\draw[dashed,red] (-5,-4)  -- (0,-4);
				\draw[dashed,red] (\dotpone,0)  -- (\dotpone,-5);
				\node[red] at (-1.5,-1.5) {$C_{012}$};
\node[red] at (-1.5,-4.5) {$C_{010}$};
\node[red] at (-4.5,-1.5) {$C_{002}$};
\node[red] at (-4.5,-4.5) {$C_{000}$};
\node[red] at (3,-4) {$C_{111}$};
\node[red] at (-3,3) {$C_{222}$};

\node[red] at (2.5,0.5) {$C_{112}$};
\node[red] at (1,2.5) {$C_{212}$};
		\fill[blue] (0.0, 0.0) circle (.1cm);	          
		\fill[blue] (\dotpone,0.0) circle (.1cm);	          
		\fill[blue] (0.0,-4.0) circle (.1cm);	      
	\node[blue] at (-0.4,0.4) {$P_0$};
		\node[blue] at (-3.7,0.4) {$P_1$};
		\node[blue] at (0.4,-4.4) {$P_2$};

	\end{tikzpicture} 
	\caption{Strictly one-to-one targeted treatment in the $3\times 3$ model}\label{fig:ternary-ternary}
	\end{figure}

These eight elemental response groups are illustrated in Figure~\ref{fig:ternary-ternary}, again with the origin in $P_0$.
    Comparing Figure~\ref{fig:ternary-ternary} with 
Figure~\ref{fig:terterum} shows the identifying power of Assumption~\ref{assn:full-targeting}.  Table~\ref{tab:map:ternary:ternary} shows which groups  take $T_i=t$ when $Z_i=z$.

\begin{table}[htbp]
    \caption{\label{tab:map:ternary:ternary} Response Groups of \Cref{ex:ternaryternary}}\centering\medskip
{\footnotesize
        \begin{tabular}{lccc}  \toprule
       & $T_i(z)=0$  & $T_i(z)=1$  & $T_i(z)=2$ \\  \midrule \\ \vspace{3mm}
$z = 0$ & $C_{000} \cup C_{010} \cup C_{002} \cup C_{012}$ 
 & $C_{111} \cup C_{112}$ & $C_{222} \cup C_{212}$ \\  \vspace{3mm}
$z = 1$ & $C_{000} \cup C_{002}$ & $C_{111} \cup C_{010} \cup C_{012} \cup C_{112} \cup C_{212}$ & $C_{222}$ \\  \vspace{3mm}
$z = 2$ & $C_{000} \cup C_{010}$ & $C_{111}$ & $C_{222} \cup C_{002} \cup C_{012} \cup C_{112} \cup C_{212}$\\  
\bottomrule \end{tabular}
}
\end{table}

Unlike the $2\times 3$ model, even under strict one-to-one targeting the $3\times 3$ model does not satisfy unordered monotonicity. One could show it with the matrix algebra in \cite{heckmanpinto-pdt}.\footnote{See \Cref{appendix:HP} for details.} It is more straightforward to note that when the instrument value changes from $z=1$ to $z=2$, observations in $C_{010}$ move to treatment value~0, while observations in $C_{002}$ leave treatment~0. This is the definition of a two-way flow, which   violates unordered monotonicity. 
Since the $3\times 3$  model has three instrument values and only two targeted treatments,   \citet[Example~4.7]{baietal:monotonicityaverage} shows that it satisfies their weaker general monotonicity assumption. As a consequence, the average potential outcomes $\mathbb{E} [Y_i(d)]$ can only be restricted by identification at infinity arguments.

%\Bernardnote{Again, I vote for deleting the blue part\ldots}
% \Simonnote{The blue part is moved to \Cref{appendix:HP}.}

\subsubsection{Identification in the $3\times 3$ Model}

We know  that 
one  restriction is missing to point-identify  the probabilities of all eight response-groups. 
The following proposition shows that the probabilities of four of the eight elemental groups are point-identified: two groups of always-takers, and two groups of compliers.  The other four probabilities are constrained by three adding-up constraints.

\begin{proposition}[Response-group probabilities in the  $3\times 3$ model under universal targeting]\label{pro:DCM:vector:ternaryternary}
The following four probabilities are identified:
$\Pr (C_{111}) = P(1 \vert 2)$, 
$\Pr (C_{222}) = P(2 \vert 1)$, 
$\Pr (C_{112}) = P(1 \vert 0)-P(1 \vert 2)$,  
and
$\Pr (C_{212}) = P(2 \vert 0)-P(2 \vert 1)$.
The remaining four response group probabilities are partially-identified and
can be parameterized as:
%              \begin{align*}
%                \Pr(C_{000}) &= p,\\
%                \Pr(C_{002}) &= P(0\vert 1) - p,\\
%                \Pr(C_{010}) &= P(0\vert 2) - p,\\
%                 \Pr(C_{012}) &= P(0\vert 0) - P(0\vert 1) - P(0\vert 2) + p, 
%              \end{align*}
%			  where the unknown $p$ satisfies
%                \[
%                    \max\{ 0, P(0\vert 1) + P(0\vert 2) - P(0\vert 0)\} \leq    p \leq \min \{ 1, P(0\vert 1),P(0\vert 2) \}.
%					\]
                $\Pr(C_{000}) = p$,
                $\Pr(C_{002}) = P(0\vert 1) - p$,
                $\Pr(C_{010}) = P(0\vert 2) - p$, and
                 $\Pr(C_{012}) = P(0\vert 0) - P(0\vert 1) - P(0\vert 2) + p$, 
			  where the unknown $p$ satisfies
                    $\max\{ 0, P(0\vert 1) + P(0\vert 2) - P(0\vert 0)\} \leq    p \leq \min \{ 1, P(0\vert 1),P(0\vert 2) \}$.
\end{proposition}

As before, the model has the following testable implications: 
	$P(1 \vert 1) \geq P(1 \vert 0) \geq  P(1 \vert 2)$,
	$P(2 \vert 2) \geq P(2 \vert 0) \geq  P(2 \vert 1)$, and
	$P(0 \vert 0) \geq \max(P(0 \vert 1), P(0 \vert 2))$.
It follows from \citet[Example 3.3]{Bai:Tabord-Meehan:2025:arXiv} that these implications are sharp.

The following proposition identifies a number of group average outcomes\footnote{Again, these could also be derived using the general framework of \citet{heckmanpinto-pdt}, even though the unordered monotonicity assumption is not satisfied---see \Cref{appendix:HP}.}. 

\begin{proposition}[Group average outcomes in the $3\times 3$ model under universal targeting]\label{pro:DCM:3by3}
The following group average outcomes are point-identified:
\begin{align*}
\mathbb{E} \left[ Y_i(0) \vert  i  \in C_{000} \cup C_{002} \right]  &= \frac{\bar{E}_1(0)}{ P(0 \vert 1)}, &
\mathbb{E} \left[ Y_i(0) \vert  i  \in C_{000} \cup C_{010} \right]  &= \frac{\bar{E}_2(0)}{ P(0 \vert 2)}, \\
\mathbb{E} \left[ Y_i(1) \vert  i  \in C_{111} \right]  &= \frac{\bar{E}_2(1)}{ P(1 \vert 2)}, &
\mathbb{E} \left[ Y_i(2) \vert  i  \in C_{222} \right]  &= \frac{\bar{E}_1(2)}{ P(2 \vert 1)}, \\
\mathbb{E} \left[ Y_i(0)  \vert  i  \in C_{010} \cup C_{012} \right] &= \frac{\bar{E}_0(0) - \bar{E}_1(0)}{P(0 \vert 0)-P(0 \vert 1)}, &
\mathbb{E} \left[ Y_i(0)  \vert  i  \in C_{002} \cup C_{012} \right] &= \frac{\bar{E}_0(0) - \bar{E}_2(0)}{P(0 \vert 0)-P(0 \vert 2)}, \\
%\end{align*}
%
%
%\begin{align*}
\mathbb{E} \left[ Y_i(1)  \vert  i  \in C_{010} \cup C_{012}  \cup C_{212} \right] &= \frac{\bar{E}_1(1) - \bar{E}_0(1)}{P(1 \vert 1)-P(1 \vert 0)}, &
%
%\mathbb{E} \left[ Y_i(1)  \vert  i  \in C_{010} \cup C_{012} \cup C_{112} \cup C_{212} \right] &= \frac{\bar{E}_1(1) - \bar{E}_2(1)}{P(1 \vert 1)-P(1 \vert 2)}, \\
\mathbb{E} \left[ Y_i(1) \vert  i  \in C_{112} \right]  &= \frac{ \bar{E}_0(1) - \bar{E}_2(1)}{P(1 \vert 0)-P(1 \vert 2)}, \\
\mathbb{E} \left[ Y_i(2)  \vert  i  \in C_{002} \cup C_{012} \cup C_{112} \right] &= \frac{\bar{E}_2(2) - \bar{E}_0(2)}{P(2 \vert 2)-P(2 \vert 0)}, &
%\mathbb{E} \left[ Y_i(2)  \vert  i  \in C_{002} \cup C_{012} \cup C_{112} \cup C_{212} \right] &= \frac{\bar{E}_2(2) - \bar{E}_1(2)}{P(2 \vert 2)-P(2 \vert 1)}, \\
\mathbb{E} \left[ Y_i(2) \vert  i  \in C_{212} \right]  &= \frac{ \bar{E}_0(2) - \bar{E}_1(2)}{P(2 \vert 0)-P(2 \vert 1)}.
\end{align*}
\end{proposition}

By itself,  \Cref{pro:DCM:3by3}  does not  allow us to identify  an average treatment effect for {\em any\/} 
(even composite) response-group.  Suppose for instance that we want to identify $\mathbb{E} (Y_i(1)-Y_i(0)\vert i\in C)$ for some group $C$. Then $C$ needs to exclude $C_{111}$, $C_{112}$, and $C_{212}$, since $\mathbb{E} (Y_i(0)\vert i\in C^\prime)$ is not identified for any group $C^\prime$ that contains  $C_{111}$, $C_{112}$, or $C_{212}$. Since we only know the mean outcome of treatment~1 for groups that contain one of these three elemental groups, the conclusion follows.

\subsubsection{Using Positive Selection}

Note that  if we assumed $\mathbb{E} (Y_i(1)\vert i\in C_{112})=\mathbb{E} (Y_i(1)\vert i\in C_{212})$, then we could 
combine the two equations in the fourth displayed line of \Cref{pro:DCM:3by3} and the probabilities of $C_{112}$ and $C_{212}$ (which are point-identified by \Cref{pro:DCM:vector:ternaryternary})
%subtract the sixth equation of~\Cref{pro:DCM:3by3} from the fifth (after reweighting) 
to obtain $\mathbb{E} (Y_i(1)\vert i \in C_{010}\cup C_{012})$. This would  point-identify the average effect of treatment~1 vs treatment~0 on this composite complier group $C_{01\ast}$. While this assumption may be overly strong, it seems natural to impose that $Y_i(\tau)$ is on average larger in a response group that has $t=\tau$ for more values of $z$. Assumption~\ref{assn:positivesel:3by3} formalizes this intuition in our setting.

\begin{assumption}[Positive selection in the $3\times 3$ model]\label{assn:positivesel:3by3}
	Either or both of the following assumptions hold:
	\begin{align}
		\mathbb{E} \left[ Y_i(1) \vert  i  \in C_{112} \right] 
		&\geq \mathbb{E} \left[ Y_i(1) \vert  i  \in C_{212} \right],
		\label{assn:positivesel:3by3:1}\\
		\mathbb{E} \left[ Y_i(2) \vert  i  \in C_{212} \right] 
		&\geq \mathbb{E} \left[ Y_i(2) \vert  i  \in C_{112} \right].
		\label{assn:positivesel:3by3:2}
	\end{align}
\end{assumption}
\Cref{assn:positivesel:3by3} states a form of positive selection into treatment, as defined in Section~\ref{sub:pos:sel}. Consider \Cref{assn:positivesel:3by3:1} for instance. It says that within the group of ``$12$-compliers'' $C_{\ast 12}=C_{012}\cup C_{112}\cup C_{212}$, those observations with $T(0)=1$ have a larger average counterfactual  $Y(1)$ than those with $T(0)=2$.  \Cref{corr:DCM:3by3} shows that this gives  bounds on the local average treatment effects for $C_{01\ast}$-compliers, with a similar result for  \Cref{assn:positivesel:3by3:2} and $C_{0\ast 2}$-compliers.

\begin{proposition}[Identifying treatment effects in the $3\times 3$ model]\label{corr:DCM:3by3} %\Bernardnote{Changes checked and approved!x}
\begin{enumerate}
	\item Under \eqref{assn:positivesel:3by3:1}, the local average treatment effect 
	\[
	\mathbb{E} \left[ Y_i(1) - Y_i(0)  \vert  i  \in C_{01\ast}  \right] 
	\]
	is at least as large as 
\[
\frac{(\bar{E}_1(1) -\bar{E}_0(1))-(\bar{E}_0(0) -\bar{E}_1(0))}{P(0 \vert 0)-P(0 \vert 1)}
-\frac{\bar{E}_0(1) -\bar{E}_2(1)}{P(1 \vert 0)-P(1 \vert 2)}
\frac{P(2 \vert 0)-P(2 \vert 1)}{P(0 \vert 0)-P(0 \vert 1)}. 
\]
\item 
Under \eqref{assn:positivesel:3by3:2}, the local average treatment effect 
\[
\mathbb{E} \left[ Y_i(2) - Y_i(0)  \vert  i  \in C_{0\ast 2}  \right] 
\]
is at least as large as 
\[
	\frac{(\bar{E}_2(2) -\bar{E}_0(2))-(\bar{E}_0(0) -\bar{E}_2(0))}{P(0 \vert 0)-P(0 \vert 2)}
-\frac{\bar{E}_0(2) -\bar{E}_1(2)}{P(2 \vert 0)-P(2 \vert 1)}
\frac{P(1 \vert 0)-P(1 \vert 2)}{P(0 \vert 0)-P(0 \vert 2)}. 
\]
\item  In both 1 and 2, ``at least as large'' can be replaced with ``equals'' if the corresponding inequality in \Cref{assn:positivesel:3by3} is an equality.
\end{enumerate}
\end{proposition}

The lower bounds given in \Cref{corr:DCM:3by3} are sharp relative to the restrictions imposed by the mean-independence of the instruments and the positive selection assumptions. Just as with  the bounds derived in \eqref{iden:late:partial:2by3:new}, this explicit sharp characterization obtains  because the composite response-group probabilities
\begin{align*}
    \Pr(i\in C_{01\ast})  &=  \Pr(i\in C_{010}\cup C_{012}) = P(0 \mid 0) - P(0 \mid 1), \\
    \Pr(i\in C_{0\ast 2}) &= \Pr(i\in C_{002}\cup C_{012}) = P(0 \mid 0) - P(0 \mid 2), 
\end{align*}
are point-identified, as established in \Cref{pro:DCM:vector:ternaryternary}.

% To interpret the  conditions in \Cref{assn:positivesel:3by3}, 
% consider a hypothetical program  to encourage college attendance.
%  Let $z=1$ be a tuition subsidy that can only be used for a  STEM curriculum, and $z=2$ a tuition subsidy for a  non-STEM curriculum.
% The treatments are:   not going to college ($t=0$), studying STEM in college ($t=1$), and  opting for a non-STEM college 
%  curriculum ($t=2$); the outcome $Y$ is later earnings. 
%  The response groups $C_{112}$ and $C_{212}$ are ``college always-takers''  who choose the major for which they receive a subsidy. They only differ in the major they would choose in the absence of a subsidy; that difference may be negligible relative to what separates them from the other groups of  compliers ($C_{0**}$), who would not go to college without a subsidy. 
%  The homogeneity assumption~\eqref{assn:positivesel:3by3} formalizes  this intution. It allows us to identify three different local average treatment effects (LATEs):
%  \begin{itemize}
% 	 \item $\mathbb{E} \left[ Y_i(1) - Y_i(0)  \vert  i  \in C_{010} \cup C_{012} \right]$ is the return to  a STEM major for ``STEM-major compliers'' ($C_{010} \cup C_{012}$);
% 	 \item  $\mathbb{E} \left[ Y_i(2) - Y_i(0)  \vert  i  \in C_{002} \cup C_{012} \right]$ is the return to non-STEM major for ``non-STEM-major compliers'' ($C_{002} \cup C_{012}$); 
% 	 \item 
% 	 and
% 	 $\mathbb{E} \left[ Y_i(1) - Y_i(2) \vert  i  \in C_{112} \cup  C_{212} \right]$, that is, the difference in earnings between STEM and non-STEM majors for ``college always-takers''.
%  \end{itemize}

\subsubsection{When is Positive Selection Plausible?}\label{sec:pos:sel:3by3:ARUM}
%\Bernardnote{I changed this to go directly to the Gaussian example with the $a_0,a_1,a_2$ and modified appendix H accordingly.}

Let us  focus on  
\eqref{assn:positivesel:3by3:2}.
Given strict one-to-one targeting, $C_{112}$ is defined by
\[
\underline{U}(1)-\bar{U}(2)\leq u_{i2}-u_{i1}\leq \underline{U}(1)-\underline{U}(2),
\; u_{i1}-u_{i0}\geq -\underline{U}(1).
\]
$C_{212}$ is defined by
\[
\underline{U}(1)-\underline{U}(2)\leq u_{i2}-u_{i1}\leq \bar{U}(1)-\underline{U}(2),
\; u_{i2}-u_{i0}\geq -\underline{U}(2).
\]
 To simplify notation, define  $\zeta_i=u_{i2}-u_{i1}$ and $\xi_i=u_{i2}-u_{i0}$, so that $u_{i1}-u_{i0}=\xi_i-\zeta_i$. The inequalities above can be rewritten as 
\begin{itemize}
    \item for $C_{112}$: 
    $\underline{U}(1)-\bar{U}(2)\leq \zeta_i\leq\underline{U}(1)-\underline{U}(2)$, 
    $\xi_i - \zeta_i \geq - \underline{U}(1)$;
    \item for $C_{212}$: $\underline{U}(1)-\underline{U}(2)\leq \zeta_i\leq \bar{U}(1)-\underline{U}(2)$,
    $\xi_i\geq -\underline{U}(2)$.
\end{itemize}

Figure~\ref{fig:ternaryroy} plots these two groups on the $\zeta_i \times \xi_i$ plane.
Group $C_{212}$ corresponds to the top-right (infinite) rectangle and 
group $C_{112}$ is partitioned into the two subgroups: 
$C_{112}^{(i)}$ is a bottom-left triangle 
and $C_{112}^{(ii)}$ is a top-left (infinite) rectangle.

 \begin{figure}[htb]
	 \begin{tikzpicture}[scale=0.5]  
 \draw [green!20!white, fill=green!20!white] (0,5) -- (0,0) -- (5,0) -- (5,5);
 % \draw [blue!20!white, fill=blue!20!white] (0,0) -- (5,5) -- (5,-5) -- (0,-5);
\draw [yellow!20!white, fill=yellow!20!white] (-5,0) -- (0,0) -- (-5,-5);	
 \draw [red!20!white, fill=red!20!white] (-5,5) -- (-5,0) -- (0,0) -- (0,5);
 %\draw [blue!20!white, fill=blue!20!white, opacity=0.5] (0,0) -- (-4,0) -- (-4,-5) -- (0,-5);	
 
	 \draw (-5,-5)  -- (-5,5);	 
      \draw (5,-5)  -- (5,5);
	 \draw[->,>=latex] (-5,0)  -- (6,0) node[below] {$\zeta_i$};
	 \draw[->,>=latex] (0,-5)  -- (0,5) node[above] {$\xi_i$};
	 \draw[dashed,yellow] (-5,-5) -- (0,0);
		% 		 \draw[dashed,red] (-4,0) -- (1,5);
		% 		 \draw[dashed,red] (-4,0)  -- (-4,-5);
				 % \node[red] at (-1.5,-1.5) {$C_{01}$};
 % \node[red] at (-4.5,-1.5) {$C_{00}$};
 % \node[red] at (-3,3) {$C_{22}$};
 
 % \node[red] at (2.5,-1.5) {$C_{11}$};
 \node[red] at (2,2.5) {$C_{212}$};
  \node[red] at (-3,-2) {$C_{112}^{(i)}$};
   \node[red] at (-3,3) {$C_{112}^{(ii)}$};
		 \fill[blue] (0.0, 0.0) circle (.1cm);	          
		 \fill[blue] (-5.0,0.0) circle (.1cm);	  
         	 \fill[blue] (5.0,0.0) circle (.1cm);	
	 \node[blue] at (-5.4,-0.4) {$P_a$};
		 \node[blue] at (0.4,0.4) {$P_b$};
         		 \node[blue] at (4.6,-0.4) {$P_c$};
	 \end{tikzpicture} 
	 \caption{Positive Selection in the Generalized $3\times3$ Roy model}\label{fig:ternaryroy}
\parbox{4in}{\footnotesize{Notes: 
$P_a=(\underline{U}(1)-\bar{U}(2),-\underline{U}(2))$,
$P_b=(\underline{U}(1)-\underline{U}(2),-\underline{U}(2))$, 
and
$P_c=(\bar{U}(1)-\underline{U}(2),-\underline{U}(2))$.}}
\end{figure}

Suppose for instance that 
\[
\mathbb{E}(Y_i(2)\mid u_{i0}, u_{i1},u_{i2})-\mathbb{E}(Y_i(2))=a_0 u_{i0}+a_1 u_{i1}+a_2 u_{i2},
\]
where $a_0$, $a_1$, and $a_2$ are constants, and $(u_{i0},u_{i1},u_{i2})$ are jointly normal and mutually uncorrelated with mean 0 and variance 1. 
Defining $\zeta_i = u_{i2} - u_{i1}$ and $\xi_i = u_{i2} - u_{i0}$, it is easy to see that
\begin{align}\label{eq:expect:y2}
  \mathbb{E}(Y_i(2)\mid \zeta_i,\xi_i)-\mathbb{E}(Y_i(2))  
&= \frac{a_2 + a_0 - 2 a_1}{3} \zeta_i + \frac{a_2 + a_1 -2 a_0}{3} \xi_i.
\end{align}
If $a_2 \geq \max(a_0,a_1)+\vert a_1-a_0\vert$, 
then the coefficients of $\zeta_i$ and $\xi_i$ in \Cref{eq:expect:y2} are both non-negative. It follows\footnote{See Appendix~\ref{appx:normal:positive_selection} for details.} from the geometry of Figure~\ref{fig:ternaryroy} that $\mathbb{E}(Y_i(2)\mid i\in C_{212}) \geq \mathbb{E}(Y_i(2)\mid i\in C_{112})$.
In summary, a sufficiently large value of $a_2$ induces positive selection, generating patterns similar to those of comparative advantage in generalized Roy models.

\subsection{What do the IV Estimators Identify in the $3 \times 3$ Model?}\label{sub:klm}

\citet[hereafter KLM]{kirkeboen2016}   used a $3\times 3$ model to study  the impact    of the  field of study on later earnings.
 Their Proposition~2 characterizes  what   two-stage least squares (TSLS) estimators 
identify under  different sets of assumptions.
The least stringent version combines a monotonicity assumption (Assumption~4 in KLM) and condition~(iii) in their Proposition~2, which they call 
``irrelevance and information on next-best alternatives''.  ``Irrelevance'' is a set of exclusion restrictions, while ``information on next-best alternatives'' assumes the availability of additional data.

\subsubsection{Monotonicity and Irrelevance}
While we take quite a different path,  our universal targeting assumption turns out to yield exactly the same identifying restrictions as the combination of monotonicity 
 and  irrelevance in KLM. We show it in~\Cref{appendix:KLM}.
 
This set of assumptions in itself is too weak to give  two-stage least squares estimates  a simple interpretation. 
 To see this, let $\beta_1$ and $\beta_2$ be the probability limits of the coefficients in a regression of $Y_i$ on the indicator variables $\uniset(T_i=1)$ and $\uniset(T_i=2)$, with instruments $Z_i$.  Remember from Table~\ref{tab:map:ternary:ternary} that under strict one-to-one targeting, five response-groups have $T(1)=1$:
 \begin{enumerate}
	\item the always-takers $C_{111}$;
	\item the ``intermediate'' group $C_{112}$, which has $T(z)=1$ unless $z=2$;
	\item the three groups $C_{010}, C_{012}$, and $C_{212}$, which have $T(z)=1$ if and only if $z=1$.
 \end{enumerate}
A similar distinction applies to the groups that have $T(2)=2$; it motivates Definition~\ref{def:12compliers}. 

\begin{definition}[1-compliers and 2-compliers]\label{def:12compliers}
	We call
	\[
		\mathcal{C}_1 = 	C_{010}\cup C_{012} \cup C_{212}, 
		\]
		the  1-compliers group and 
		\[
		\mathcal{C}_2 = 	C_{002}\cup C_{012}\cup C_{112} 
		\]
		the   2-compliers group.
	\end{definition}
The  $\beta_1$ and $\beta_2$ coefficients turn out to be  weighted averages of the LATEs on these two groups and on the intermediate groups $C_{112}$ and $C_{212}$.

\begin{proposition}[TSLS in the 3 $\times$ 3 model under universal targeting]\label{prop:IV:3by3}
	The parameters $\beta_1$ and $\beta_2$ satisfy 
	\begin{multline*}
	\begin{pmatrix}
	\Pr(i\in  \mathcal{C}_1) &  - \Pr(i\in C_{212}) \\
	-\Pr(i\in C_{112}) & \Pr(i\in  \mathcal{C}_2)
	\end{pmatrix}
	\begin{pmatrix}
	\beta_1 \\
	\beta_2
	\end{pmatrix}
	\\
	=  
	\begin{pmatrix}
	\mathbb{E} [ \{ Y_i(1)-Y_i(0)  \} \uniset( i\in \mathcal{C}_1) ] 
	- \mathbb{E} [\{  Y_i(2)-Y_i(0) \} \uniset( i\in C_{212})] 
	\\
	\mathbb{E} [\{  Y_i(2)-Y_i(0) \} \uniset( i\in \mathcal{C}_2)]
	- \mathbb{E} [ \{ Y_i(1)-Y_i(0)  \} \uniset( i\in C_{112}) ]
	\end{pmatrix}.
	\end{multline*}
	\end{proposition}

	\Cref{prop:IV:3by3} implies that $\beta_1$ and $\beta_2$ are weighted averages of 
	the   four  local average treatment effects on the right-hand side of this system of two equations.
	 The  weights are  functions of the four probabilities on the left-hand side,  which are point identified by \Cref{pro:DCM:vector:ternaryternary}. However, these weights  
	may be  positive or negative. This  complicates interpretation further\footnote{\cite{MTW:2020}  give a set of assumptions under which the weights are positive in a model with multiple binary instruments.}.

 \subsubsection{Additional Assumptions}

 \paragraph{Next-best alternatives}
 	Using the additional information on next-best alternatives in KLM
	amounts, in our notation, to dropping  the ``intermediate'' response-groups  $C_{212}$ and $C_{112}$ from the data. Then the system of equations in \Cref{prop:IV:3by3}  becomes diagonal and it yields
	\begin{align*}
	\beta_1 &= \mathbb{E} [  Y_i(1)-Y_i(0)   | i\in \mathcal{C}_1], \\
	\beta_2 &= \mathbb{E} [ Y_i(2)-Y_i(0)  | i\in \mathcal{C}_2],
	\end{align*}
	where now $\mathcal{C}_1$ reduces to $C_{010}\cup C_{012}$ and $\mathcal{C}_2$ reduces to $C_{002}\cup C_{012}$.
	This is exactly Proposition~2~(iii) of KLM.
 Alternatively, one may simply assume that the response-groups  $C_{212}$ and $C_{112}$ are empty. This is the path taken by \citet{Bhuller22}\footnote{See their Corollary~5 and Table~1 for details.}. 

\paragraph{Positive Selection}
Additional information of the type used by KLM often is not available.  Moreover, assuming away $C_{112}$ and $C_{212}$ seems rather strong. 
On the other hand, reasonable assumptions can be used to generate bounds on the local average treatment effects for 1-compliers and 2-compliers. \Cref{cor:IV:3by3} illustrates this.

\begin{proposition}[TSLS in the $3 \times 3$ model revisited]\label{cor:IV:3by3}
 Assume that 
 \begin{align}\label{rank-condition-3by3}
	\mathcal{D}\equiv \Pr(i\in  \mathcal{C}_1 )\Pr(i\in  \mathcal{C}_2) 
	- \Pr(i\in C_{212}) \Pr(i\in C_{112}) \neq 0.
	\end{align}
Let 
\[
	\mathcal{D}_1\equiv \mathbb{E} \left(
   Y_i(1)-Y_i(0)\vert i\in\mathcal{C}_1
   \right)
-
\mathbb{E} \left(
   Y_i(1)-Y_i(0)\vert i\in C_{112}
   \right)
   \]
   and 
\[
\mathcal{D}_2\equiv \mathbb{E} \left(
	   Y_i(2)-Y_i(0)\vert i\in \mathcal{C}_2
	   \right)
   - \mathbb{E} \left(
	   Y_i(2)-Y_i(0)\vert i\in C_{212}
   \right).
   \]
If $\mathcal{D}_1 \mathcal{D}_2>0$, then $\beta_1-\mathbb{E} \left(
		Y_i(1)-Y_i(0)\vert i\in\mathcal{C}_1
		\right)$ 
		and $\beta_2-\mathbb{E} \left(
			Y_i(2)-Y_i(0)\vert i\in\mathcal{C}_2
			\right)$ have the sign of $\mathcal{D}$.
\end{proposition}

Note that the KLM result of the previous paragraph is the limit case  where $\mathcal{D}_1=\mathcal{D}_2=0$.

%   \large \bf Try the inequality version?} 
% \begin{align}\label{homogenous-effect-3by3}
% \begin{split}
%  \mathbb{E} [  Y_i(1)-Y_i(0)   | i\in C_{112}]
% = \mathbb{E} [  Y_i(1)-Y_i(0)  | i\in  \mathcal{C}_1], \\
%  \mathbb{E} [ Y_i(2)-Y_i(0)  | i\in C_{212}]
% = \mathbb{E} [  Y_i(2)-Y_i(0)  | i\in \mathcal{C}_2].
% \end{split}
% \end{align}
% %
%  Then, if 

% the two-stage least squares estimators	 $\beta_1$ and $\beta_2$ satisfy
% \begin{align*}
% \beta_1 &= \mathbb{E} [  Y_i(1)-Y_i(0)   | i\in C_{010}\cup C_{012}\cup C_{212} \cup C_{112} ], \\
% \beta_2 &= \mathbb{E} [ Y_i(2)-Y_i(0)  | i\in C_{002}\cup C_{012}\cup C_{112} \cup C_{212}].
% \end{align*}
% \end{corollary}

The regularity condition~\eqref{rank-condition-3by3}  ensures that 
the $2 \times 2$ matrix that premultiplies $(\beta_1, \beta_2)^\prime$ in 
\Cref{prop:IV:3by3}
is invertible\footnote{
It holds  if $C_{212}$ and $C_{112}$ have positive probability and either $C_{010}\cup C_{012}$ or $C_{002}\cup C_{012}$ has positive probability.}. 
To interpret the assumptions on signs, suppose that $\mathcal{D}_1$ is positive.
	Since $\mathcal{C}_1=C_{010}\cup C_{012}\cup C_{212}$, the positivity of $\mathcal{D}_1$ states that the average effect of treatment~1 on $C_{010}\cup C_{012}\cup C_{212}$ is at least as large as on $C_{112}$. This is a form of positive selection that is in the same spirit as (but different from) \Cref{assn:positivesel:3by3}. If this  form of positive selection holds for both  treatments, then the TSLS estimates overestimate the LATEs on the corresponding compliers if  $\mathcal{D}>0$, and they underestimate them if $\mathcal{D}<0$.

To summarize,  the TSLS estimators in the $3 \times 3$ model are difficult to interpret 
unless  additional information is available and/or some additional assumptions are imposed.
If the groups $C_{112}$ and $C_{212}$ are indeed empty, then both the TSLS estimators and those we obtained in \Cref{corr:DCM:3by3} should  identify the LATEs on the 1- and 2-compliers. Comparing their values 
is a useful (if informal) way of testing the assumptions and of exploring further the heterogeneity of the treatment effects.

% \Bernardnote{Done this far.}

\section{Empirical Example: The Head Start Impact Study}\label{sec:empirical} 
We now reexamine the \citeapos{kline2016} analysis of  the Head Start Impact Study (HSIS)
using our framework. 
 We use  exactly the same data as they did; we only 
apply different identifying assumptions\footnote{\citet{Kamat:2019}, \cite{KamatNorrisPecenco2023} and \cite{AngristSantosTecchio2025} analyze the HSIS using different identifying approaches.}.

Head Start is a federal program in the US that addresses various factors affecting children's development in low-income families. It provides early childhood education (hereafter ``preschool'') and health and nutrition services.  HSIS  was a longitudinal study conducted from 2002 to 2010 to assess the program's impact on  cognitive, social-emotional, and health outcomes. It focused on 84 communities where the demand for Head Start services was  larger than the supply. HSIS randomly assigned about   5,000 three and four year old preschool children to either a treatment group which was offered  Head Start services, or  a control group which received no such offer. Children in either group could also attend other preschool centers if offered a slot.

% We use  exactly the same set of covariates as they did;  we only apply different identifying assumptions.}

The structure of HSIS is identical to that of \Cref{ex:binaryinstrument0}: it is a $2\times 3$ model.
The treatments here  consist of no preschool ($n$), Head Start ($h$), and other preschool centers ($c$): $\calT=\{n, h, c\}$. 
The instrument is binary, with a control group  ($z=0$) and a group that is offered admission to Head Start ($z=1$). 
   The outcome variable is test scores, measured in standard deviations from their mean.

In our notation, $\mathcal{Z}=\{0,1\}$ and $\mathcal{Z}^\ast=\{1\}$; the instrument $z=1$ uniquely and strictly targets the Head Start treatment $h$. In the terminology of this paper, this assignment structure satisfies universal targeting.
The set of inequalities in \eqref{eq:test:targeting1} is empty, as $z=0$ is the only non-targeting instrument value. On the other hand, \eqref{eq:test:targeting2} yields the following sharp testable implications:
\[
P(n\vert 1) < P(n\vert 0) \quad \text{and} \quad P(c\vert 1) < P(c\vert 0).
\]
In words, the Head Start offer must strictly reduce the choice probabilities of the other two alternatives (no preschool $n$ and other centers $c$). The data is fully consistent with these implications, as evidenced by the proportions of the two complier groups reported in Panel~A of \Cref{HS:cm-te} (see also \citet[Table 3]{kline2016}).

%    In our notation, $\calZ=\{0,1\}$ and $\calZs=\{1\}$, which uniquely and strictly targets treatment value $h$. In the terminology of this paper, treatment assignment satisfies universal targeting
%  \footnote{\citet{Kamat:2019} analyzes HSIS using  a different approach that focuses on how   the choice sets available to a child vary with the value of the instrument.}.   The set of inequalities  \Cref{eq:test:targeting1} are empty since~$z=0$ is the only non-targeting instrument value. On the other hand, \Cref{eq:test:targeting2} yields the testable implications
%  \[
% P(n\vert 1) < P(n\vert 0) \text{ and }
% P(c\vert 1) < P(c\vert 0):
%  \]
%   Head Start offers   reduce the probabilities of the  other two treatments. The data comfortably fails to reject these implications.
%  }
%  \Bernardnote{Added this.}

Figure~\ref{fig:kw:2by3} reproduces Figure~\ref{fig:ternary/binary} in this setting.  In addition to the three always-taker groups
$C_{nn}$, $C_{cc}$, and $C_{hh}$, there are two complier groups: 
$C_{nh}$,
and
$C_{ch}$.
In Sections~\ref{sub:kw:replic} and~\ref{sub:kw:subst}, we focus on the LATEs on the two complier groups $C_{nh}$ and $C_{ch}$. Section~\ref{sub:kw:rationed} embeds the model into a larger, $3\times 3$ model in order to evaluate the marginal value of the public funds used in Head Start.

\begin{figure}[htb]
	\begin{tikzpicture}[scale=0.5]  
\draw [blue!20!white, fill=blue!20!white] (-5,0) -- (-4,0) -- (-4,-5) -- (-5,-5);
\draw [green!20!white, fill=green!20!white] (-5,0) -- (-4,0) -- (1,5) -- (-5,5);
\draw [blue!20!white, fill=blue!20!white] (0,0) -- (5,5) -- (5,-5) -- (0,-5);
\draw [yellow!20!white, fill=yellow!20!white] (0,0) -- (5,5) -- (1,5) -- (-4,0);	
\draw [red!20!white, fill=red!20!white] (0,0) -- (0,-5) -- (-4,-5) -- (-4,0);
%\draw [blue!20!white, fill=blue!20!white, opacity=0.5] (0,0) -- (-4,0) -- (-4,-5) -- (0,-5);	

	\draw[->,>=latex] (-5,0)  -- (5,0) node[below] {$u_{ih}-u_{in}$};
	\draw[->,>=latex] (0,-5)  -- (0,5) node[above] {$u_{ic}-u_{in}$};
	\draw[dashed,red] (0,0) -- (5,5);
				\draw[dashed,red] (-4,0) -- (1,5);
				\draw[dashed,red] (-4,0)  -- (-4,-5);
				\node[red] at (-1.5,-1.5) {$C_{nh}$};
\node[red] at (-4.5,-1.5) {$C_{nn}$};
\node[red] at (-3,3) {$C_{cc}$};

\node[red] at (2.5,-1.5) {$C_{hh}$};
\node[red] at (1,2.5) {$C_{ch}$};
		\fill[blue] (0.0, 0.0) circle (.1cm);	          
		\fill[blue] (-4.0,0.0) circle (.1cm);	          
	\node[blue] at (0.4,-0.4) {$P_n$};
		\node[blue] at (-4.2,0.4) {$P_h$};
	\end{tikzpicture} 
	\caption{The \citet{kline2016} Model of Preschool Choice}\label{fig:kw:2by3}
	\end{figure}

 \subsection{Group proportions and counterfactual means}\label{sub:kw:replic}
 
Our estimates of the proportions  of the two complier groups in the sample use \eqref{prob:vectors:ternary/binary} in \Cref{pro:DCM:vector:binaryinstrument}; they are shown in Panel~A of~\Cref{HS:cm-te}. As expected, they coincide 
 with those in~\citet{kline2016}.

Panel B of \Cref{HS:cm-te} shows the counterfactual means of test scores for the complier groups, as per~\Cref{pro:DCM:late:binaryinstrument}.
% Among those that are point-identified, the average test scores are the highest for the groups  who always choose other preschool centers 
% (about 0.3 standard deviation). 
While $\mathbb{E}[ Y_i(n) | i \in C_{nh}]$ 
is negative,  $\mathbb{E}[ Y_i(c) | i \in C_{ch}]$ is above $0.1$ standard deviation---not a negligible value in this field.
This  suggests that some of the children who enter Head Start 
would have been  at a good preschool otherwise.  \citet{kline2016} call this pattern the ``substitution effect'' of Head Start.  However, 
they do not report estimates of $\mathbb{E}[ Y_i(n) | i \in C_{nh}]$  and $\mathbb{E}[ Y_i(c) | i \in C_{ch}]$.

%\Bernardnote{How is it new? it is a different path, but isn't it the same thing in the end?}

\begin{table}[tbh]
{\small
\caption{\label{HS:cm-te} Proportions, Counterfactual Means and Treatment Effects by Response Groups}\medskip
\begin{center}
\begin{tabular}{lccc} \hline \hline
 & 3-year-olds  & 4-year-olds  & Pooled  \\  \hline 
 \multicolumn{4}{l}{Panel A. Proportions of Response Groups via \Cref{pro:DCM:vector:binaryinstrument}} \\ 
 & & & \\
Compliers from $n$ to $h$  ($C_{nh}$) &     0.505 &     0.393 &     0.454 \\  
Compliers from $c$ to $h$ ($C_{ch}$) &     0.198 &     0.272 &     0.232 \\  \hline
 \multicolumn{4}{l}{Panel B. Counterfactual Means of Test Scores via \Cref{pro:DCM:late:binaryinstrument}} \\ 
 & & & \\
$\mathbb{E}[ Y_i(n) | i  \in C_{nh}]$   &    -0.027 &    -0.116 &    -0.062 \\  
$\mathbb{E}[ Y_i(c) | i  \in C_{ch}]$   &     0.112 &     0.144 &     0.129 \\  \hline 
%  \multicolumn{4}{l}{Panel C. Counterfactual Means of Test Scores via \Cref{corr:DCM:late:binaryinstrument}} \\ 
%  & & & \\ 
% $\mathbb{E}[ Y_i(h)  | i  \in C_{nh}]=\mathbb{E}[ Y_i(h)  | i  \in C_{ch}]$ &     0.252 &     0.169 &     0.216 \\     \hline
 \multicolumn{4}{l}{Panel C. Treatment Effects via \Cref{prop:DCM:late:binaryinstrument:PI}} \\
 & & & \\ 
Upper Bound on $\mathbb{E}[ Y_i(h) - Y_i(n)  | i  \in C_{nh} ]$   &     0.279 &     0.285 &     0.278 \\  
 & (0.063) &     (0.076) &     (0.050) \\ 
 Lower Bound on $\mathbb{E}[ Y_i(h) - Y_i(c)  | i  \in C_{ch} ]$  &     0.140 &     0.025 &     0.087 \\ 
 &     (0.089) &     (0.097) &     (0.063) \\ 
Upper Bound on &     0.139 &     0.260 &     0.191 \\ 
$\mathbb{E}[ Y_i(h) - Y_i(n)  | i  \in C_{nh}] - \mathbb{E}[ Y_i(h) - Y_i(c)  | i  \in C_{ch}]$&     (0.098) &     (0.115) &     (0.071) \\
\hline \end{tabular}
\end{center}
\parbox{5.5in}{Notes: Head Start ($h$), other centers ($c$), no preschool ($n$). 
Standard errors in parentheses are clustered at the Head Start center level.}
}
\end{table}

\subsection{Treatment Effects}\label{sub:kw:subst}

To fully measure the  substitution effect, one needs to identify 
$\mathbb{E} \left[ Y_i(h)  |  i  \in C_{nh} \right]$ and  
$\mathbb{E} \left[ Y_i(h)  |  i  \in C_{ch} \right]$.
However,  we know from  \Cref{pro:DCM:late:binaryinstrument} that
they are only partially identified by 
\begin{multline*}
\alpha_0 \mathbb{E} \left[ Y_i(h)  |  i  \in C_{ch} \right]+(1-\alpha_0) \mathbb{E} \left[ Y_i(h)  |  i  \in C_{nh} \right]  
=\frac{\mathbb{E} \left[ Y_i \uniset(T_i = h)  | Z_i = 1 \right] - \mathbb{E} \left[ Y_i \uniset(T_i = h)  | Z_i = 0 \right]}{P(h\vert 1)-P(h\vert 0)}.
\end{multline*}
where $\alpha_0=(P(c\vert 0) - P(c\vert 1))/(P(h\vert 1)-P(h\vert 0))$.
This is exactly the formula on \citet[pp.1811]{kline2016}: as they  point out, the LATE for Head Start is a weighted average of ``subLATEs'' with weights determined by the proportion of $C_{ch}$ among compliers, which is identified from the data\footnote{Our $\alpha_0$ is denoted $S_c$ in their paper.}.

\citet{kline2016} first tried to identify $\mathbb{E}[ Y_i(h) - Y_i(c) | i \in C_{ch}]$ and 
$\mathbb{E}[ Y_i(h) - Y_i(n) | i \in C_{nh}]$ separately using 
interactions of the instrument  with covariates or experimental sites.
 They acknowledged the limitations of this approach
and resorted  to a parametric selection model \`{a} la \citet{Heckman1979}  instead. 
%Their \citet[Table VIII, column (4) full model]{kline2016} reports 
  They report\footnote{See \citet[Table VIII,  column (4), full model]{kline2016}.}  estimates of the  local average treatment effects of $0.370$ for   $C_{nh}$ and 
$-0.093$ for $C_{ch}$, with respective standard errors $0.088$ and $0.154$.
   The resulting point estimate of the difference is quite large, at $0.463$ standard deviation.

% \todo[inline]{Above estimates: for 3 yr old, 4 yr, pooled?}
 
 Our \Cref{prop:DCM:late:binaryinstrument:PI} provides an alternative approach to separating the two treatment effects. 
Given that compliers coming from  other preschools ($C_{ch}$) had better test scores than compliers not originally in preschools ($C_{nh}$), it seems reasonable to assume that  they also have better test scores under Head Start:
\begin{equation}
	\mathbb{E} \left[ Y_i(h)  |  i  \in C_{ch} \right]  \geq \mathbb{E} \left[ Y_i(h)  |  i  \in C_{nh} \right].\label{eq:hsis:positive}
\end{equation}
This is a ``positive selection'' that fits within the framework of \Cref{prop:DCM:late:binaryinstrument:PI}.  It can be derived in a simple model in which better students benefit more from Head Start and other preschools; and students choose schools as a function of their expected outcome. We show in Appendix~\ref{appx:pos:sel:binary} that this model generates positive selection under reasonable assumptions.
The pooled cohort estimates in Panel~C of \Cref{HS:cm-te} indicate that the upper bound on $\mathbb{E}[Y_i(h)-Y_i(n)\mid i\in C_{nh}]$ is $0.28$, while the lower bound on $\mathbb{E}[Y_i(h)-Y_i(c)\mid i\in C_{ch}]$ is $0.09$.\footnote{Under monotone treatment response, $Y_i(h)\ge Y_i(n)$, the lower bound on $\mathbb{E}[Y_i(h)-Y_i(n)\mid i\in C_{nh}]$ would be zero, yielding a 95\% confidence  interval $0\le \mathbb{E}[Y_i(h)-Y_i(n)\mid i\in C_{nh}]\le 0.36$.}
The difference between these two numbers gives an upper bound of $0.19$ for the difference of these two LATEs,
with a standard error of $0.07$. Negative selection (reversing the inequality~\eqref{eq:hsis:positive}) would make $0.19$ a  {\em lower\/} bound  for the difference of the LATEs. At the same time, it would imply that the lower bound is negative; this is soundly rejected by the data.

%\medskip

%{\em A proof, just between us:} going through the proof of Corollary~\Cref{corr:DCM:late:binaryinstrument}, a sort of Corollary 2bis says that if $\mathbb{E}(Y(1)\vert C_{01})\leq \mathbb{E}(Y(1)\vert C_{21})$, then
%\[
%	\mathbb{E}(Y(1)-Y(0)\vert C_{01})	\geq 
%	\frac{\bar{E}_1(1)-\bar{E}_0(1)}{P(1\vert 1)-P(1\vert 0)}
%	-\frac{\bar{E}_0(0)-\bar{E}_1(0)}{P(0\vert 0)-P(0\vert 1)}
%\]
%and 
%\[
%	\mathbb{E}(Y(1)-Y(2)\vert C_{21})	\leq 
%	\frac{\bar{E}_1(1)-\bar{E}_0(1)}{P(1\vert 1)-P(1\vert 0)}
%	-\frac{\bar{E}_0(2)-\bar{E}_1(2)}{P(2\vert 0)-P(2\vert 1)}.
%\]
%That's exactly the reverse of the inequalities in~\Cref{iden:late:partial:2by3}; therefore instead of having an upper bound of $0.19$, we have a lower bound of $0.19$ under negative selection, with the same standard error. Since negative selection implies a negative lower bound, we can reject it very soundly.

%\medskip

Our upper bound of $0.19$   is much lower than the point estimate reported 
by \cite{kline2016}. In fact, 
our 95\% and 99\% one-sided confidence intervals for 
\[
	\mathbb{E}[ Y_i(h) - Y_i(n)  | i  \in C_{nh}] - \mathbb{E}[ Y_i(h) - Y_i(c)  | i  \in C_{ch}]
	\]
are $(-\infty, 0.308)$ and $(-\infty, 0.356)$. We conclude that the  $0.463$ estimate in \citet{kline2016}
may overstate the difference between the two complier groups: it  can only be rationalized under  negative  selection, which is a much less plausible assumption.

\subsection{Expanding Access to Head Start}\label{sub:kw:rationed}

\citet{kline2016} sought to evaluate the welfare effect of increasing the number of slots in Head Start, as summarized by the marginal value of public funds (MVPF). They  note that  any expansion of Head Start  will  vacate some slots at competing preschools, which are oversubscribed. The relaxation of this rationing must be counted as an  effect of Head Start expansions.  This is what they call ``rationed substitutes''\footnote{See  \citet[Sections V.D and IX.A]{kline2016} for details.}.

The  children who move from $T_i=n$ to $T_i=c$ when a slot is vacated by a child who moves to Head Start  constitute a $C_{nc}$ group that is ruled out by the $2\times 3$ model. These children increase their grades by $Y_i(c)-Y_i(n)$, whose average generates a LATE that we denote $\textrm{LATE}_{nc}$.  Equation~(9) in \citet[p. 1816]{kline2016} shows that the value  of $\textrm{LATE}_{nc}$ is a crucial input in the computation of the MVPF of a Head Start expansion.  Identifying it  requires either data on offers to all preschools, which \citet{kline2016} do not have\footnote{See footnote 19 in their paper.}, or additional modeling assumptions. They used their parametric selection model to construct an estimate for $\textrm{LATE}_{nc}$. Their estimate of  $\textrm{LATE}_{nc} = 0.294$ results in a high MVPF estimate of $2.02$ (see Table IX in their paper).

We take a different approach by embedding the $2\times 3$ model within  a $3 \times 3$ model. In this richer model, the instrument can take three values: in addition to 
the control group  ($z=0$) and those offered admission to Head Start ($z=1$), we have a new group that we denote $z=2$.  This group receives a direct offer of admission to a competing preschool. Economically, such offers arise when a Head Start expansion vacates slots at competing preschools---what \citet{kline2016} call ``rationed substitutes''---but formally $z=2$ is simply an additional instrument value requiring only that it satisfies \Cref{assn:valid}. Note that this maintains strict, one-to-one targeting. Since potential outcomes $Y_i(t)$  only depend on child $i$'s own treatment $t$, SUTVA still holds: we condition on the market equilibrium in which $z=2$ offers exist, consistent with the partial equilibrium framework of \citet{kline2016}.

Figure~\ref{fig:kw:3by3}  shows the resulting treatment assignment, using tildes to denote the complier groups of the $3\times 3$ model\footnote{Again, it is just Figure~\ref{fig:ternary-ternary} with different notation.}. Using this notation, $\textrm{LATE}_{nc}$ can be written as
\begin{align*}
\textrm{LATE}_{nc} = \mathbb{E}[ Y_i(c) - Y_i(n)  | i  \in \tilde{C}_{n*c}],
\end{align*}
where $\tilde{C}_{n*c} = \tilde{C}_{nnc} \cup \tilde{C}_{nhc}$ is the composite group of $n\to c$ compliers.
 Comparing Figure~\ref{fig:kw:3by3} with Figure~\ref{fig:kw:2by3} shows that  the other complier groups of the two models are linked by 
\begin{align*}
	C_{nh} &=   \tilde{C}_{nh*} = \tilde{C}_{nhn}\cup  \tilde{C}_{nhc}\\
	C_{ch} &=   \tilde{C}_{ch*} = \tilde{C}_{chc}.
\end{align*}

\begin{figure}[htb]
	\begin{tikzpicture}[scale=0.5]  
		\def\dotpone{-3.5}
\draw [blue!20!white, fill=blue!20!white] (0,-5) -- (0,-4) -- (5,1) -- (5,-5);
\draw [green!20!white, fill=green!20!white] (-5,0) -- (\dotpone,0) -- (1,5) -- (-5,5);
\draw [orange!20!white, fill=orange!20!white] (0,0) -- (5,5) -- (5,1) -- (0,-4);
\draw [yellow!20!white, fill=yellow!20!white] (0,0) -- (5,5) -- (1,5) -- (\dotpone,0);	
\draw [red!20!white, fill=red!20!white, opacity=0.5] (0,0) -- (0,-4) -- (-5,-4) -- (-5,0);
\draw [blue!20!white, fill=blue!20!white, opacity=0.5] (0,0) -- (\dotpone,0) -- (\dotpone,-5) -- (0,-5);	
\draw [green!20!white, fill=green!20!white] (-5,-4) -- (-3.5,-4) -- (-3.5,-5) -- (-5,-5);

	\draw[->,>=latex] (-5,0)  -- (5,0) node[below] {$u_{ih}-u_{in}$};
	\draw[->,>=latex] (0,-5)  -- (0,5) node[above] {$u_{ic}-u_{in}$};
	\draw[dashed,red] (0,0) -- (5,5);
		\draw[dashed,red] (0,-4) -- (5,1);
				\draw[dashed,red] (\dotpone,0) -- (1,5);
		\draw[dashed,red] (-5,-4)  -- (0,-4);
				\draw[dashed,red] (\dotpone,0)  -- (\dotpone,-5);
				\node[red] at (-1.5,-1.5) {$\tilde{C}_{nhc}$};
\node[red] at (-1.5,-4.5) {$\tilde{C}_{nhn}$};
\node[red] at (-4.5,-1.5) {$\tilde{C}_{nnc}$};
\node[red] at (-4.5,-4.5) {$\tilde{C}_{nnn}$};
\node[red] at (3,-4) {$\tilde{C}_{hhh}$};
\node[red] at (-3,3) {$\tilde{C}_{ccc}$};

\node[red] at (2.5,0.5) {$\tilde{C}_{hhc}$};
\node[red] at (1,2.5) {$\tilde{C}_{chc}$};
		\fill[blue] (0.0, 0.0) circle (.1cm);	          
		\fill[blue] (\dotpone,0.0) circle (.1cm);	          
		\fill[blue] (0.0,-4.0) circle (.1cm);	      
	\node[blue] at (-0.4,0.4) {$P_n$};
		\node[blue] at (-3.7,0.4) {$P_h$};
		\node[blue] at (0.4,-4.4) {$P_c$};

	\end{tikzpicture} 
	\caption{Embedding Preschool Choice in a $3\times 3$ Model}\label{fig:kw:3by3}
	\end{figure}

	Now consider the new group of $n\to c$ compliers. It differs from $\tilde{C}_{chc}$ in that its members will not go to a preschool  unless they are offered a slot. We show in Appendix~\ref{appx:pos:sel:ternary} that the  structural model that we used in the binary instrument case  predicts the following  inequality:
	\begin{equation}\label{eq:possel:3x3first}
		\mathbb{E}[ Y_i(c)  | i  \in \tilde{C}_{n*c}] \leq \mathbb{E}[ Y_i(c)   | i  \in \tilde{C}_{chc}].
		\end{equation}

 Now consider the composite response-groups $\tilde{C}_{n*c} = \tilde{C}_{nnc} \cup \tilde{C}_{nhc}$ and $\tilde{C}_{nn\ast}=\tilde{C}_{nnc}\cup \tilde{C}_{nnn}$. 
As Figure~\ref{fig:kw:3by3} shows, they
only differ by the substitution of $\tilde{C}_{nnn}$ for $\tilde{C}_{nhc}$. The former never go to Head Start or to another preschool, while the latter are full compliers.   Our structural model generates the inequality
	\[
		\mathbb{E}[ Y_i(n)  | i  \in \tilde{C}_{nnn}] \leq \mathbb{E}[ Y_i(n)  | i  \in \tilde{C}_{nhc}]
	\]
 which implies
	\begin{equation}\label{eq:possel:3x3second}
		\mathbb{E}[ Y_i(n)  | i  \in \tilde{C}_{nn*}] \leq \mathbb{E}[ Y_i(n)  | i  \in \tilde{C}_{n*c}].
	\end{equation}
Inequalities~\eqref{eq:possel:3x3first} and~\eqref{eq:possel:3x3second} again are  ``positive selection'' assumptions that fall under our  \Cref{corr:DCM:3by3}.

Since $\tilde{C}_{nn\ast}$ coincides with $C_{nn}$ and $\tilde{C}_{chc}$ is $C_{ch}$, we already know the values of the right-hand sides of both inequalities. Applying the same logic as in \Cref{corr:DCM:3by3} gives us an upper bound for $\textrm{LATE}_{nc}$: 
	\begin{align*}
	\textrm{LATE}_{nc} 
	%= \mathbb{E}[ Y_i(c) - Y_i(n)  | i  \in C_{n*c}^{(3 \times 3)}]
	\leq \mathbb{E}[ Y_i(c)   | i  \in \tilde{C}_{chc}] - \mathbb{E}[ Y_i(n)   | i  \in \tilde{C}_{nn*}]
	= \mathbb{E}[ Y_i(c)   | i  \in C_{ch}] - \mathbb{E}[ Y_i(n)   | i  \in C_{nn}].
	\end{align*}

	As the $\textrm{MVPF}$ is an increasing function of $\textrm{LATE}_{nc}$, this gives us in turn an upper bound on its value\footnote{Online Appendix~\ref{appx:mvpf} derives the formula for the MVPF in this model.}. We obtain
	$\textrm{LATE}_{nc} \leq 0.164$ and $\textrm{MVPF} \leq 1.55$. These upper bounds are noticeably smaller than the point estimates  that result from the parametric selection model of \citet{kline2016}.

\section*{Concluding Remarks}

We   have shown that the idea of targeting is a useful way to analyze models with multivalued treatments and multivalued instruments. The testable implications that we derived suggest  a natural three-step procedure to elucidate targeting patterns  in the data, supplementing any qualitative information provided by the empirical context.

\paragraph{Step 1: Screening Candidates.}
For each instrument-target pair $(z,t)\in\mathcal{Z}\times\mathcal{T}$, test the hypothesis $H_{0,zt}: S_{zt}\le 0$ against $H_{1,zt}: S_{zt}>0$, where
\[
S_{z t} = P(0\mid z)+\max_{z''\in\mathcal{Z}, z^{\prime\prime}\neq z} P(t\mid z'')-1.
\]
 Let $\widehat{\mathcal{V}}$ denote the set of surviving pairs:
\[
\widehat{\mathcal V} \equiv \{(z,t)\in\mathcal{Z}\times\mathcal{T}: H_{0,zt}\ \text{is not rejected}\}.
\]
If for a fixed $z$, no pair $(z,t)$ survives in $\widehat{\mathcal V}$, we reject the hypothesis that $z$ targets any treatment in $\mathcal{T}$.

\paragraph{Step 2: Checking Compatibility}
For every pair of  candidates $(z,t)$ and $(z',t')$ in $\widehat{\mathcal V}$ with $t\neq t'$, we test  $H_{0,c}: S_{c} \le 0$, where
\[
S_{c} = P(t'\mid z)+P(t\mid z')-1.
\]
A rejection implies that $(z,t)$ and $(z',t')$ cannot hold simultaneously.

%\paragraph{Step 3:} if after steps 1 and 2 targeting patterns suggest universal targeting, one can further use the inequalities in~\eqref{eq:test:targeting3}  to check whether the data rejects it. 

\paragraph{Step 3: Testing for Universal Targeting}
The set of valid mappings that survive Steps 1 and 2 may or may not be consistent with universal targeting. One can further use the inequalities in \eqref{eq:test:targeting3} to formally test it.

As usual,
one should be careful about controlling the size of this three-step procedure by applying standard multiple-testing corrections within each step and/or bootstrapping. We leave this for further research.

Our paper only analyzed discrete-valued instruments and treatments. Some of the notions we used would extend naturally to continuous instruments and treatments: the definitions of targeting,  one-to-one targeting, and positive selection would translate directly. Strict targeting, on the other hand, is less appealing in a context in which continuous values may denote intensities. Our earlier paper \citep{LeeSalanie2018} as well as \citeapos{Mountjoy:2022} can be seen as analyzing continuous-instruments/discrete-treatments models. 
Extending our analysis to  models with continuous treatments is an obvious topic for further research.
It would  also be interesting to apply the partial identification approach of \citet{MST:2018} in our setting.
Finally, there has been  a surge of recent interest on understanding  the properties of OLS and 2SLS estimands when treatment effects vary with the covariates  \citep{BBMT,Sloc22,GHK:24}. 
%Finally, \citet{Navjeevan:Pinto:Santos:2023} developed doubly robust moment conditions for a class of potential outcomes models with covariates.
%\Bernardnote{I am not sure this NPS2023 is relevant here.}
We believe that the targeting concept   and the identifying assumptions explored in this paper  should be relevant in this context and that  they merit further investigation.

%\newpage

\appendix

%{\huge \bf Appendices}

\section{Proofs for \Cref{sec:targeting}\label{appx:proofs:sec2}}

\begin{proof}[Proof of \Cref{proprespvectors1}]
As is explained in the main text, it is a direct consequence of \Cref{prop:csq-target} under one-to-one targeting. 
\end{proof}

\begin{proof}[Proof of \Cref{corclasses-strict-11}]
(i): Suppose that $z$ is not a targeting instrument. Then $T_i(z)=\underline{t}_i$ (which does not depend on $z$ under strict targeting). If $z$ is targeting but $z\not\in A_i$, then $T_i(z)\neq z$;  since $z$ targets $z$ uniquely, $t^\ast_i(z)=z$ and $T_i(z)$  can only be $\underline{t}_i$.

(ii): If $\underline{t}_i$ is in $\calTs$, we must have $\underline{V}_i=\underline{U}(\underline{t}_i)+u_{i\underline{t}_i}< \bar{U}(\underline{t}_i)+u_{i\underline{t}_i}=V^\ast_i(\underline{t}_i)$; therefore $T_i(\underline{t}_i)=\underline{t}_i$ and $\underline{t}_i$ is in $A_i$.

(iii): We already know that each observation $i$ is in a $C(A,\tau)$ set. Then $T_i(t)=t$ over $A$ and $T_i(z)=\tau$ for all other $z$. This defines an elemental response-group. If $i$ belongs to both $C(A,\tau)$ and $C(A^\prime,\tau^\prime)$, then $A$ must equal $A^\prime$ since they are both defined as the subset of $\calTs$ where $T_i(t)=t$. Moreover, $T_i(0)=\tau=\tau^\prime$ since $0\not\in\calTs$. Therefore $\tau=\tau^\prime$ and the various $C(A,t)$ form a partition of the set of observations. 
\end{proof}

\section{Proofs for \Cref{sec:identif}\label{appx:proofs}}

\begin{proof}[Proof of \Cref{lemma:ezt-identity}]
%For the sake of completeness, we provide the proof. 
We start from the definition:
\begin{align*}
  \bar{E}_z(t) &= \mathbb{E}(Y_i\uniset(T_i=t) \mid Z_i=z)    \\
  &= \sum_{C} \mathbb{E}(Y_i\uniset(T_i=t) \mid i\in C, Z_i=z) \Pr(i\in C \mid Z_i=z) \\
  &= \sum_{C} \mathbb{E}(Y_i(t) \mid T_i=t, i\in C, Z_i=z) \Pr(T_i=t, i\in C \mid Z_i=z) \\
  &= \sum_{C} \mathbb{E}(Y_i(t) \mid T_i(z)=t, i\in C, Z_i=z) \Pr(T_i(z)=t, i\in C \mid Z_i=z),
  \end{align*}
  where the sum runs over all complier groups and the third equality uses Assumption~\ref{assn:valid}(i).
  
Now if $C=C_R$ and $R(z)\neq t$, $\Pr(T_i(z)=t,i\in C\mid Z_i=z)$ is zero; and if $R(z)=t$, then we can drop the ``$T_i(z)=t$'' terms:
\[
\mathbb{E}(Y_i(t) \mid T_i(z)=t, i\in C, Z_i=z)=
\mathbb{E}(Y_i(t) \mid  i\in C, Z_i=z),
\]
which equals $\mathbb{E}(Y_i(t) \mid  i\in C)$ by Assumption~\ref{assn:valid}(ii);
 and
\[
\Pr(T_i(z)=t, i\in C \mid Z_i=z) = \Pr(i\in C \mid Z_i=z),
\]
which equals $\Pr(i\in C)$ by Assumption~\ref{assn:valid}(ii).
Therefore we can rewrite the sum as
\[
 \bar{E}_z(t)=\sum_{C_R \; \vert \; R(z)=t} \mathbb{E}(Y_i(t) \; \vert i\in C_R)\Pr(i\in C_R).
\]
This proves the second result in the lemma.

The first result in the lemma follows the same logic by observing that $\Pr(T_i=t\mid Z_i=z) = \mathbb{E}[\uniset(T_i=t)\mid Z_i=z]$, which is equivalent to the second result where $Y_i=1$ (and thus $Y_i(t)=1$).
\end{proof}

\begin{proof}[Proof of \Cref{prop:ident-ufstrict11}]
Take $z\in\calZ$ and $t\in\calT$.  Consider any  observation $i$ and the corresponding   $A_i\subset \calTs$ and $\underline{t}_i  \in A_i \cup (\calTnotTs)$.  There are only two ways to obtain $T_i(z)=t$:
	\begin{itemize}
		\item if  $z\not\in A_i$, then   $T_i(z)=\underline{t}_i$; therefore  $\underline{t}_i=t$. Summing over all subsets $A$ of $\calTs$ that exclude $z$   gives the first term of~\eqref{eq:identuf:strict11}.
		\item if $z\in A_i$ (which implies  $z\in\calTs$), we know that $T_i(z)=z$ no matter what the value of $\underline{t}_i$ is; hence $t$ must equal $z$. Summing over all subsets $A$ that include $z$ and all values of $\underline{t}_i\in A \cup (\calTnotTs)$  gives the second line  in~\eqref{eq:identuf:strict11}.
	\end{itemize} 
	By construction, each $C(A,t)$ completely defines the mapping from instrument values 
	 to treatment values; therefore each $C(A,t)$ is an elemental group. 
	 Their union is clearly the set of all observations. 
	 If $i\in C(A,t)\cup C(A^\prime,t^\prime)$, then $A^\prime=A=A_i$ by the definition of $A_i$, and $t^\prime=t=\underline{t}_i$. Therefore the $C(A,t)$ partition the set of observations.
\end{proof}

\begin{proof}[Proof of \Cref{pro:DCM:vector:binaryinstrument}]
The proof is in the text, with the exception of $\Pr(C_{11})=P(1\vert 0)$ which follows from the fact that the probabilities add up to~1. 
		\end{proof}

\begin{proof}[Proof of \Cref{pro:DCM:late:binaryinstrument}]
Lemma~\ref{lemma:ezt-identity} gives $2\abs{\calT}$ equations:
	\begin{align}\label{eq:eztbinary}
	\begin{split}
		\bar{E}_0(1) &=  \mathbb{E}(Y_i(1)\vert i\in C(\{1\}, 1))\Pr(i\in C(\{1\},1))=\mathbb{E}(Y_i(1)\vert i\in C_{11})\Pr(C_{11}) \\
		\text{ for } t\neq 1, \; \bar{E}_0(t) &=  
		\mathbb{E}(Y_i(t)\vert i\in C(\emptyset,t))\Pr(i\in C(\emptyset,t)) \\
		&\;\;\;+\mathbb{E}(Y_i(t)\vert i\in C(\{1\},t))\Pr(i\in C(\{1\},t)) \\
		& =\mathbb{E}(Y_i(t)\vert i\in C_{tt})\Pr(C_{tt})
		+\mathbb{E}(Y_i(t)\vert i\in C_{t1})\Pr(C_{t1}) \\
		\bar{E}_1(1) &=  
		\sum_{\tau\in \calT} \mathbb{E}(Y_i(1)\vert i\in C(\{1\},\tau))\Pr(i\in C(\{1\},\tau))
		 \\
		&= \sum_{\tau\in \calT} \mathbb{E}(Y_i(1)\vert i\in C_{\tau 1})\Pr(C_{\tau 1}) \\
		\text{ for } t\neq 1, \;\bar{E}_1(t) &=  \mathbb{E}(Y_i(t)\vert i\in  C(\emptyset,t))\Pr(i\in  C(\emptyset,t))
		= \mathbb{E}(Y_i(t)\vert i\in C_{tt})\Pr(C_{tt}). 
	\end{split}	
	\end{align}   
Since \Cref{pro:DCM:vector:binaryinstrument} identifies all type probabilities,
 the first and fourth equations in~\eqref{eq:eztbinary} give directly 
 $\mathbb{E}(Y_i(t)\vert i\in C_{tt})$ for all $t$. Then the second equation
  identifies $\mathbb{E}(Y_i(t)\vert i\in C_{t1})$ for $t\neq 1$.

By subtraction, we obtain
\begin{align*}
	& (\bar{E}_1(1)-\bar{E}_0(1))-\sum_{t\neq 1}(\bar{E}_0(t)-\bar{E}_1(t))\\
	&= \sum_{t\neq 1}
	\mathbb{E} \left[ Y_i(1) -Y_i(t)  \vert  i  \in C_{t1} \right] \Pr( i \in C_{t1}).
\end{align*}
Combining  these results with \Cref{pro:DCM:vector:binaryinstrument} and~\Cref{lemma:ezt-identity} yields the formula in the Proposition.   The denominator 
\[
	\sum_{t\neq 1} (P(t\vert 0) -P(t \vert 1))=P(1\vert 1)-P(1\vert 0)
	\]
	is positive, since all terms in the sum are positive. It follows that all $\alpha_t$ weights are positive and sum to~1.
\end{proof}

\begin{proof}[Proof of \Cref{prop:DCM:late:binaryinstrument:PI}]
First, we consider the baseline counterfactuals. By \Cref{pro:DCM:late:binaryinstrument}, the expectations of $Y_i(0)$ and $Y_i(2)$ for the complier groups are point identified as:
\begin{align*}
    \mathbb{E}[ Y_i(0) \mid i \in C_{01} ] &= \frac{\bar{E}_0(0) - \bar{E}_1(0)}{P(0 \mid 0) - P(0 \mid 1)}, \\
    \mathbb{E}[ Y_i(2) \mid i \in C_{21} ] &= \frac{\bar{E}_0(2) - \bar{E}_1(2)}{P(2 \mid 0) - P(2 \mid 1)}.
\end{align*}
Next, we consider the treated counterfactuals. Recall from~\eqref{eq:eztbinary} that when $T=3$,
\[
    \bar{E}_1(1)-\bar{E}_0(1)
    =
    \mathbb{E} [ Y_i(1) \mid i \in C_{01} ] \Pr( i \in C_{01})
    +
    \mathbb{E} [ Y_i(1) \mid i \in C_{21} ] \Pr( i \in C_{21}).
\]
This implies that the weighted sum of $L_{01}$ and $L_{21}$ satisfies the equality given in \eqref{half-line-eq}.
Furthermore, under assumption \eqref{assn:positivesel:ineq:new}, we have
\begin{align*}
    &\mathbb{E} [ Y_i(1) \mid i \in C_{01} ] \{ \Pr( i \in C_{01}) + \Pr( i \in C_{21}) \} \\
    &\quad \leq \bar{E}_1(1)-\bar{E}_0(1) \\
    &\quad \leq \mathbb{E} [ Y_i(1) \mid i \in C_{21} ] \{ \Pr( i \in C_{01}) + \Pr(i \in C_{21}) \}.
\end{align*}
Substituting the point-identified expressions for $\mathbb{E}[ Y_i(0) \mid i \in C_{01} ]$ and $\mathbb{E}[ Y_i(2) \mid i \in C_{21} ]$ derived above, and using the identity $\Pr(i \in C_{01}) + \Pr(i \in C_{21}) = P(1 \mid 1)-P(1 \mid 0)$, yields the inequalities in \eqref{iden:late:partial:2by3:new}. The point identification results under the special case when \eqref{assn:positivesel:ineq:new} holds with equality follow immediately.
\end{proof}

%\Bernardnote{Proofs fixed up to here.}

%	\begin{proof}[Proof of \Cref{pro:DCM:vector:nocontrol:tb}]
%	It is straightforward from Figure~\ref{fig:ternary/binary:nocontrol}.
%	\end{proof}

\begin{proof}[Proof of \Cref{pro:DCM:vector:ternaryternary}]

By \Cref{tab:map:ternary:ternary}, the treatment probabilities decompose as follows.
For treatment~$1$: $P(1 \vert 2) = \Pr(C_{111})$ and $P(1 \vert 0) = \Pr(C_{111}) + \Pr(C_{112})$.
The first equation gives $\Pr(C_{111}) = P(1\vert 2)$, and subtracting yields $\Pr(C_{112}) = P(1\vert 0) - P(1\vert 2)$.
For treatment~$2$: $P(2 \vert 1) = \Pr(C_{222})$ and $P(2 \vert 0) = \Pr(C_{222}) + \Pr(C_{212})$.
The first equation gives $\Pr(C_{222}) = P(2\vert 1)$, and subtracting yields $\Pr(C_{212}) = P(2\vert 0) - P(2\vert 1)$.

For the remaining groups, set $p = \Pr(C_{000})$.
From the table: $P(0\vert 1) = \Pr(C_{000}) + \Pr(C_{002}) = p + \Pr(C_{002})$, so that $\Pr(C_{002}) = P(0\vert 1) - p$.
Similarly, $P(0\vert 2) = \Pr(C_{000}) + \Pr(C_{010}) = p + \Pr(C_{010})$, so that $\Pr(C_{010}) = P(0\vert 2) - p$.
Finally, $P(0\vert 0) = \Pr(C_{000}) + \Pr(C_{010}) + \Pr(C_{002}) + \Pr(C_{012})$, which gives $\Pr(C_{012}) = P(0\vert 0) - P(0\vert 1) - P(0\vert 2) + p$.

Non-negativity of $\Pr(C_{002})$, $\Pr(C_{010})$, $\Pr(C_{012})$, and $\Pr(C_{000}) = p$ imposes:
\[
  \max\{0,\, P(0\vert 1) + P(0\vert 2) - P(0\vert 0)\} \;\leq\; p \;\leq\; \min\{P(0\vert 1),\, P(0\vert 2)\}.
\]
Every value of $p$ in this interval is attainable since the eight probabilities sum to~1.

\end{proof}

\begin{proof}[Proof of \Cref{pro:DCM:3by3}]
By \Cref{lemma:ezt-identity}, we obtain 
\begin{align*}
\bar{E}_1(0) &= \mathbb{E} \left[ Y_i(0) \vert  i  \in C_{000} \cup C_{002}  \right] \Pr( i \in C_{000} \cup C_{002}),  \\
\bar{E}_2(0) &=  \mathbb{E} \left[ Y_i(0) \vert  i  \in C_{000} \cup C_{010} \right] \Pr( i \in C_{000} \cup C_{010}),  \\
\bar{E}_2(1) &= \mathbb{E} \left[ Y_i(1) \vert  i  \in C_{111} \right] \Pr( i \in C_{111}),  \\
\bar{E}_1(2) &=  \mathbb{E} \left[ Y_i(2) \vert  i  \in C_{222} \right] \Pr( i \in C_{222}),  \\
\bar{E}_0(0) - \bar{E}_1(0) 
&= \mathbb{E} \left[ Y_i(0) \vert  i  \in C_{010} \cup C_{012} \right] \Pr( i \in C_{010} \cup C_{012}), \\
\bar{E}_0(0) - \bar{E}_2(0) 
&= \mathbb{E} \left[ Y_i(0) \vert  i  \in C_{002} \cup C_{012} \right] \Pr( i \in C_{002} \cup C_{012}), \\
\bar{E}_1(1) - \bar{E}_0(1) 
&= \mathbb{E} \left[ Y_i(1) \vert  i  \in C_{010} \cup C_{012} \cup C_{212} \right] \Pr( i \in C_{010} \cup C_{012} \cup C_{212}), \\
%
%\bar{E}_1(1) - \bar{E}_2(1) 
%&= \mathbb{E} \left[ Y_i(1) \vert  i  \in C_{010} \cup C_{012} \cup C_{112} \cup C_{212} \right] \Pr( i \in C_{010} \cup C_{012} \cup C_{112} \cup C_{212}), \\
%
 \bar{E}_0(1) - \bar{E}_2(1) &= \mathbb{E} \left[ Y_i(1) \vert  i  \in C_{112} \right] \Pr( i \in  C_{112} ), \\
\bar{E}_2(2) - \bar{E}_0(2) 
&= \mathbb{E} \left[ Y_i(2) \vert  i  \in C_{002} \cup C_{012} \cup C_{112} \right] \Pr( i \in C_{002} \cup C_{012} \cup C_{112}), \\
%
%\bar{E}_2(2) - \bar{E}_1(2) 
%&= \mathbb{E} \left[ Y_i(2) \vert  i  \in C_{002} \cup C_{012} \cup C_{112} \cup C_{212} \right] \Pr( i \in C_{002} \cup C_{012} \cup C_{112} \cup C_{212}), \\
 %
 \bar{E}_0(2) - \bar{E}_1(2) &= \mathbb{E} \left[ Y_i(2) \vert  i  \in C_{212} \right] \Pr( i \in  C_{212} ).
 \end{align*}
Then, the results follow from the fact that all group probabilities are identified. 
\end{proof}

\begin{proof}[Proof of \Cref{corr:DCM:3by3}]
First note that $C_{01\ast}=C_{010} \cup C_{012}$. Under \eqref{assn:positivesel:3by3:1}, we have
\begin{align*}
	\mathbb{E} (Y_i(1)\uniset(i\in C_{010} \cup C_{012})) &= 
	\mathbb{E} (Y_i(1)\uniset(i\in C_{010} \cup C_{012} \cup C_{212}))
	- \mathbb{E} (Y_i(1)\uniset(i\in C_{212}))\\
	&=\mathbb{E} (Y_i(1)\uniset(i\in C_{010} \cup C_{012} \cup C_{212}))
	- \mathbb{E} (Y_i(1)\vert i\in C_{212})\Pr(i\in C_{212})\\
	&\geq \mathbb{E} (Y_i(1)\vert i\in C_{010} \cup C_{012} \cup C_{212})\times \Pr(i\in C_{010} \cup C_{012} \cup C_{212})\\
	&- \mathbb{E} (Y_i(1)\vert i \in C_{112}) \Pr(i\in C_{212}).
\end{align*}
Replacing the probabilities and conditional expectations with their values from   \Cref{pro:DCM:vector:ternaryternary} and \Cref{pro:DCM:3by3}, we obtain $\Pr(i\in C_{010}\cup C_{012})=P(0\vert 0)-P(0\vert 1)$ and 
\[
	\mathbb{E} (Y_i(1)\uniset(i\in C_{010} \cup C_{012}))  \geq 
\bar{E}_1(1)-\bar{E}_0(1) -\frac{\bar{E}_0(1)-\bar{E}_2(1)}{P(1\vert 0)-P(1\vert 2)}(P(2\vert 0)-P(2\vert 1)).
	\]
	Finally, writing
	\[
		\mathbb{E} (Y_i(1)-Y_i(0)\vert i\in C_{010} \cup C_{012})	
		= \frac{\mathbb{E} (Y_i(1)\uniset(i\in C_{010} \cup C_{012}))}{\Pr(i\in C_{010}\cup C_{012})}-
		\frac{\bar{E}_0(0)-\bar{E}_1(0)}{P(0\vert 0)-P(0\vert 1)}
	\]
	gives the result.

	The proof under \eqref{assn:positivesel:3by3:2} is similar: we start from  $C_{0\ast 2}=C_{002} \cup C_{012}$. Under \eqref{assn:positivesel:3by3:2}, we have
	\begin{align*}
		\mathbb{E} (Y_i(2)\uniset(i\in C_{002} \cup C_{012})) &= 
		\mathbb{E} (Y_i(2)\uniset(i\in C_{002} \cup C_{012} \cup C_{112}))
		- \mathbb{E} (Y_i(2)\uniset(i\in C_{112}))\\
		&=\mathbb{E} (Y_i(2)\uniset(i\in C_{002} \cup C_{012} \cup C_{112}))
		- \mathbb{E} (Y_i(2)\vert i\in C_{112})\Pr(i\in C_{112})\\
		&\geq \mathbb{E} (Y_i(2)\vert i\in C_{002} \cup C_{012} \cup C_{112})\times \Pr(i\in C_{002} \cup C_{012} \cup C_{112})\\
		&- \mathbb{E} (Y_i(2)\vert i \in C_{212}) \Pr(i\in C_{112}).
	\end{align*}
	Replacing the probabilities and conditional expectations with their values from   \Cref{pro:DCM:vector:ternaryternary} and \Cref{pro:DCM:3by3}, we obtain $\Pr(i\in C_{002}\cup C_{012})=P(0\vert 0)-P(0\vert 2)$ and 
	\[
		\mathbb{E} (Y_i(2)\uniset(i\in C_{002} \cup C_{012}))  \geq 
	\bar{E}_2(2)-\bar{E}_0(2) -\frac{\bar{E}_0(2)-\bar{E}_1(2)}{P(2\vert 0)-P(2\vert 1)}(P(1\vert 0)-P(1\vert 2)).
		\]
		Finally, writing
		\[
			\mathbb{E} (Y_i(2)-Y_i(0)\vert i\in C_{002} \cup C_{012})	
			= \frac{\mathbb{E} (Y_i(2)\uniset(i\in C_{002} \cup C_{012}))}{\Pr(i\in C_{002}\cup C_{012})}-
			\frac{\bar{E}_0(0)-\bar{E}_2(0)}{P(0\vert 0)-P(0\vert 2)}
		\]
		gives the result.
\end{proof}

\section{Positive Selection in Head Start}\label{appx:pos:sel}
%\Simonnote{Simon checks}
%\Bernardnote{I generalized the model a bit to allow for a fixed student effect. It also gives a non-degenerate positive selection inequality in the ternary case. }\

Let realized grades be given by 
\[
Y_i(t)= f_i + k_t + m_i p_t +\zeta_{it},
\]
where $m_i>0$, $p_h > 0$, $p_c > 0$, and $p_n = 0$.
These conditions imply that children with a larger $m_i$ benefit more from preschool (either Head Start or other centers); however, $m_i$ does not play a role if $i$ goes to neither type of preschool.
The $\zeta_{it}$ shocks are zero mean and idiosyncratic; we suppose that each subject $i$ expects $\mathbb{E}_i (Y_i(t))= f_i + k_t +m_i p_t$.
Preference shocks depend positively on expected grades: 
\[
u_{it}= a_i+b_i \mathbb{E}_i (Y_i(t))+\varepsilon_{it}
    =a_i+ b_if_i + b_i k_t + b_i m_i p_t+\varepsilon_{it},
\]
 with $b_i>0$.

 Let us define $v_{it}=u_{it}-u_{in}$; $\eta_{it}=\varepsilon_{it}-\varepsilon_{in}$; and $d_t=k_t-k_n$  for $t=c,h$.
With this specification, we have 
\[
v_{it}=b_i d_t + b_i m_i p_t+\eta_{it}.
\]
We assume that $b_i$, $m_i$, and the random vectors $(\eta_{ic}, \eta_{ih})$ and $(\zeta_{in},\zeta_{ic},\zeta_{ih})$ are mutually independent; and that $f_i$ is independent of $(\eta_{ic}, \eta_{ih}, \zeta_{in},\zeta_{ic},\zeta_{ih})$ while $\mathbb{E}(f_i\mid m_i=m)$ is non-decreasing in $m$.

We will use the following lemma:

% \begin{lemma}\label{appx:lemma:AandB}
% Let $A(\eta_{ic},\eta_{in}, b_i)$ and $B(\eta_{ic},\eta_{in}, b_i)$ be random subsets of $\mathbb{R}$ such that 
% \[
% \sup A(\eta_{ic},\eta_{in}, b_i) \leq \inf B(\eta_{ic},\eta_{in}, b_i)
% \]
% with probability one. Then for $t=c,h$,
% \[
% \mathbb{E}(Y_i(t) \; \vert \; m_i \in A(\eta_{ic},\eta_{in}, b_i)) \leq 
% \mathbb{E}(Y_i(t) \vert m_i \in B(\eta_{ic},\eta_{in}, b_i)).
% \]
% \end{lemma}

\begin{lemma}\label{appx:lemma:AandB}
Let $A(\eta_{ic},\eta_{ih}, b_i)$ and $B(\eta_{ic},\eta_{ih}, b_i)$ be random subsets of $\mathbb{R}$ such that 
\[
\sup A(\eta_{ic},\eta_{ih}, b_i) \leq \inf B(\eta_{ic},\eta_{ih}, b_i)
\]
with probability one. Then for $t=c,h$,
\[
\mathbb{E}(Y_i(t) \mid m_i \in A(\eta_{ic},\eta_{ih}, b_i)) \leq 
\mathbb{E}(Y_i(t) \mid m_i \in B(\eta_{ic},\eta_{ih}, b_i)).
\]
\end{lemma}

\begin{proof}[Proof of \Cref{appx:lemma:AandB}]
Take $t\in \{c,h\}$. Since $\mathbb{E}(Y_i(t)\mid m_i=m) = \mathbb{E}(f_i\mid m_i=m) + k_t+m p_t$, it is a non-decreasing function of $m$ (recalling that $\mathbb{E}(f_i\mid m_i=m)$ is non-decreasing and $p_t > 0$). Fix $(\eta_{ic},\eta_{ih}, b_i)$; obviously, the distribution of $m_i$ conditional on $m_i \in B(\eta_{ic},\eta_{ih}, b_i)$ first-order stochastically dominates that of $m_i$ conditional on $m_i\in A(\eta_{ic},\eta_{ih}, b_i)$. Therefore,
\[
\mathbb{E}(Y_i(t) \mid m_i \in A(\eta_{ic},\eta_{ih}, b_i), \eta_{ic},\eta_{ih}, b_i) \leq 
\mathbb{E}(Y_i(t) \mid m_i \in B(\eta_{ic},\eta_{ih}, b_i), \eta_{ic},\eta_{ih}, b_i).
\]
Taking the expectation over $(\eta_{ic},\eta_{ih}, b_i)$ completes the proof.
\end{proof}

% \begin{proof}[Proof of \Cref{appx:lemma:AandB}]
% Take $t\in \{c,h\}$. Since  $\mathbb{E}(Y_i(t)\vert m_i=m) =  E(f_i\vert m_i=m) + }k_t+m p_t$,  it  is an increasing function of $m$. Fix $(\eta_{ic},\eta_{ih}, b_i)$; obviously, the distribution of  $m_i$  conditional on $m_i \in B(\eta_{ic},\eta_{in}, b_i)$ first-order stochastically dominates that of  $m_i$ conditional on $m_i\in A(\eta_{ic},\eta_{in}, b_i)$. Therefore 
% \[
% \mathbb{E}(Y_i(t) \; \vert \; m_i \in A(\eta_{ic},\eta_{in}, b_i), \eta_{ic},\eta_{ih}, b_i) \leq 
% \mathbb{E}(Y_i(t) \; \vert \; m_i \in B(\eta_{ic},\eta_{in}, b_i), \eta_{ic},\eta_{ih}, b_i).
% \]
% Taking the expectation over $(\eta_{ic},\eta_{ih}, b_i)$ completes the proof.
% \end{proof}

\subsection{The Binary Instrument Case}\label{appx:pos:sel:binary}

In Section \ref{sub:kw:subst}, subjects  who are assigned $z=1$ receive a Head Start offer; those with $z=0$ do not. 
The complier group $C_{ch}$ has
\begin{align*}
 \underline{U}_c + v_{ic}   &\geq \max(0, \underline{U}_h+v_{ih}), \\
 \bar{U}_h + v_{ih}   &\geq \max(0, \underline{U}_c+v_{ic}).
\end{align*}
and the  complier group $C_{nh}$ has
\begin{align*}
0 &\geq \max(\underline{U}_c+v_{ic}, \underline{U}_h+v_{ih}), \\
 \bar{U}_h + v_{ih}   &\geq \max(0, \underline{U}_c+v_{ic}).
\end{align*}
Note that $v_{ic}\geq-\underline{U}_c$ in $C_{ch}$ and  $v_{ic}\leq -\underline{U}_c$ in $C_{nh}$.  Fix values $\bar{\eta}_c, \bar{\eta}_h,$ and $\bar{b}$ and take $i\in C_{ch}$ and $j\in C_{nh}$ such that 
\[
 \eta_{ic}   =\eta_{jc}=\bar{\eta}_c; \;    
 \eta_{ih}   =\eta_{jh}=\bar{\eta}_h; \;
 b_{i}   = b_{j}=\bar{b}.
\]
Since $v_{kc}=\bar{b} d_c + \bar{b} m_k p_c+\bar{\eta}_{c}$ for $k=i,j$ and $p_c \bar{b}>0$, it follows that $m_i\geq m_j$. 
  Therefore we can apply Lemma~\ref{appx:lemma:AandB} with $t=h$ to obtain
  
\[
\mathbb{E}(Y_i(h) \vert i \in C_{ch})\geq \mathbb{E}(Y_i(h) \vert i \in C_{nh}),
\]
which is our version of positive selection in the binary case.

%\Bernardnote{OK, this one works.}

\subsection{The Ternary Instrument Case}\label{appx:pos:sel:ternary}

In our setup in Section~\ref{sub:kw:rationed}, subjects who are assigned $z=1$ receive a Head Start offer,  and those who are assigned $z=2$ are offered admission to another preschool; those with $z=0$ receive neither.

First note that under our assumptions,
\[
\mathbb{E}(Y_i(n) \mid \eta_{ic},\eta_{ih},b_i)= \mathbb{E}(f_i \mid \eta_{ic},\eta_{ih},b_i) + k_n.
\]
Given our assumptions on the distribution of $f_i$, we have
\[
\mathbb{E}(f_i \mid \eta_{ic},\eta_{ih},b_i) = \mathbb{E}\left[\mathbb{E}(f_i \mid \eta_{ic},\eta_{ih},b_i, m_i) \mid \eta_{ic},\eta_{ih},b_i\right]= \mathbb{E}\left[\mathbb{E}(f_i \mid m_i) \mid \eta_{ic},\eta_{ih},b_i\right].
\]
Figure~\ref{fig:kw:3by3} shows that if $i\in \tilde{C}_{nnn}$ and $j\in \tilde{C}_{nhc}$, then $v_{ih}\leq v_{jh}$ and $v_{ic}\leq v_{jc}$. Since $v_{it}=b_i d_t + b_i m_i p_t+\eta_{it},$ the distribution of $m$ in $\tilde{C}_{nnn}$ is dominated by its distribution in $\tilde{C}_{nhc}$. Since $\mathbb{E}(f_i\mid m_i=m)$ is non-decreasing in $m$, it follows that
\[
\mathbb{E}(Y_i(n)\mid i \in \tilde{C}_{nhc})
\geq 
\mathbb{E}(Y_i(n)\mid i \in \tilde{C}_{nnn}).
\]

Now let us consider the response-groups $\tilde{C}_{n*c}$ and $\tilde{C}_{chc}$. 
$\tilde{C}_{n*c}$ is defined by the inequalities 
\begin{equation}\label{ineq:first:1}
\bar{U}_c+v_{ic} > 0 > \max(\underline{U}_h+v_{ih}, \underline{U}_c+v_{ic});
\end{equation}
and $\tilde{C}_{chc}$ is defined by the inequalities 
\begin{equation}\label{ineq:first:2}
\bar{U}_h+v_{ih} > \underline{U}_c+v_{ic} > \max(\underline{U}_h+v_{ih}, 0).
\end{equation}
\eqref{ineq:first:1} implies that $v_{ic}=b_i d_c + b_i m_i p_c +\eta_{ic}<-\underline{U}_c$, while \eqref{ineq:first:2} implies the reverse inequality. Here also, applying Lemma~\ref{appx:lemma:AandB} with $t=c$ directly gives the conclusion:
\[
\mathbb{E}(Y_i(c)\mid i \in \tilde{C}_{n*c})
\leq 
\mathbb{E}(Y_i(c)\mid i \in \tilde{C}_{chc}).
\]

\section*{Acknowledgements}

We thank the editor, the associate editor, and three anonymous referees, as well as Josh Angrist, Junlong Feng, Len Goff, and Jim Heckman for helpful comments.

\bigskip\noindent\textbf{Funding:} This work was supported by the European Research Council [ERC-2014-CoG-646917-ROMIA] and the UK Economic and Social Research Council [ES/P008909/1] to the CeMMAP.

{\singlespacing
\bibliographystyle{economet}
\bibliography{citations-TETI}
}

\newpage

\newpage

\section*{Online Appendices to ``Treatment Effects with Targeting Instruments''}

\renewcommand{\thepage}{\roman{page}}
\setcounter{page}{1}

\section{Proofs for  \Cref{sub:klm}}\label{appendix:KLM}

Let us  first translate \citeapos{kirkeboen2016} assumptions  in our notation to show that their assumptions are equivalent to strict one-to-one targeting.

KLM impose the following in their Assumption~4:
\begin{itemize}
   \item if $T_i(0)=1$ then $T_i(1)=1$
   \item if $T_i(0)=2$ then $T_i(2)=2$.
\end{itemize}
This can be viewed as a monotonicity assumption. It  excludes the twelve response groups $C_{10\ast}, C_{12\ast}, C_{2\ast 0}$, and $C_{2\ast 1}$. 

Their Proposition~2 proves point-identification of response-groups when one of three alternative assumptions is added to their Assumption~4.  We focus here on 
the irrelevance assumption in their Proposition~2~(iii), which is the weakest of the three and the one their application relies on.  In our notation, it states that:
\begin{itemize}
   \item if ($T_i(0)\neq 1$ and $T_i(1)\neq 1$),  then ($T_i(0)=2$ iff $T_i(1)=2$)
   \item if ($T_i(0)\neq 2$ and $T_i(2)\neq 2$),  then ($T_i(0)=1$ iff $T_i(2)=1$).
\end{itemize}
These complicated statements can be simplified. Take the first part. If both $T_i(0)$ and $T_i(1)$ are not~1, then they can only be~0 or~2. Therefore we are requiring $T_i(0)=T_i(1)$. Applying the same argument to the second part,  the irrelevance assumption becomes:
\begin{itemize}
   \item if ($T_i(0)\neq 1$ and $T_i(1)\neq 1$),  then $T_i(0)=T_i(1)$
   \item if ($T_i(0)\neq 2$ and $T_i(2)\neq 2$),  then $T_i(0)=T_i(2)$.
\end{itemize}
It therefore excludes the  response-groups $C_{02\ast}, C_{20\ast}, C_{0\ast 1}$, and $C_{1\ast 0}$. The response-group $C_{021}$ appears twice in this list; and  four other response-groups   were already ruled out by Assumption~4. The reader can easily check that the $3^3-12-(11-4)=8$ response-groups left are exactly the same as in our~\Cref{fig:ternary-ternary}. 

\begin{proof}[Proof of \Cref{prop:IV:3by3}]

   The  moment conditions that define $\beta_0$, $\beta_1$ and $\beta_2$ are 
   \begin{equation}\label{eq:moments:tsls}
	   \mathbb{E} \left[\left(Y_i-\beta_0-\beta_1
	   \uniset(T_i=1)-\beta_2\uniset(T_i=2)\right)\uniset(Z_i=z)\right]=0
   \end{equation}
   for $z=0,1,2$. 

Using counterfactual notation, we write
\begin{equation}\label{cf-y-eq}
Y_i = Y_i(0) +  (Y_i(1) - Y_i(0)) \uniset(T_i = 1)
+  (Y_i(2) - Y_i(0)) \uniset(T_i = 2),
\end{equation}
which allows us to write \Cref{eq:moments:tsls} as
\begin{equation}\label{eq:moments2:tsls}
   \mathbb{E} \left[\left(Y_i(0)-\beta_0 + b_i(1)
   \uniset(T_i=1)+b_i(2)\uniset(T_i=2)\right)\uniset(Z_i=z)\right]=0,
\end{equation}
where  $b_i(t) \equiv Y_i(t)-Y_i(0)-\beta_t$ for $t=1,2$.

Now since
\begin{align*}
\uniset(T_i = t) &= \uniset(T_i(0) = t) 
+ (\uniset(T_i(1) = t)  - \uniset(T_i(0) = t))\uniset(Z_i = 1)\\
&+ (\uniset(T_i(2) = t)  - \uniset(T_i(0) = t))\uniset(Z_i = 2),
\end{align*}
we can expand 
\begin{align*}
   &\left[Y_i(0)-\beta_0 + b_i(1)
   \uniset(T_i=1)+b_i(2)\uniset(T_i=2)\right]
\times \uniset(Z_i = z) \\ 
   &=\left[Y_i(0)-\beta_0 + b_i(1)
   \uniset(T_i(0)=1)+b_i(2)\uniset(T_i(0)=2)\right.\\
   &\left. \; \; + b_i(1) (\uniset(T_i(z) = 1)  - \uniset(T_i(0) = 1))+ b_i(2) (\uniset(T_i(z) = 2)  - \uniset(T_i(0) = 2))\right]
   \times \uniset(Z_i=z).
\end{align*}
 By \Cref{assn:valid}(ii), for all response groups $C$ and all $(t,z)$,
\begin{equation*}
  \Pr(i\in C\mid Z_i=z)=\Pr(i\in C)
  \quad\text{and}\quad
  \mathbb{E}(Y_i(t)\mid i\in C,\,Z_i=z)=\mathbb{E}(Y_i(t)\mid i\in C).
\end{equation*}
 
Fix $z^\prime$ and $t$ and a  response group $C$.
Since $\uniset(T_i(z')=t)$ is constant within  $C$,
 the term $b_i(t)\bigl(\uniset(T_i(z')=t)-\uniset(T_i(0)=t)\bigr)$  equals  $d_Cb_i(t)\in\{-1,0,1\}$ for some constant $d_C$, so that
\begin{align*}
&\mathbb{E}\!\left[b_i(t)\bigl(\uniset(T_i(z')=t)-\uniset(T_i(0)=t)\bigr)\uniset(Z_i=z)\right]\\
  &\quad=\sum_C d_C\,
     \mathbb{E}\!\left[b_i(t)\mid i\in C,\,Z_i=z\right]\Pr(i\in C\mid Z_i=z)\,P(Z_i=z).
\end{align*}
 \Cref{assn:valid}(ii) gives $\Pr(i\in C\mid Z_i=z)=\Pr(i\in C)$,  and since $b_i(t)=Y_i(t)-Y_i(0)-\beta_t$ is an affine function of $Y_i(\cdot)$, 
 \[
\mathbb{E}\!\left[b_i(t)\mid i\in C,\,Z_i=z\right]=
 \mathbb{E}\!\left[b_i(t)\mid i\in C\right].
 \]
 It follows that
\begin{align*}
&\mathbb{E}\!\left[b_i(t)\bigl(\uniset(T_i(z')=t)-\uniset(T_i(0)=t)\bigr)\uniset(Z_i=z)\right]\\
  &\quad=\sum_C d_C\,
     \mathbb{E}\!\left[b_i(t)\mid i\in C\right]\Pr(i\in C)\,P(Z_i=z)\\
  &\quad=\mathbb{E}\!\left[b_i(t)\bigl(\uniset(T_i(z')=t)-\uniset(T_i(0)=t)\bigr)\right]P(Z_i=z).
\end{align*}
 
The same argument applies to the remaining terms in~\eqref{eq:moments2:tsls}.
Dividing by $P(Z_i=z)$ for $z=0,1,2$ yields
\begin{align*}
&\mathbb{E} \left[Y_i(0)-\beta_0 + b_i(1)
\uniset(T_i(0)=1)+b_i(2)\uniset(T_i(0)=2)\right.\\
&\left. + b_i(1) (\uniset(T_i(z) = 1)  - \uniset(T_i(0) = 1))+ b_i(2) (\uniset(T_i(z) = 2)  - \uniset(T_i(0) = 2))\right]=0.	
\end{align*}
When $z=0$, the second line is zero; therefore
\[
   \mathbb{E} \left(Y_i(0)-\beta_0 + b_i(1)
\uniset(T_i(0)=1)+b_i(2)\uniset(T_i(0)=2)\right)=0.
\]
The other two equations  become
\[
   \mathbb{E} \left(b_i(1) (\uniset(T_i(z) = 1)  - \uniset(T_i(0) = 1))+ b_i(2) (\uniset(T_i(z) = 2)  - \uniset(T_i(0) = 2))\right)=0
   \]
   for $z=1,2$. Remembering that $b_i(t)=Y_i(t)-Y_i(0)-\beta_t$ for $t=1,2$, we obtain
   \begin{align*}
   &\mathbb{E} \left[(Y_i(1)-Y_i(0)) (\uniset(T_i(z) = 1)  - \uniset(T_i(0) = 1))+ (Y_i(2)-Y_i(0)) (\uniset(T_i(z) = 2)  - \uniset(T_i(0) = 2))\right]\\
   &=\beta_1 \mathbb{E} (\uniset(T_i(z) = 1)  - \uniset(T_i(0) = 1))
   +\beta_2 \mathbb{E} (\uniset(T_i(z) = 2)  - \uniset(T_i(0) = 2)).
   \end{align*}
   Proposition~\ref{prop:IV:3by3} follows after noting that given Table~\ref{tab:map:ternary:ternary}, 
   \begin{itemize}
	   \item  the variable $\uniset(T_i(z) = 1)  - \uniset(T_i(0) = 1)$  is $\uniset(i\in \mathcal{C}_1)$ for $z=1$ and $-\uniset(i\in C_{112})$ for $z=2$;
	   \item  the variable $\uniset(T_i(z) = 2)  - \uniset(T_i(0) = 2)$  is $\uniset(i\in \mathcal{C}_2)$ for $z=2$ and $-\uniset(i\in C_{212})$ for $z=1$.
   \end{itemize}
\end{proof}

\begin{proof}[Proof of \Cref{cor:IV:3by3}]
Solving the system of equations in \Cref{prop:IV:3by3} gives, after elementary calculations,
\begin{align*}
   \beta_1 \mathcal{D} &= \Pr(i\in C_{212})
   \left[
   \mathbb{E} \left(
	   (Y_i(2)-Y_i(0))\uniset(i\in \mathcal{C}_2)
	   \right)
	   - \mathbb{E} \left(
		   (Y_i(1)-Y_i(0))\uniset(i\in C_{112})
		   \right)
		   \right]\\
   &+
   \Pr(i\in \mathcal{C}_2)\left[
	   \mathbb{E} \left(
		   (Y_i(1)-Y_i(0))\uniset(i\in \mathcal{C}_1)
		   \right)
	   -  \mathbb{E} \left(
		   (Y_i(2)-Y_i(0))\uniset(i\in C_{212})
		   \right)
		   \right]\\
		   &= \Pr(i\in \mathcal{C}_1)\Pr(i\in \mathcal{C}_2)
			   \mathbb{E} \left(
				   Y_i(1)-Y_i(0)\vert i\in \mathcal{C}_1
				   \right)
				   - \Pr(i\in C_{112})\Pr(i\in C_{212})
				   \mathbb{E} \left(
				   Y_i(1)-Y_i(0)\vert i\in C_{112}
				   \right)\\
   &+ \Pr(i\in C_{212})\Pr(i\in \mathcal{C}_2)
   \left[
   \mathbb{E} \left(
	   Y_i(2)-Y_i(0)\vert i\in \mathcal{C}_2
	   \right)
	   - 
	   \mathbb{E} \left(
	   Y_i(2)-Y_i(0)\vert i\in C_{212}
	   \right)
	   \right].
\end{align*}
The difference of treatment effects in the last line is simply $\mathcal{D}_2$; note that it is multiplied by a non-negative term.  
Suppose for instance that  $\mathcal{D}_1, \mathcal{D}_2\geq 0$. Then 
\begin{align}\label{ineq:beta1}
\begin{split}
&	\beta_1 \mathcal{D}  \\
& \geq
\Pr(i\in \mathcal{C}_1)\Pr(i\in \mathcal{C}_2)
			   \mathbb{E} \left(
				   Y_i(1)-Y_i(0)\vert i\in \mathcal{C}_1
				   \right) 
				   - \Pr(i\in C_{112})\Pr(i\in C_{212})
				   \mathbb{E} \left(
				   Y_i(1)-Y_i(0)\vert i\in C_{112}
				   \right). 
\end{split}
\end{align}
Moreover, it is easy to prove the following: define $r=(\alpha a-\beta b)/(a-b)$ with $a,b \geq 0$ and  $a\neq b$.  Then 
\begin{enumerate}
   \item if $(\alpha-\beta)$ and $(a-b)$ have the same sign, $r\geq \max(\alpha,\beta)$
   \item if $(\alpha-\beta)$ and $(a-b)$ have different signs, $r\leq\min(\alpha,\beta)$.
\end{enumerate}
Now take 
\begin{align*}
a &=\Pr(i\in \mathcal{C}_1)\Pr(i\in \mathcal{C}_2)\\
b&=\Pr(i\in C_{112})\Pr(i\in C_{212})\\
\alpha &= \mathbb{E} \left(
   Y_i(1)-Y_i(0)\vert i\in \mathcal{C}_1
   \right)\\
   \beta &= \mathbb{E} \left(
	   Y_i(1)-Y_i(0)\vert i\in C_{112}
	   \right).
\end{align*}
Note that $a$ and $b$ are non-negative, and $a-b=\mathcal{D}\neq 0$. Suppose that $\mathcal{D}>0$ so that \Cref{ineq:beta1} becomes $\beta_1  \geq r$.
Since $\alpha-\beta=\mathcal{D}_1\geq 0$,  we can apply result 1 and we get 
\[
   \beta_1 \geq \max(\alpha,\beta)=\alpha= \mathbb{E} \left(
	   Y_i(1)-Y_i(0)\vert i\in \mathcal{C}_1
	   \right).
\]
If on the other hand $\mathcal{D}$ is negative, then we have $\beta_1 \leq r$ and since $\mathcal{D}$ and $\mathcal{D}_1$ have different signs result~2 gives 
\[
   \beta_1 \leq \min(\alpha,\beta)=\beta 
\]
and a fortiori $\beta_1\leq \alpha$.

Similar arguments apply to $\beta_2$, as well as to  the case when $\mathcal{D}_1$ and $\mathcal{D}_2$ are non-positive.
\end{proof}

% {\color{blue}

% \Simonnote{Appendix on HP is now added}

\section{Revisiting the $2 \times 3$ and $3 \times 3$ Models via \citet{heckmanpinto-pdt}}\label{appendix:HP} 

\subsection{Notation}

We first adapt \citet[HP hereafter]{heckmanpinto-pdt}'s notation to our framework. As in the main text, we focus on identifying the probabilities of the various response groups $\Pr(i\in C)$ and the  group average  outcomes $\mathbb{E}(Y_i(t)\vert i\in C)$. The following population quantities are directly identified from data for all treatment values $t$:
\begin{align*}
\bm{P}_Z(t) &= \left(P\left(T=t \mid Z=z\right)\right)_{z\in\calZ}, \\
\bm{Q}_Z(t) &=\left(
\mathbb{E}\left(Y \; \mathbf{1}(T=t) \mid Z=z\right)\right)_{z\in\calZ}.
\end{align*}
We also define $\bm{P}_Z=\left(\bm{P}_Z\left(t\right)\right)_{t\in\calT}$.

We choose an arbitrary ordering $(C^1,\ldots,C^S)$ of the $S$ non-empty response groups  and we define the $S$ dummy  variables $c_{i}^s=\mathbf{1}(i\in C^s)$. 
The {\em response vector\/} $\bm{S}$ is $\{c^1, \ldots, c^{S} \}$.
With  this notation, our main objects of interest are
\begin{align*}
\boldsymbol{P}_S &= \mathbb{E}\bm{S} \\
\boldsymbol{Q}_S(t) &= \mathbb{E}\left(Y(t)\bm{S}\right) \text{ for } t \in \calT,
\end{align*}
from which we obtain $\Pr(i\in C^s)= \boldsymbol{P}_S^s$ and $\mathbb{E}(Y_i(t)\vert i\in C^s)=\boldsymbol{Q}^s_S(t)/\boldsymbol{P}_S^s$.

As in HP, $\boldsymbol{B}_t$ denotes a binary matrix with dimension  $\abs{\calZ} \times S$ whose element in row $z$ and column $s$ equals 1 if response group $C^s$ has $T_i=t$ when $Z_i=z$, and zero otherwise. Finally, let $\bcalB$ be the binary matrix of dimension $\left(\abs{\calZ} \cdot \abs{\calT}\right) \times S$ generated by stacking  the matrices $\boldsymbol{B}_t$ vertically: $\bcalB=\left[\boldsymbol{B}_{0}^{\prime}, \ldots, \boldsymbol{B}_{\abs{\calT} -1}^{\prime}\right]^{\prime}$.

\subsection{Theorem T-2 in HP}
Let $\boldsymbol{M}^{\dagger}$ denote the Moore-Penrose pseudo-inverse of a matrix $\boldsymbol{M}$.  We define
$$
\boldsymbol{K}_t=\boldsymbol{I}_{S}-\boldsymbol{B}_t^{\dagger} \boldsymbol{B}_t  \; \text{ and } \;
\bcalK=\boldsymbol{I}_{S}-\bcalB^{\dagger} \bcalB,
$$
where $\boldsymbol{I}_{S}$ denotes the identity matrix of dimension $S$. Note that $\bcalK$ and $\boldsymbol{K}_t$ are orthogonal projection matrices in $\mathbb{R}^S$ that only depend  on the  binary matrices $\bcalB$ and $\boldsymbol{B}_t$.
Theorem T-2 in HP shows that 
\begin{align}
\boldsymbol{P}_S & =\bcalB^{\dagger}
\boldsymbol{P}_Z+\bcalK \boldsymbol{\lambda}, \label{HP-Thm2-a} \\
\boldsymbol{Q}_S(t) & =\boldsymbol{B}_t^{\dagger} \boldsymbol{Q}_Z(t)+\boldsymbol{K}_t \tilde{\boldsymbol{\lambda}}, \label{HP-Thm2-b}
\end{align}
where $\boldsymbol{\lambda}$ and $\tilde{\boldsymbol{\lambda}}$ are arbitrary $S$-dimensional vectors.

\subsection{Identification in the $2 \times 3$ model}\label{appx:hp:2by3}
We can now  re-derive our identification results for the 2 by 3 model using the theorems in HP.
To do so, we order the response-types as $\{ C_{00}, C_{11}, C_{22}, C_{01}, C_{21} \}$.
Then the binary matrices $\boldsymbol{B}_0, \boldsymbol{B}_1$, and $\boldsymbol{B}_2$  are
\begin{align*}
\boldsymbol{B}_0 
&=
\begin{bmatrix}
1 & 0 & 0 & 1 & 0 \\
1 & 0 & 0 & 0 & 0
\end{bmatrix}, \\
\boldsymbol{B}_1 
&=
\begin{bmatrix}
0 & 1 & 0 & 0 & 0 \\
0 & 1 & 0 & 1 & 1
\end{bmatrix}, \\
\boldsymbol{B}_2 
&=
\begin{bmatrix}
0 & 0 & 1 & 0 & 1 \\
0 & 0 & 1 & 0 & 0
\end{bmatrix},
\end{align*}
and
\begin{align*}
\bcalB
&=
\begin{bmatrix}
1 & 0 & 0 & 1 & 0 \\
1 & 0 & 0 & 0 & 0 \\
0 & 1 & 0 & 0 & 0 \\
0 & 1 & 0 & 1 & 1 \\
0 & 0 & 1 & 0 & 1 \\
0 & 0 & 1 & 0 & 0
\end{bmatrix}.
\end{align*}
It is easy to see that $\bcalB$ has full column rank; it follows that $\bcalB^\dagger \bcalB=\boldsymbol{I}_5$ and  $\bcalK$ is the 5 by 5 matrix with all elements zero. 
Therefore  by Theorem T-2 in HP (see equation \eqref{HP-Thm2-a} above), $\boldsymbol{P}_S$ is point-identified as
$\boldsymbol{P}_S  =\bcalB^{\dagger} \boldsymbol{P}_Z$.

Since 
\begin{align*}
\bcalB^{\dagger}
&=\frac16
\begin{bmatrix}
 1&  5&  1& -1&  1& -1 \\
       -1&  1&  5&  1& -1&  1 \\
        1& -1&  1& -1&  1&  5 \\
        4& -4& -2&  2& -2&  2 \\
       -2&  2& -2&  2&  4& -4
\end{bmatrix},
\end{align*}
this  is not very transparent, however. 
To derive our identification results, we use~\eqref{HP-Thm2-b}  instead. Note that the equation $\boldsymbol{Q}_S(t) =\boldsymbol{B}_t^{\dagger} \boldsymbol{Q}_Z(t)+\boldsymbol{K}_t \tilde{\boldsymbol{\lambda}}$ holds for any function of $Y(t)$. If we take it to be a constant function of $Y(t)$, we get $\boldsymbol{Q}_S(t)=\mathbb{E}\bm{S}=\boldsymbol{P}_S$ and $\boldsymbol{Q}_Z(t)=\boldsymbol{P}_Z(t)$, so that~\eqref{HP-Thm2-b}  boils down to 
\begin{equation}\label{eq:getPS}
    \boldsymbol{P}_S=\boldsymbol{B}_t^{\dagger} \boldsymbol{P}_Z(t)+\boldsymbol{K}_t \tilde{\boldsymbol{\lambda}} \; \text{ for all values of } t.
\end{equation}

Now
\begin{align*}
\boldsymbol{B}_0^{\dagger} 
=
\begin{bmatrix}
0 & 1  \\
0 & 0 \\
0 & 0 \\
1 & -1 \\
0 & 0
\end{bmatrix}
&\Longrightarrow
\boldsymbol{K}_{0} 
=
\begin{bmatrix}
0 & 0 & 0 & 0 & 0 \\
0 & 1 & 0 & 0 & 0 \\
0 & 0 & 1 & 0 & 0 \\
0 & 0 & 0 & 0 & 0 \\
0 & 0 & 0 & 0 & 1
\end{bmatrix};
\end{align*}
\begin{align*}
\boldsymbol{B}_1^{\dagger} 
=
\begin{bmatrix}
0 & 0  \\
1 & 0 \\
0 & 0 \\
-1/2 & 1/2 \\
-1/2 & 1/2
\end{bmatrix}
&\Longrightarrow
\boldsymbol{K}_{1} 
=
\begin{bmatrix}
1 & 0 & 0 & 0 & 0 \\
0 & 0 & 0 & 0 & 0 \\
0 & 0 & 1 & 0 & 0 \\
0 & 0 & 0 & 1/2 & -1/2 \\
0 & 0 & 0 & -1/2 & 1/2
\end{bmatrix};
\end{align*}
and 
\begin{align*}
\boldsymbol{B}_2^{\dagger} 
=
\begin{bmatrix}
0 & 0  \\
0 & 0 \\
0 & 1 \\
0 & 0 \\
1 & -1
\end{bmatrix}
&\Longrightarrow
\boldsymbol{K}_{2} 
=
\begin{bmatrix}
1 & 0 & 0 & 0 & 0 \\
0 & 1 & 0 & 0 & 0 \\
0 & 0 & 0 & 0 & 0 \\
0 & 0 & 0 & 1 & 0 \\
0 & 0 & 0 & 0 & 0
\end{bmatrix}.
\end{align*}
Let $(\bm{e}_s)_{s=1,\ldots,5}$ denote the standard basis vectors in $\mathbb{R}^5$.  If $\bm{e}^\prime_s \bm{K}_t=0$, then~\eqref{eq:getPS} point-identifies $\Pr(i\in C^s)=\bm{e}^\prime_s \bm{P}_S = \bm{e}^\prime_s \bm{B}_t^{\dagger} \bm{P}_Z(t)$. Clearly,
\[
\bm{e}^\prime_1\bm{K}_0=\bm{e}^\prime_4\bm{K}_0=\bm{e}^\prime_2\bm{K}_1=\bm{e}^\prime_3\bm{K}_2=\bm{e}^\prime_4\bm{K}_2=0;
\]
this reproduces our identification results for $\boldsymbol{P}_S$ in \Cref{pro:DCM:vector:binaryinstrument}: 
\begin{align*}
\boldsymbol{P}_S=
\left(
\bm{e}_1' \boldsymbol{B}_0^{\dagger} \bm{P}_Z(0),
\bm{e}_2' \boldsymbol{B}_1^{\dagger} \bm{P}_Z(1),
\bm{e}_3' \boldsymbol{B}_2^{\dagger} \bm{P}_Z(2),
\bm{e}_4' \boldsymbol{B}_0^{\dagger} \bm{P}_Z(0),
\bm{e}_5' \boldsymbol{B}_2^{\dagger} \bm{P}_Z(2)
\right)'.
\end{align*}
Returning to the counterfactual outcomes  $Y(t)$, the same argument 
results in \Cref{pro:DCM:late:binaryinstrument}: 
\begin{align*}
\mathbb{E} \left(Y_i(0) \cdot \mathbf{1}\left[i\in C_{00}\right]\right) &= \bm{e}_1' \boldsymbol{B}_0^{\dagger} \bm{Q}_Z(0), \\
\mathbb{E} \left(Y_i(1) \cdot \mathbf{1}\left[i\in C_{11}\right]\right) &= \bm{e}_2' \boldsymbol{B}_1^{\dagger} \bm{Q}_Z(1), \\
\mathbb{E} \left(Y_i(2) \cdot \mathbf{1}\left[i\in C_{22}\right]\right) &= \bm{e}_3' \boldsymbol{B}_2^{\dagger} \bm{Q}_Z(2), \\
\mathbb{E} \left(Y_i(0) \cdot \mathbf{1}\left[i\in C_{01}\right]\right) &= \bm{e}_4' \boldsymbol{B}_0^{\dagger} \bm{Q}_Z(0), \\
\mathbb{E} \left(Y_i(2) \cdot \mathbf{1}\left[i\in C_{21}\right]\right) &= \bm{e}_5' \boldsymbol{B}_2^{\dagger} \bm{Q}_Z(2).
\end{align*}
We conclude that while the  first part of Theorem T-2 in HP 
(i.e., $\boldsymbol{P}_S  =\bcalB^{\dagger} \boldsymbol{P}_Z+\bcalK \boldsymbol{\lambda}$)
is useful to determine the degrees of identification by checking the rank of $\bcalB$, it does not yield the most constructive form of identification.
To get the objects of interest, it is better to invoke 
the second part of Theorem T-2 (i.e., $\boldsymbol{Q}_S(t) =\boldsymbol{B}_t^{\dagger} \boldsymbol{Q}_Z(t)+\boldsymbol{K}_t \tilde{\boldsymbol{\lambda}}$).
Note that since the $2 \times 3$ model satisfies the unordered monotonicity assumption, we could also obtain the same results using Theorem T-6 in HP.

\subsection{Identification in the $3 \times 3$ model}
We now turn to  our 3 by 3 model. We sort the response-types as 
\[
\left\{ C_{000}, C_{111}, C_{222}, 
C_{010}, C_{002}, C_{012}, C_{112}, C_{212}\right\}.
\]
Now 
\begin{align*}
\boldsymbol{B}_0 
&=
\begin{bmatrix}
1 & 0 & 0 & 1 & 1 & 1 & 0 & 0 \\
1 & 0 & 0 & 0 & 1 & 0 & 0 & 0 \\
1 & 0 & 0 & 1 & 0 & 0 & 0 & 0
\end{bmatrix}, \\
\boldsymbol{B}_1 
&=
\begin{bmatrix}
0 & 1 & 0 & 0 & 0 & 0 & 1 & 0 \\
0 & 1 & 0 & 1 & 0 & 1 & 1 & 1 \\
0 & 1 & 0 & 0 & 0 & 0 & 0 & 0 
\end{bmatrix}, \\
\boldsymbol{B}_2 
&=
\begin{bmatrix}
0 & 0 & 1 & 0 & 0 & 0 & 0 & 1 \\
0 & 0 & 1 & 0 & 0 & 0 & 0 & 0 \\
0 & 0 & 1 & 0 & 1 & 1 & 1 & 1 
\end{bmatrix}.
\end{align*}
%
\begin{comment}
Suppose that the unordered monotonicity assumption were to be held in our 3 by 3 model.
Among all possible combinations of $(t,i)$ for $\Sigma_t(i)$, the following cases consist of a single response-type:
$\Sigma_0(3) =  C_{000}$,
$\Sigma_0(1) =  C_{012}$,
$\Sigma_1(3) =  C_{111}$, 
$\Sigma_1(2) =  C_{112}$, 
$\Sigma_2(3) =  C_{222}$,
and
$\Sigma_2(2) =  C_{212}$. 
Then, using by Theorem T-6 in HP (see equation \eqref{HP-Thm6-a} above),
we could identify the popular shares of $C_{000}$ and $C_{012}$ by
$\boldsymbol{b}_0(3)' \boldsymbol{B}_0^{\dagger} \bm{P}_Z(0)$
and
$\boldsymbol{b}_0(1)' \boldsymbol{B}_0^{\dagger} \bm{P}_Z(0)$
with 
$\boldsymbol{b}_0(3) = \bm{e}_1$,
$\boldsymbol{b}_0(1) = \bm{e}_6$.
\end{comment}
%
Note that  
\begin{align*}
\boldsymbol{B}_0^{\dagger} 
=
\begin{bmatrix}
-0.25 & 0.5 & 0.5  \\
0 & 0 & 0 \\
0 & 0 & 0 \\
0.25 & -0.5 &  0.5 \\
0.25 & 0.5 &  -0.5 \\
0.75 & -0.5 &  -0.5 \\
0 & 0 & 0 \\
0 & 0 & 0
\end{bmatrix}
&\Longrightarrow
\boldsymbol{K}_{0} 
=
\begin{bmatrix}
0.25 & 0 & 0 & -0.25 & -0.25 & 0.25 & 0 & 0 \\
0 & 1 & 0 & 0 & 0 & 0 & 0 & 0 \\
0 & 0 & 1 & 0 & 0 & 0 & 0 & 0 \\
-0.25 & 0 & 0 & 0.25 & 0.25 & -0.25 & 0 & 0 \\
-0.25 & 0 & 0 & 0.25 & 0.25 & -0.25 & 0 & 0 \\
0.25 & 0 & 0 & -0.25 & -0.25 & 0.25 & 0 & 0 \\
0 & 0 & 0 & 0 & 0 & 0 & 1 & 0 \\
0 & 0 & 0 & 0 & 0 & 0 & 0 & 1 
\end{bmatrix}.
\end{align*}
Let $(\bm{e}_s)_{s=1,\ldots,8}$ denote the standard basis vectors in $\mathbb{R}^8$. 
Since $\bm{K}_0$ has no zero column,  none of the $\bm{e}^\prime_s\bm{K}_0$ is zero and the argument in Section~\ref{appx:hp:2by3}  shows that no  population share $\Pr(i\in C^s)$ is point-identified by $\boldsymbol{B}_0^{\dagger} \bm{P}_Z(0)$. 
%Thus, we can conclude that the unordered monotonicity assumption does not hold for our 3 by 3 model (assuming that Theorem T-6 is correct :-). 
On the other hand, 
\begin{align*}
\boldsymbol{B}_1^{\dagger} 
=
\begin{bmatrix}
0 & 0 & 0 \\
0 & 0 & 1 \\
0 & 0 & 0 \\
-1/3 & 1/3 &  0 \\
0 & 0 &  0 \\
-1/3 & 1/3 &  0 \\
1 & 0 & -1 \\
-1/3 & 1/3 & 0
\end{bmatrix}
&\Longrightarrow
\boldsymbol{K}_{1} 
=
\begin{bmatrix}
1 & 0 & 0 & 0 & 0 & 0 & 0 & 0 \\
0 & 0 & 0 & 0 & 0 & 0 & 0 & 0 \\
0 & 0 & 1 & 0 & 0 & 0 & 0 & 0 \\
0 & 0 & 0 & 2/3 & 0 & -1/3 & 0 & -1/3 \\
0 & 0 & 0 & 0 & 1 & 0 & 0 & 0 \\
0 & 0 & 0 & -1/3 & 0 & 2/3 & 0 & -1/3 \\
0 & 0 & 0 & 0 & 0 & 0 & 0 & 0 \\
0 & 0 & 0 & -1/3 & 0 & -1/3 & 0 & 2/3 
\end{bmatrix}
\end{align*}
so that $\bm{e}^\prime_2\bm{K}_1=\bm{e}^\prime_7\bm{K}_1=0$, which point-identifies the population shares of $C_{111}$ and $C_{112}$. Similarly,
\begin{align*}
\boldsymbol{B}_2^{\dagger} 
=
\begin{bmatrix}
0 & 0 & 0 \\
0 & 0 & 0 \\
0 & 1 & 0 \\
0 & 0 & 0 \\
-1/3 & 0 &  1/3 \\
-1/3 & 0 &  1/3 \\
-1/3 & 0 &  1/3 \\
1 & -1 & 0 
\end{bmatrix}
&\Longrightarrow
\boldsymbol{K}_{2} 
=
\begin{bmatrix}
1 & 0 & 0 & 0 & 0 & 0 & 0 & 0 \\
0 & 1 & 0 & 0 & 0 & 0 & 0 & 0 \\
0 & 0 & 0 & 0 & 0 & 0 & 0 & 0 \\
0 & 0 & 0 & 1 & 0 & 0 & 0 & 0 \\
0 & 0 & 0 & 0 & 2/3 & -1/3 & -1/3 & 0 \\
0 & 0 & 0 & 0 & -1/3 & 2/3 & -1/3 & 0 \\
0 & 0 & 0 & 0 & -1/3 & -1/3 & 2/3 & 0 \\
0 & 0 & 0 & 0 & 0 & 0 & 0 & 0 
\end{bmatrix}
\end{align*}
and $\bm{e}^\prime_3\bm{K}_2=\bm{e}^\prime_8\bm{K}_2=0$ so that 
the shares of 
 $C_{222}$ and $C_{212}$
are point-identified. On the other hand, the shares of 
$C_{010}$, $C_{002}$, and $C_{012}$ are not identified.
The results in Propositions~\ref{pro:DCM:vector:ternaryternary} and~\ref{pro:DCM:3by3} follow.

Finally, note that  $\boldsymbol{B}_0$ has the following $2 \times 2$ sub-matrix: 
\begin{align*}
\begin{pmatrix}
     & [C_{010} & C_{002}] \\
[z = 1] & 1 & 0  \\
[z = 2]  & 0 & 1 
\end{pmatrix},
\end{align*}
where we indicate the relevant columns and rows of matrix $\boldsymbol{B}_0$. 
Given this pattern, Theorem T-3 and Remark 6.3 in \citet{heckmanpinto-pdt} imply that
the unordered monotonicity assumption is not satisfied for the $3 \times 3$ model.
As mentioned in the main text, when the instrument value switches from 1 to 2, observations in $C_{010}$ move to treatment $0$, while those in $C_{002}$  move out of treatment $0$. Recall that the ARUM structure rules out ``direct two-way flows'' (that is, instrument values $1$ and $2$, respectively, make treatments $1$ and $2$ more favorable for everyone). However, the $3 \times 3$ model allows for ``indirect two-way flows'', where treatment 0 is not targeted by either $z=1$ or $z=2$.  
Unordered monotonicity is more restrictive than ARUM in that it rules out both direct and indirect two-way flows.

\section{The $3\times 3$ Model of \citet{pinto2021}}\label{appendix:model:pinto}
 
 \citet{pinto2021} has proposed a $3\times 3$ model of the Moving to Opportunity (MTO) experiment. Here  we use our framework to identify response-group probabilities and several counterfactual averages. 

We follow the notation  in \citet{pinto2021}.
Let $\mathcal{Z} = \{z_c, z_e, z_8\}$ and $\mathcal{T} = \{t_h, t_l, t_m\}$,
where 
\begin{itemize}
    \item $z_c$ refers to control families,
    $z_e$ those who received the experimental voucher,
    and
    $z_8$ those who received a Section 8 voucher;
    \item 
    $t_h$ refers to families who did not move and chose high-poverty neighborhoods,
    $t_l$ those who moved to low-poverty neighborhoods,
    and
    $t_m$ those who moved to medium-poverty neighborhoods.
\end{itemize}

There are 7 response types in \citet{pinto2021}: the three always-taker groups  $C_{hhh}$, $C_{lll}$, and $C_{mmm}$, and four complier groups:
\begin{itemize}
\item $C_{hlm}$: families who choose high-poverty without vouchers, low-poverty with the experimental voucher,
and medium-poverty with Section 8 vouchers (Pinto calls this group full-compliers);
\item $C_{hll}$: families  who choose high-poverty without vouchers, low-poverty with either voucher;
\item $C_{mlm}$: families   who choose medium-poverty without the experimental voucher, low-poverty with it;
\item $C_{hhm}$: families  who  choose high-poverty without a Section 8 voucher, medium-poverty with it.
\end{itemize}

\begin{figure}[htb]
	\begin{tikzpicture}[scale=0.5]  
		\def\dotpone{-3.5}
\draw [green!20!white, fill=green!20!white] (0,-5) -- (0,0) -- (5,5) -- (5,-5);
\draw [green!20!white, fill=green!20!white] (-5,0) -- (\dotpone,0) -- (1,5) -- (-5,5);
\draw [blue!20!white, fill=blue!20!white] (0,0) -- (\dotpone,\dotpone) -- (\dotpone,0);
\draw [yellow!20!white, fill=yellow!20!white] (0,0) -- (5,5) -- (1,5) -- (\dotpone,0);	
\draw [red!20!white, fill=red!20!white, opacity=0.5] (\dotpone,0) -- (\dotpone,\dotpone) -- (-5,\dotpone) -- (-5,0);
\draw [orange!20!white, fill=orange!20!white, opacity=0.5] (0,0) -- (0,-5) -- (\dotpone,-5) -- (\dotpone,\dotpone);	
\draw [green!20!white, fill=green!20!white] (-5,\dotpone) -- (\dotpone,\dotpone) -- (\dotpone,-5) -- (-5,-5);

	\draw[->,>=latex] (-5,0)  -- (5,0) node[below] {$u_{il}-u_{ih}$};
	\draw[->,>=latex] (0,-5)  -- (0,5) node[above] {$u_{im}-u_{ih}$};
	\draw[dashed,red] (0,0) -- (5,5);
	\draw[dashed,red] (\dotpone,\dotpone) -- (0,0);
	\draw[dashed,red] (\dotpone,0) -- (1,5);
	\draw[dashed,red] (-5,\dotpone)  -- (\dotpone,\dotpone);
	\draw[dashed,red] (\dotpone,0)  -- (\dotpone,\dotpone);

\node[red] at (-2.5,-1) {$C_{hlm}$};
\node[red] at (-1.5,-4.5) {$C_{hll}$};
\node[red] at (-4.5,-1.5) {$C_{hhm}$};
\node[red] at (-4.5,-4.5) {$C_{hhh}$};
\node[red] at (3,-4) {$C_{lll}$};
\node[red] at (-3,3) {$C_{mmm}$};
\node[red] at (1,2.5) {$C_{mlm}$};

\fill[blue] (0.0, 0.0) circle (.1cm);	          
\fill[blue] (\dotpone,0.0) circle (.1cm);	          
\fill[blue] (\dotpone,\dotpone) circle (.1cm);	      
	
\node[blue] at (-0.4,0.5) {$P_c$};
\node[blue] at (-3.7,0.5) {$P_e$};
\node[blue] at (-3.9,-3.9) {$P_8$};

	\end{tikzpicture} 
	\caption{MTO}\label{fig:mto}
	\end{figure}

The seven response groups are illustrated in Figure~\ref{fig:mto} and in Table~\ref{tab:map:mto}.

\begin{table}[htbp]
    \caption{\label{tab:map:mto} Response Groups in MTO}\centering\medskip
{\small
        \begin{tabular}{lccc}  \toprule
       & $T_i(z)=t_h$  & $T_i(z)=t_l$  & $T_i(z)=t_m$ \\  \midrule \\ \vspace{3mm}
$z = z_c$ & $C_{hhh} \cup C_{hhm} \cup C_{hlm} \cup C_{hll}$ 
 & $C_{lll}$ & $C_{mmm} \cup C_{mlm}$ \\  \vspace{3mm}
$z = z_e$ & $C_{hhh} \cup C_{hhm}$ & $C_{lll} \cup C_{mlm} \cup C_{hlm} \cup C_{hll}$ & $C_{mmm}$ \\  \vspace{3mm}
$z = z_8$ & $C_{hhh}$ & $C_{lll} \cup C_{hll}$ & $C_{mmm} \cup C_{hhm} \cup C_{hlm} \cup C_{mlm} $\\  
\bottomrule \end{tabular}
}
    \end{table}

\begin{proposition}[Response-group probabilities in MTO]\label{pro:DCM:vector:mto}
The following probabilities are identified:
 \begin{align}\label{prob:vectors:mto}
 \begin{split}
\Pr (C_{hhh}) &= P(t_h \vert z_8),  \\
\Pr (C_{lll}) &= P(t_l \vert z_c),  \\
\Pr (C_{mmm}) &= P(t_m \vert z_e), \\
\Pr (C_{hhm}) &= P(t_h \vert z_e) - P(t_h \vert z_8),  \\
\Pr (C_{hll}) &= P(t_l \vert z_8) - P(t_l \vert z_c), \\
\Pr (C_{mlm}) &= P(t_m \vert z_c) - P(t_m \vert z_e), \\
\Pr (C_{hlm})  &= 1 - P(t_h \vert z_e) - P(t_l \vert z_8) - P(t_m \vert z_c).
\end{split}
\end{align}
 The model has the following testable implications: 
\begin{align}\label{testable:mto}
\begin{split}
	P(t_h \vert z_e) &\geq P(t_h \vert z_8), \\
	P(t_l \vert z_8) &\geq P(t_l \vert z_c), \\
	P(t_m \vert z_c) &\geq P(t_m \vert z_e), \\
	1 &\geq P(t_h \vert z_e) + P(t_l \vert z_8) + P(t_m \vert z_c).
\end{split}	
\end{align}
\end{proposition}

  The following proposition identifies a number of group average outcomes.

\begin{proposition}[Identification in MTO]\label{pro:DCM:mto}
The following group average outcomes are point-identified:
\begin{align*}
\mathbb{E} \left[ Y_i(t_h) \vert  i  \in C_{hhh} \right]  &= \frac{\bar{E}_{z_8}(t_h)}{P(t_h \vert z_8)},  \\
\mathbb{E} \left[ Y_i(t_l) \vert  i  \in C_{lll} \right]  &= \frac{\bar{E}_{z_c}(t_l)}{P(t_l \vert z_c)},  \\
\mathbb{E} \left[ Y_i(t_m) \vert  i  \in C_{mmm} \right]  &= \frac{\bar{E}_{z_e}(t_m)}{P(t_m \vert z_e)},  \\
\mathbb{E} \left[ Y_i(t_h) \vert  i  \in C_{hhm} \right]  &= \frac{\bar{E}_{z_e}(t_h) - \bar{E}_{z_8}(t_h)}{P(t_h \vert z_e) - P(t_h \vert z_8)},  \\
%\Pr (C_{hhm}) &= P(t_h \vert z_e) - P(t_h \vert z_8),  \\
\mathbb{E} \left[ Y_i(t_l) \vert  i  \in C_{hll} \right]  &= \frac{\bar{E}_{z_8}(t_l) - \bar{E}_{z_c}(t_l)}{P(t_l \vert z_8) - P(t_l \vert z_c)},  \\
%\Pr (C_{hll}) &= P(t_l \vert z_8) - P(t_l \vert z_c), \\
\mathbb{E} \left[ Y_i(t_m) \vert  i  \in C_{mlm} \right]  &= \frac{\bar{E}_{z_c}(t_m) - \bar{E}_{z_e}(t_m)}{P(t_m \vert z_c) - P(t_m \vert z_e)},  \\
%\Pr (C_{mlm}) &= P(t_m \vert z_c) - P(t_m \vert z_e), \\
%\Pr (C_{hlm})  &= 1 - P(t_h \vert z_e) - P(t_l \vert z_8) - P(t_m \vert z_c).
%
\mathbb{E} \left[ Y_i(t_h)  \vert  i  \in C_{hll} \cup C_{hlm} \right] &= \frac{\bar{E}_{z_c}(t_h) - \bar{E}_{z_e}(t_h)}{P(t_h \vert z_c)-P(t_h \vert z_e)}, \\
\mathbb{E} \left[ Y_i(t_l)  \vert  i  \in C_{mlm} \cup C_{hlm} \right] &= \frac{\bar{E}_{z_e}(t_l) - \bar{E}_{z_8}(t_l)}{P(t_l \vert z_e)-P(t_l \vert z_8)}, \\
\mathbb{E} \left[ Y_i(t_m)  \vert  i  \in C_{hhm} \cup C_{hlm} \right] &= \frac{\bar{E}_{z_8}(t_m) - \bar{E}_{z_c}(t_m)}{P(t_m \vert z_8)-P(t_m \vert z_c)}.
\end{align*}
\end{proposition}

The proofs of \Cref{pro:DCM:vector:mto,pro:DCM:mto} are straightforward; we omit the details.

\section{The MVPF of Extending Head Start}\label{appx:mvpf}

Recall our ternary instrument setting:
\begin{itemize}
\item 	$Z=0$ means no offer of admission to Head Start or to another preschool;
\item 	$Z=1$ means an offer of admission to Head Start	only;
\item 	$Z=2$ means an offer of admission to another preschool	only.
\end{itemize}
$Z=0$ does not preclude other ways to get into $h$ or $c$, $Z=1$ does not preclude other ways  to get into $c$, and $Z=2$ does not preclude  other ways  to get into $h$.

We denote $p(z)$ the probability that $Z=z$.  We are considering an increase in $p(1)$: more offers of admission to Head Start. As $p(1)$ increases, we also increase $p(2)$ to maintain the number of slots in alternative preschools constant. Like \citet{kline2016}, we assume that this increase in $p(2)$ only brings into alternative preschools children that would otherwise not attend preschools.

The MVPF is the ratio of the  benefits $dB$ of increasing $p(1)$ by $dp(1)$ to its budgetary costs $dC$. 
We have $B=(1-\tau)p \mathbb{E} Y$, where $p$ is the pre-tax return to expected scores, and $\tau$ the tax rate. Hence
\[
dB = (1-\tau) p d\mathbb{E} Y.
\]
The budget costs are the subsidies  ($\phi_j$ per student) to Head Start and other preschools, minus the tax receipts:
\[
C= \phi_h \Pr(D=h)+\phi_c \Pr(D=c)-\tau p \mathbb{E} Y.
\]
Therefore
\[
\text{MVPF} = \frac{(1-\tau)p d\mathbb{E} Y/dp(1)}{\phi_h d\Pr(D=h)/dp(1)-\tau p d\mathbb{E} Y/dp(1)}.
\]
In order to compute the MVPF, we start by evaluating the marginal return in expected outcomes $d\mathbb{E} Y/dp(1)$.

\subsection{The Expected Change in Outcomes}
Since 
\[
\Pr(D=c)=\Pr(D(0)=c) + \sum_{z=1,2} p(z) (\Pr(D(z)=c)-\Pr(D(0)=c))
\]
to keep it constant we must have 
\[
\frac{dp(2)}{dp(1)}= \frac{\Pr(D(0)=c)-\Pr(D(1)=c)}{\Pr(D(2)=c)-\Pr(D(0)=c)}.
\]
 $D(0)=c$ implies $D(2)=c$ since $Z=2$ targets $c$. Therefore $\Pr(D(2)=c)-\Pr(D(0)=c)=\Pr(D(2)=c, D(0)\neq c)$. Since $Z=1$ targets $h$, $D(1)=c$ implies $D(0)=c$; and $D(0)=c$ implies that $D(1)$ can only be $c$ or $h$. This gives us
\[
\Pr(D(0)=c)-\Pr(D(1)=c)=\Pr(D(0)=c, D(1)\neq c)=
\Pr(D(0)=c, D(1)=h),
\]
which is the proportion of the group $C_{ch}$ in the $2\times 3$ model.
Therefore
\[
\frac{dp(2)}{dp(1)}= \frac{\Pr(i \in C_{ch})}{\Pr(D(2)=c, D(0)\neq c)}.
\]
The resulting change in expected scores is
\[
d\mathbb{E} Y = dp(1) \mathbb{E}(Y(D(1))-Y(D(0))) + dp(2) \mathbb{E}(Y(D(2))-Y(D(0))).
\]
Now an offer of Head Start ($Z=1$) can only move children  to Head Start: $D(1)\neq D(0)$ implies that $D(1)=h$. As a consequence, 
\[
\mathbb{E}(Y(D(1))-Y(D(0)))=\mathbb{E}(Y(h)-Y(D(0)) \vert D(1)=h, D(0)\neq h) \times \Pr(D(1)=h, D(0)\neq h)
\]
and by the same argument,
\[
\mathbb{E}(Y(D(2))-Y(D(0)))=\mathbb{E}(Y(c)-Y(D(0)) \vert D(2)=c, D(0)\neq c) \times \Pr(D(2)=c, D(0)\neq c).
\]
Putting things together gives
\begin{align*}
\frac{d\mathbb{E} Y}{dp(1)} &= \Pr(D(1)=h, D(0)\neq h) \times\\
&\left(
\mathbb{E}(Y(h)-Y(D(0)) \vert D(1)=h, D(0)\neq h) +
\mathbb{E}(Y(c)-Y(D(0)) \vert D(2)=c, D(0)\neq c) 
\times S_c
\right) \\ 
&=  \Pr(D(1)=h, D(0)\neq h) \times \left(\text{LATE}_h +S_c \text{LATE}_c\right),
\end{align*}
where
\[
S_c = \frac{\Pr(i\in C_{ch})}{\Pr(D(1)=h, D(0)\neq h)}
\]
is, as  in the text of the paper, the proportion of the $h$-compliers that come from $c$.

\subsection{The MVPF}
We still need to compute the denominator $d\Pr(D=h)/dp(1)$. It is 
\[
\left(\Pr(D(1)=h)-\Pr(D(0)=h)\right)- 
\frac{dp(2)}{dp(1)} (\Pr(D(0)=h)-\Pr(D(2)=h)).
\]
The first term in the difference is $\Pr(D(1)=h, D(0)\neq h)$, the proportion of $h$-compliers. The second term equals
\[
\frac{\Pr(i \in C_{ch})}{\Pr(D(2)=c, D(0)\neq c)}(\Pr(D(0)=h)-\Pr(D(2)=h)).
\]
Since $Z=2$ targets $c$, the difference $\Pr(D(0)=h)-\Pr(D(2)=h)$ represents the proportion of children who  would get to Head Start  under $Z=0$ and leave it when offered admission to another preschool ($Z=2$) as $p(1)$ increases. Since these children can only have $D(2)=c$, our assumption rules out this group and the second term of the difference is zero.

As in \citet{kline2016}, $\text{LATE}_c=\text{LATE}_{nc}$; we end up with
\[
\text{MVPF} = \frac{(1-\tau)p \left(\text{LATE}_h +S_c \text{LATE}_{nc}\right)}{\phi_h-\tau p \left(\text{LATE}_h +S_c \text{LATE}_{nc}\right)},
\]
which happens to coincide with the formula used by \citet{kline2016}.

\section{A Roy Model with Positive Selection}\label{appx:normal:positive_selection}

In this section, we provide details for the example in Section~\ref{sec:pos:sel:3by3:ARUM}.
Recall that we assumed that
\[
\mathbb{E}(Y_i(2)\mid u_{i0}, u_{i1},u_{i2})-\mathbb{E}(Y_i(2))=a_0 u_{i0}+a_1 u_{i1}+a_2 u_{i2}\equiv v_i,
\]
where $a_0$, $a_1$, and $a_2$ are some constants. 
Since $(u_{i0},u_{i1},u_{i2})$ are jointly normal, taking expectations conditional on $u_{i2}-u_{i1}=\zeta_i$ and 
$u_{i2}-u_{i0}=\xi_i$ gives
\[ 
\mathbb{E}(Y_i(2)\mid \zeta_i, \xi_i)-\mathbb{E}(Y_i(2))
= 
\begin{pmatrix}
\text{Cov}(v_i, \zeta_i) & \text{Cov}(v_i, \xi_i)
\end{pmatrix}
\left(
\mathbb{V}\mathrm{ar}
\begin{pmatrix}
\zeta_i \\ \xi_i 
\end{pmatrix}
\right)^{-1} 
\begin{pmatrix}
\zeta_i \\ \xi_i 
\end{pmatrix}.
\]
Simple calculations give 
\[
 \mathbb{E}(Y_i(2)\mid \zeta_i,\xi_i)-\mathbb{E}(Y_i(2))  
=   \frac{a_2 + a_0 - 2 a_1}{3} \zeta_i + \frac{a_2 + a_1 -2 a_0}{3} \xi_i.
\]
The coefficients of $\zeta_i$ and $\xi_i$ are non-negative if and only if  
\[
a_2 + a_0 \geq 2 a_1 \text{ and } a_2 + a_1 \geq 2 a_0,
\]
which can be written more compactly as $a_2\geq \max(a_0,a_1)+\vert a_1-a_0\vert$. 
Assume that this inequality holds, and consider Figure~\ref{fig:ternaryroy}. Every point in $C_{112}^{(i)}$ has lower values of both $\zeta_i$ and $\xi_i$ than any point in $C_{212}$. Therefore, the expected value of $Y_i(2)$ in this triangle is smaller than $\mathbb{E}(Y_i(2)\mid i\in C_{212})$. 
Every point in $C_{112}^{(ii)}$ has a smaller value of $\zeta_i$ than any point in $C_{212}$ (fixing the value of $\xi_i \geq -\underline{U}(2)$ on both sides). The expected value of $Y_i(2)$ in this rectangle is again smaller than $\mathbb{E}(Y_i(2)\mid i\in  C_{212})$.
Combining these two inequalities gives $\mathbb{E}(Y_i(2)\mid i\in C_{212})\geq \mathbb{E}(Y_i(2)\mid i\in C_{112})$, that is \eqref{assn:positivesel:3by3:2}.

\end{document}